\documentclass[aps,pra,10pt,floatfix,amsmath,onecolumn,notitlepage,nofootinbib,longbibliography,superscriptaddress]{revtex4-2}

\usepackage[ruled]{algorithm2e}
\SetKwFor{For}{for (}{) $\lbrace$}{$\rbrace$}
\usepackage{amssymb}
\usepackage{graphicx}
\usepackage{graphics}
\usepackage{amsmath}
\usepackage{amsthm}
\usepackage{color}
\usepackage{dsfont}
\usepackage{wrapfig}
\usepackage[normalem]{ulem}
\usepackage{longtable}

\usepackage{tabularray}
\usepackage{float}
\usepackage{mathtools}
\usepackage[colorlinks=true, linkcolor=blue, citecolor=red]{hyperref}
\usepackage{verbatim}
\usepackage{xcolor}
\usepackage[nodayofweek]{datetime}
\usepackage{wrapfig}

\usepackage{multirow}

\newcommand{\p}{\mathrm{p}}
\newcommand{\pr}{\mathrm{PR}}
\newcommand{\chsh}{\mathrm{CHSH}}

\usepackage{color}

\def\>{\rangle}
\def\<{\langle}

\def\id{\mathsf{id}}

\renewcommand{\geq}{\geqslant}
\renewcommand{\leq}{\leqslant}

\newcommand{\A}{\mathbf{A}}
\newcommand{\B}{\mathbf{B}}
\newcommand{\C}{\mathbf{C}}

\newtheorem{theorem}{Theorem}
\newtheorem*{theorem*}{Theorem}

\newtheorem{lemma}{Lemma}
\newtheorem*{lemma*}{Lemma}

\theoremstyle{definition}

\newtheorem{definition}{Definition}
\newtheorem*{definition*}{Definition}

\theoremstyle{remark}

\newcommand\nseq{\stackrel{\mathclap{\normalfont\mbox{\small ns}}}{=}}

\newcommand{\xtens}{\mathbin{\mathop{\otimes}\limits_{\max}}}
\newcommand{\ntens}{\mathbin{\mathop{\otimes}\limits_{\min}}}

\newcommand{\tr}{{\rm Tr}}

\newcommand{\ket}[1]{|#1\rangle}

\newcommand{\npr}{g}

\newcommand{\1}{\mathds{1}}

\def\chsh{{\rm CHSH}}

\definecolor{cool_green}{rgb}{0.0, 0.5, 0.0}

\begin{document}

\title{Correlation Self-Testing of Quantum Theory against Generalised Probabilistic Theories with Restricted Relabelling Symmetry}
\author{Kuntal Sengupta}
\email{kuntal.sengupta@neel.cnrs.fr}
\affiliation{Univ.\ Grenoble Alpes, CNRS, Grenoble INP, Institut N\'eel, 38000 Grenoble, France} 
\affiliation{Department of Mathematics, University of York, Heslington, York, YO10 5DD, United Kingdom}
\author{Mirjam Weilenmann}
\email{mirjam.weilenmann@inria.fr}
\affiliation{Inria, Télécom Paris - LTCI, Institut Polytechnique de Paris, 91120 Palaiseau, France}
\affiliation{Department of Applied Physics, University of Geneva, Switzerland}
\author{Roger Colbeck}
\email{roger.colbeck@kcl.ac.uk}
\affiliation{Department of Mathematics, King's College London, Strand, London, WC2R 2LS, United Kingdom}
\affiliation{Department of Mathematics, University of York, Heslington, York, YO10 5DD, United Kingdom}

\date{$4^{\text{th}}$ November 2025}

\begin{abstract} 
  Correlation self-testing of quantum theory involves identifying a task or set of tasks whose optimal performance can be achieved only by theories that can realise the same set of correlations as quantum theory in every causal structure. Following this approach, previous work has ruled out various classes of generalised probabilistic theories whose joint state spaces have a certain regularity in the sense of a (discrete) rotation symmetry of the bipartite state spaces. Here we consider theories whose bipartite state spaces lack this regularity. We form them by taking the convex hull of all the local states and a finite number of non-local states. We show that a criterion of compositional consistency is needed in such theories: for a measurement effect to be valid, there must exist at least one measurement that it is part of. This goes beyond previous consistency criteria and corresponds to a strengthening of the no-restriction hypothesis. We show that quantum theory outperforms these theories in a task called the adaptive CHSH game, which shows that they can be ruled out experimentally. We further show a connection between compositional consistency and Tsirelson's bound.
\end{abstract}
\maketitle

\section{Introduction}

Within quantum theory, separated parties can realise correlations that are impossible to create classically. This is known as nonlocality and, as well as being a striking foundational feature, it also has applications, e.g., in cryptography~\cite{Pirandola20}. These nonlocal correlations remain non-signalling, i.e., they do not allow the separated parties to communicate, and form a subset of all non-signalling correlations~\cite{Cirelson93,pr}. Understanding why only a subset of the non-signalling correlations can be realised in quantum theory is an important open question in quantum foundations.

One way to approach this question is to start from non-signalling correlations and identify information-theoretic principles that restrict this set to the set of quantum correlations. A few proposed principles include non-triviality of communication complexity~\cite{PhysRevLett.96.250401}, impossibility of nonlocal computation~\cite{PhysRevLett.99.180502}, information causality~\cite{Pawlowski2009}, macroscopic locality~\cite{Navascues2009} and local orthogonality~\cite{Fritz2013}. Although these approaches provide insight into the properties of quantum correlations and reduce the set of allowed non-signalling correlations, none is known to single out the set of quantum correlations~\cite{Navascues2015}. Whether a given principle is natural or not is somewhat subjective, which leads us to instead consider the possibility of a task in which quantum theory performs optimally. \textit{Correlation self-testing} (which we abbreviate to \textit{self-testing} henceforth) of quantum theory~\cite{PhysRevLett.125.060406,PhysRevA.102.022203} follows this approach and asks whether there is an information-theoretic task that can only be optimally performed using quantum correlations. If such a task were found then the underlying information-theoretic requirement for optimally winning the task might point to a physical principle.

In~\cite{PhysRevLett.125.060406,PhysRevA.102.022203} the Adaptive CHSH (ACHSH) game was proposed as a candidate task for correlation self-testing of quantum theory. There, quantum theory was self-tested against a significant collection of alternative theories defined in the framework of Generalised Probabilistic Theories (GPTs). For instance, it was self-tested against theories constructed by min- and max-tensor products of any finite dimensional single system state spaces and locally tomographic~\cite{Janotta_2011} theories with specific single system state spaces, independent of the composition. It is known from two examples~\cite{realQM,PhysRevA.110.022225} that the ACHSH game itself is not sufficient for self-testing quantum theory in general. Real quantum mechanics can win the ACHSH game with the same score as quantum theory but can be ruled out with another game in the same causal structure~\cite{realQM}. In~\cite{PhysRevA.110.022225}, a theory was found that outperforms quantum theory in the ACHSH game. However, so far this theory does not recover all correlations in the bipartite Bell scenario, see e.g.,~\cite[Section 5.6]{wreo36425} for a treatment of the chained Bell inequalities. Whether an extension of the theory of~\cite{PhysRevA.110.022225} may resolve this issue in the future is unknown.  
For the theories considered in this work, the ACHSH game will turn out to be sufficient, hence we restrict our considerations to this game.

The quantum advantage in the ACHSH game comes from agents being able to perform entanglement swapping in which two parties, each of whom separately share entanglement with a third can end up sharing an entangled state by post-selecting on the outcome of a joint measurement performed by the third party. Theories whose compositions are formed with both min- and max-tensor products of single system state spaces do not allow this: in the min-tensor product, all bipartite states are local, while in the max-tensor product, the set of joint measurements is not rich enough (cf.\ Definition~\ref{def::MinAndMax}). Therefore, both these type of compositions perform no better than local theories in the ACHSH game~\cite{PhysRevLett.125.060406,PhysRevA.102.022203}. In contrast, examples of state spaces with bipartite composition rules beyond the min- and max-tensor product are known to allow entanglement swapping, while being able to realise all non-signalling correlations~\cite{PhysRevLett.102.110402,Skrzypczyk_2009}. It is hence natural to ask whether such theories could achieve higher scores than quantum theory in the ACHSH game. Our main results answer this question in the negative.

In~\cite{PhysRevLett.125.060406,PhysRevA.102.022203}, state spaces were considered that are closed under all relabelling operations (i.e., relabellings of inputs, outputs and parties). Using the set of correlations alone it is not always possible to deduce whether the underlying state space has all of these symmetries (see Section~\ref{SubSec::Correlations} for an example).   
It is hence reasonable to consider theories without these symmetries.\footnote{Note that the theory that outperforms quantum theory in the ACHSH game~\cite{PhysRevA.110.022225} also dropped some of these symmetries.} 
In this paper, we consider theories whose state spaces are symmetric at the level of single systems, but not at the level of bipartite systems (in contrast to, for example, the min- and max-tensor product compositions). Our examples include state spaces obtained by removing particular extremal states from the boxworld state space, and modifications of these where each extremal state is mixed with noise. We consider cases with bipartite systems in which local tomography requires either two or three binary outcome measurements. Since the smallest quantum state requires three measurements for state tomography, the case with three measurements provides a closer comparison to quantum theory. Although our state space models are asymmetric in general, we retain symmetry under party swap. We present a series of results to show that quantum theory can be self-tested against all these asymmetric theories using the ACHSH game. 

A consideration that is significant in this analysis is that ensuring compositional consistency in such theories requires more stringent restrictions on the measurement effects than the usual no-restriction hypothesis~\cite{PhysRevA.81.062348} and other criteria considered in~\cite{PhysRevLett.102.110402,PhysRevA.73.012101,PhysRevLett.119.020401,dallarno2023signaling}. The criterion we propose requires both that when an effect is applied to subsystems of a larger system the remaining subsystems reside in a valid state~\cite{PhysRevLett.102.110402,PhysRevA.73.012101} and that this effect can be part of a measurement in which every effect is consistent in the manner just stated. We call this property (when applied to arbitrarily large systems) complete state space preservability, by analogy with complete positivity in quantum theory. Because of the difficulty of dealing with arbitrarily large systems, we present a necessary condition, \textit{minimal $k$-preservability}, that an effect must satisfy in order to be completely state space preserving (where $k$ relates to the number of subsystems in the larger system (see later)). For the family of theories we looked at, \textit{minimal $k$-preservability} imposes stronger constraints than previous criteria~\cite{PhysRevLett.102.110402,PhysRevA.73.012101,PhysRevLett.119.020401,dallarno2023signaling} in each case. In addition, we show that for the state spaces considered here, minimal 2-preservability can be used to recover Tsirelson's bound.

The structure of our paper is as follows: in Section~\ref{Section::Preliminaries}, we review the mathematical framework of GPTs and previous results. In Sections~\ref{Section::EffectPolytope} and~\ref{Section::EffectPolytope3Fid}, we provide an analytic construction of the effect polytopes for any given noisy asymmetric state space in our model. In addition, we provide a complete list of the set of extreme effects for each such case and a formula for calculating the number of extreme effects of such effect spaces. To do this, we introduce an algorithm that finds the vertices of a polytope by taking a larger polytope whose vertices are known and cutting it to form the polytope of interest. In Section~\ref{Sec::MinimalPreservability}, we introduce our preservability criterion (minimal $k$-preservability), before checking the effects obtained in Sections~\ref{Section::EffectPolytope} and~\ref{Section::EffectPolytope3Fid} using it in Section~\ref{Section::Min2PresCouplers}. We find that entanglement swapping is impossible in most of the state spaces considered. Section~\ref{Section::CorrelationSelfTesting} contains the main results, in particular that quantum theory can be correlation self-tested using the ACHSH game against every state space considered in this work. Finally, in Section~\ref{Section::TsirelsonBound} we draw a connection between minimal $k$-preservability and Tsirelson's bound for the state spaces in question.  

\section{Preliminaries}
\label{Section::Preliminaries}
\subsection{Correlation Self-Testing and the Adaptive CHSH Game}
\label{subsection::achsh_game}
The setup of correlation self-testing is as follows: given a physical theory, $\mathcal{P}$, and a theory, $\mathcal{T}$, if $\mathcal{P}$ can produce correlations within a causal structure\footnote{A causal structure is a collection of variables arranged as nodes of a directed acyclic graph where some of the nodes are labelled as observed. It represents the causal relations among the variables.} that cannot be produced by $\mathcal{T}$ in the same causal structure, then there is an information processing task in which $\mathcal{P}$ outperforms $\mathcal{T}$.  More generally, for a set of theories $\{\mathcal{T}_i\}_{i=1}^{n}$, suppose there is a set of tasks (or just one) that singles out $\mathcal{P}$ from the set $\{\mathcal{T}_i\}_{i=1}^{n}$. Such a set of tasks is said to be a \textit{correlation self-test} of $\mathcal{P}$ against $\{\mathcal{T}_i\}_{i=1}^{n}$. The overall goal is to find a set of tasks that single out quantum theory within all GPTs.

\begin{figure}[h]
    \centering
    \includegraphics[width=0.35\textwidth]{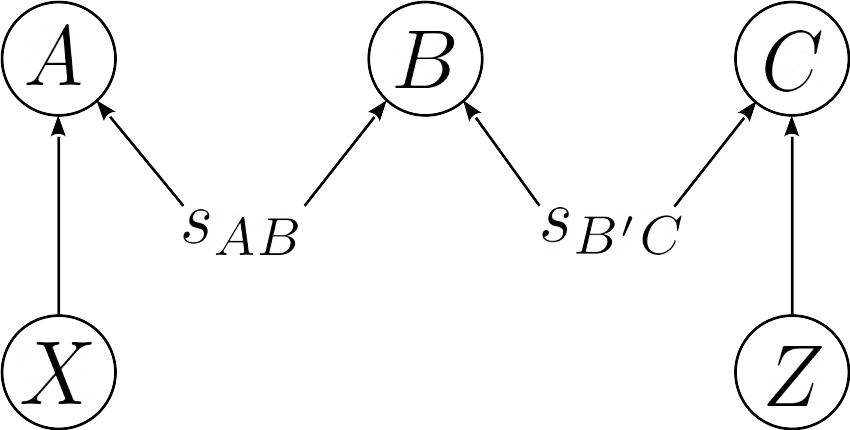}
    \caption{Causal structure for the Adaptive CHSH game. Bob shares the resource $s_{AB}$ with Alice and the resource $s_{B'C}$ with Charlie. A referee asks questions to Alice and Charlie labelled by random variables $X$ and $Z$ respectively. Bob performs a joint measurement on his share of resources, the outcomes of which are labelled by the random variable $B$. Alice and Charlie perform local measurements on their subsystems, the outcomes of which are labeled by random variables $A$ and $C$. The value of all the random variables determine the score in the game. There are no non-classical tripartite resources shared by all the three parties (shared tripartite randomness is allowed).}
    \label{ACHSH}
\end{figure}

The Adaptive CHSH game has been shown to rule out a variety of theories in this way~\cite{PhysRevLett.125.060406,PhysRevA.102.022203}. It uses the CHSH game~\cite{PhysRevLett.23.880}, which is played between two cooperating parties, Alice and Bob. A referee asks them random questions labelled by the random variables $X$ and $Y$ which can take values $x,y \in \{0,1\}$. They return answers labelled by the random variables $A$ and $B$ taking values $a,b \in \{0,1\}$. The parties cannot communicate during the game and they win if $a \oplus b = xy$ (where $\oplus$ denotes the {\sc xor}). Using quantum theory, the maximum winning probability is $\left(1+1/\sqrt{2}\right)/2 \approx 0.85$, also known as Tsirelson's bound~\cite{Cirelson1980}. There are 8 variations of this game (equivalent to one another by relabellings of inputs or outputs) whose winning conditions are $a \oplus b = xy \oplus \gamma_{1}x \oplus\gamma_2 y \oplus \gamma_0$, where $\gamma_0,\gamma_1,\gamma_2 \in \{0,1\}$. The standard game described above corresponds to the choice $\gamma_0=\gamma_1=\gamma_2=0$. There exists a maximally nonlocal theory that can perfectly win this game and therefore outperform quantum theory~\cite{Cirelson93,pr}. Hence, the CHSH game cannot correlation self-test quantum theory.

The Adaptive CHSH game is as follows: in the causal structure displayed in Fig.~\ref{ACHSH}, three players Alice, Bob and Charlie play a cooperative game, in which Bob supplies two bits denoted by $B$ taking values $(b_0,b_1) \in \{(0,0),(0,1),(1,0),(1,1)\}$ to a referee, and Alice and Charlie are each asked uniformly distributed binary questions denoted by $X$ and $Z$ and provide binary answers $A$ and $C$. The parties win the game if $a \oplus c= xz \oplus (b_0 \oplus b_1)x \oplus z \oplus b_0$ is satisfied. Note that the possible values of $B$ correspond to four variations of the CHSH game mentioned in the previous paragraph. In quantum theory, this game can be won with a maximum winning probability of $\left(1+1/\sqrt{2}\right)/2 $~\cite{PhysRevLett.125.060406,PhysRevA.102.022203}. In short, an optimal quantum strategy involves Bob sharing two copies of a maximally entangled qubit state, one with Alice and one with Charlie and then performing a Bell basis measurement on his pair of qubits. Alice and Charlie perform local measurements that generate probability distributions that win one of the CHSH games with a high probability.  For completeness we present an optimal strategy in Appendix~\ref{sub::QStrat} (see also~\cite{PhysRevLett.125.060406,PhysRevA.102.022203} for other results related to this game).

\subsection{Generalised Probabilistic Theories}
A typical lab experiment involves \textit{state preparation}, describing the initial states of the systems involved, \textit{transformations} of states, describing how the systems evolve, and \textit{measurements} on them. \textit{Generalised Probabilistic Theories} (GPTs) provide a general mathematical framework in which such processes can be studied. 

\begin{definition}\label{Def::GPT}\textbf{(State Space, Effect Space, GPT)} Let $\mathbb{V}$ be a finite dimensional real vector space  and $\mathbb{V}^*$ its dual vector space. 
\begin{enumerate}
    \item A \emph{state space} $\mathcal{S}$  is a  compact and convex subset of $\mathbb{V}$ such that there exists an element $u \in \mathbb{V}^*$ (called the \emph{unit} effect) with the property that  $\left<u,s\right>$=1 for any $s \in \mathcal{S}$. The sub-normalised state space $\mathcal{S}_{\leq}$ associated with $\mathcal{S}$ is defined as the convex hull of $\mathcal{S}$ and the zero vector. The \emph{state space cone} is the set of positive multiples of every state.
    \item \label{item::effect_space} The \emph{maximal effect space} $\mathcal{E_S}$ of a state space $\mathcal{S}$ is a compact and convex subset of $\mathbb{V}^*$ defined as $$
    \mathcal{E}_{\mathcal S} \coloneqq \{e \in \mathbb{V}^*\ |\ \left<e,s\right> \in [0,1]\ \forall \ s \in \mathcal{S}\} 
    .$$
    An \emph{effect space} $\mathcal{E}$ is any subset of $\mathcal{E_S}$  such that $u,v_0 \in \mathcal{E}$ where $v_0$ is the zero effect with the property $\left<v_0,s\right>=0 \ \forall \ s \in \mathcal{S}$.
    The \emph{effect space cone} is the set of positive multiples of all effects.
    
    \item Let $\mathbf{S}$ and $\mathbf{E}$ be a collection of state and effect spaces (one for each type of elementary system). A \emph{GPT} is a triple $\left(\mathbf{S},\mathbf{E}, \boxtimes \right)$, where $\boxtimes$ is a set of composition rules that specify how to form composite state spaces from smaller ones. This composition rule has the following properties when acting on elementary systems:
    \begin{enumerate}
        \item for two state spaces $\mathcal{S}_\A, \mathcal{S}_\B\in\mathbf{S}$, with underlying vector spaces $\mathbb{V}_\A$ and $\mathbb{V}_\B$ describing systems labelled by $\A$ and $\B$, and $\boxtimes_{\A\B} \in \boxtimes$, the set $\mathcal{S}_{\A\B}\coloneqq \mathcal{S}_\A \boxtimes_{\A\B} \mathcal{S}_\B$ is a state space with underlying vector space $\mathbb{V}_\A\otimes\mathbb{V}_\B$ and for all $s_{\A\B} \in \mathcal{S}_{\A\B}$, $e_\A\in \mathcal{E}_\A$ and $e_\B\in \mathcal{E}_\B$
    $$
    (\id_\A \otimes e_\B) \left(s_{\A\B}\right)  \in \mathcal{S}_{\A_\leq} \quad \text{and } \quad (e_\A \otimes \id_\B) \left(s_{\A\B}\right) \in \mathcal{S}_{\B_\leq}
    $$
    where $\id_{\A/\B} : \mathbb{V}_{\A/\B} \to \mathbb{V}_{\A/\B}$ is the identity map and $\otimes$ is the tensor product,
    \item for any collection of states $\{(r_{i})_{\A}\}_{i=1}^n \subseteq \mathcal{S}_\A,\{(s_{i})_{\B}\}_{i=1}^n \subseteq \mathcal{S}_\B$, and set $\{\lambda_i\}$ satisfying $\lambda_i\geq0$, $\sum_i\lambda_i=1$, 
    $$\sum_{i=1}^n \lambda_i (r_i)_{\A} \otimes (s_i)_{\B} \in \mathcal{S}_\A \boxtimes_{\A\B} \mathcal{S}_\B,$$ 
    for any $\mathcal{S}_\A, \mathcal{S}_\B \in \mathbf{S}$ and any composition rule $\boxtimes_{\A\B} \in \boxtimes$ between $\mathcal{S}_{\A}$ and $\mathcal{S}_{\B}$,
    \item the effect space $\mathcal{E}_{\A\B}$ of $\mathcal{S}_{\A\B}$, is a subset of the maximal effect space $\mathcal{E}_{\mathcal{S}_{\A\B}}$ defined in~\eqref{item::effect_space}.
    \end{enumerate}
These properties extend naturally to allow multiple combinations of elementary systems, where the composition rule should be associative, i.e., $(\mathcal{S}_\A\boxtimes_{\A\B}\mathcal{S}_\B)\boxtimes_{(\A\B)(\C)}\mathcal{S}_\C=\mathcal{S}_\A\boxtimes_{\A(\B\C)}(\mathcal{S}_\B\boxtimes_{\B\C}\mathcal{S}_\C)$ etc.
\end{enumerate} 
\end{definition}
The definition of the state space implies that the description of any state in the theory can be completely encoded in the entries of a finite vector and that state tomography can be performed using a finite number of measurements. A minimal set of measurements with which state tomography can be performed are called a set of \emph{fiducial} measurements. In addition, since the bipartite state space is a subset of $\mathbb V_\A \otimes \mathbb V_\B$, any bipartite state can be identified by local tomography. This means that the state spaces for any GPT can be described using probability tables corresponding to the probabilities of outcomes when local fiducial measurements are made on the state. For instance, in the case with two fiducial measurements labelled $x\in\{0,1\}$ each having two outcomes labelled $a\in\{0,1\}$, any local state can be written in the form
\begin{equation}
    \left(\begin{array}{c}\p(0|0)\\\p(1|0)\\\hline\p(0|1)\\\p(1|1)\end{array}\right),
\end{equation}
where each probability is $\p(a|x)$, and any bipartite state can be written
\begin{equation}
    \left(\begin{array}{cc|cc}
\p(0,0|0,0) & \p(0,1|0,0) & \p(0,0|0,1) & \p(0,1|0,1)\\[6pt]
\p(1,0|0,0) &\p(1,1|0,0) & \p(1,0|0,1) & \p(1,1|0,1)\\[6pt]
\hline
\vphantom{\Big(}\p(0,0|1,0) &\p(0,1|1,0) & \p(0,0|1,1) & \p(0,1|1,1) \\[6pt]
\p(1,0|1,0) & \p(1,1|1,0) & \p(1,0|1,1) & \p(1,1|1,1)\\
\end{array}\right),
\label{Eq::TsirelsonNotation}
\end{equation}
where each probability is $\p(a,b|x,y)$ ($y\in\{0,1\}$ being the measurement on the second system and $b\in\{0,1\}$ its outcome). Although displayed as a matrix for convenience, we will usually consider this state as a vectors in $\mathbb{R}^{16}$. This representation naturally generalises to more parties, inputs and outputs.

For any pair of effects $e_\A,e_\B \in \mathcal{E}$ and state $s_{\A\B} \in \mathcal{S}_{\A\B}$, since both $(\id_\A \otimes e_\B) \left(s_{\A\B}\right)$ and $(e_\A \otimes \id_\B )\left(s_{\A\B}\right)$ are valid sub-normalised states,
\begin{equation}
    (e_\A \otimes e_\B) \left(s_{\A\B}\right) = (\id_\A \otimes e_\B)(e_\A \otimes \id_\B) \left(s_{\A\B}\right) = (e_\A \otimes \id_\B) (\id_\A \otimes e_\B)\left(s_{\A\B}\right) \in [0,1],
\end{equation}
and hence all product effects $e_\A \otimes e_\B$ are elements of the bipartite effect space. Additionally, local actions of effects on respective subsystems always commute. Thus, the \textit{no-signalling} conditions naturally emerge, i.e.,
\begin{equation}
\begin{split}
\sum_b \p(a,b|x,y) &= \sum_b\p(a,b|x,y') =\p(a|x) \quad \text{for all }a,x,y,y'\ \ \ \ \text{and}\\
    \sum_a \p(a,b|x,y) &= \sum_a\p(a,b|x',y) =\p(b|y) \quad \text{for all }b,x,x',y\, .
\end{split}  
\label{Eq::NoSignalling}
\end{equation}

Although in this work we use constructions that start with the state space, an alternative approach is to first define an effect space $\mathcal{E}$ and then choose a set $\mathcal{S} \subseteq \mathcal{S_E}$ appropriately as the state space, where $\mathcal{S_E}:=\{s\in\mathbb{V}\ |\ \langle e,s\rangle\in[0,1]\ \forall\ e\in\mathcal{E}, \langle u,s\rangle = 1 \}$. 

Given two state spaces $\mathcal{S}_\A$ and $\mathcal{S}_\B$, a composition rule $\boxtimes_{\A\B}$
specifies a composite state space. We present here two examples of composition rules that
allow one to construct the joint state space, regardless of the types of the systems being composed. These are the minimal and maximal tensor product compositions.

\begin{definition}\label{def::MinAndMax}\textbf{(Minimal and Maximal Tensor Products)} Let $\mathcal{S}_\A \subset \mathbb{V}_{\A}$ and $\mathcal{S}_\B  \subset \mathbb{V}_\B$ be two  state spaces and $\mathcal{E}_{\mathcal{S}_\A}$ and $\mathcal{E}_{\mathcal{S}_\B}$ be their corresponding maximal effect spaces. Then
\begin{itemize}
    \item the \emph{minimal (min-) tensor product} of $\mathcal{S}_\A$ and $\mathcal{S}_\B$ is defined as
    $$
    \mathcal{S}_\A \ntens \mathcal{S}_\B \coloneqq \mathrm{ConvHull}\{s_\A \otimes s_\B\ |\  s_\A \in \mathcal{S}_\A, s_\B \in \mathcal{S}_\B\},
    $$
       \item the \emph{maximal (max-) tensor product} of $\mathcal{S}_\A$ and $\mathcal{S}_\B$ is defined as
    $$
    \mathcal{S}_\A \xtens \mathcal{S}_\B \coloneqq \{s_{\A\B} \in \mathbb{V}_\A \otimes \mathbb{V}_\B \ |\ \left<e_\A \otimes e_\B, s_{\A\B}\right> \in [0,1]\ \forall \  e_\A \in \mathcal{E}_{\mathcal{S}_\A}, e_\B \in \mathcal{E}_{\mathcal{S}_\B}, \left<u_\A \otimes u_\B, s_{\A\B}\right> = 1\ \}.
    $$
\end{itemize}
\end{definition}

 The maximal tensor product state space, as defined, is the largest set of bipartite states for which marginalisation to single system gives a valid state in the single system state space. Therefore, for any arbitrary composition $\boxtimes_{\A\B}$ of state spaces, we have $\mathcal{S}_{\A} \ntens \mathcal{S}_{\B} \subseteq \mathcal{S}_{\A} \boxtimes_{\A\B} \mathcal{S}_{\B} \subseteq \mathcal{S}_{\A} \xtens \mathcal{S}_{\B}$.

In the following, we provide examples of GPTs to illustrate the GPT framework. This treatment follows that of Barrett~\cite{PhysRevA.75.032304}, but we include a detailed description of specific effect spaces that we will use later. We give an example of how to recast qubit quantum theory as a GPT, with states given by probability distributions in Appendix~\ref{Appendix::QubitQTasGPT}.

\subsubsection{Generalised Local Theories}
\label{Subsection::ExampleGLT}
A \textit{generalised local theory} (GLT) refers to any GPT where the single system state space allows all probability distributions and in which every multipartite state is separable across all bi-partitions\footnote{Classical probability theory is a GLT (see Appendix~\ref{Appendix::ClassicalGPT} for an example). GPTs are non-classical if they require more than one fiducial measurement to characterise.}. We provide two examples of non-classical GLTs that are relevant for this paper. Consider the \textit{gbit} state space $\mathcal{G}_{\rm m}^{\rm n}$ of a single system for which state tomography requires $\rm m$ fiducial measurements having $\rm n$ outcomes each\footnote{In principle, one can have a theory where the number of outcomes depends on the choice of measurement but we avoid this for simplicity of notation.}, such that any valid probability distributions on these measurements and outcomes are possible. The min-tensor product of two such gbit state spaces $\mathcal{G}_{\rm m_\A}^{\rm n_\A}$ and $\mathcal{G}_{\rm m_\B}^{\rm n_\B}$ is always a generalised local theory. The two examples are cases of this composition when i) $\rm m_\A=m_\B=n_\A=n_\B=2$ and when ii) $\rm m_\A=m_\B=3$ and $\rm n_\A=\rm n_\B=2$.

When $\rm m=n=2$, the state space $\mathcal{G}_{\rm 2}^{\rm 2}$ can be characterised as the convex hull of four extreme deterministic states, in particular
$$
s_1=\begin{pmatrix}
    1\\
    0\\
    \hline
    1\\
    0\\
\end{pmatrix}, \quad s_2=\begin{pmatrix}
    1\\
    0\\
    \hline
    0\\
    1\\
\end{pmatrix}, \quad s_3=\begin{pmatrix}
    0\\
    1\\
    \hline
    1\\
    0\\
\end{pmatrix}, \quad s_4=\begin{pmatrix}
    0\\
    1\\
    \hline
    0\\
    1\\
\end{pmatrix}.
$$

The extreme states of the min-tensor product composition of two $\mathcal{G}_2^2$ state spaces are locally deterministic and can be calculated by taking the tensor product of $s_i$ with $s_j$ for $i,j\in \{1,2,3,4\}$, e.g., 
\begin{equation}\label{eq:LDS}
s_1 \otimes s_1 = \left(\begin{array}{cc|cc}
    1 & 0 & 1 & 0  \\
    0 & 0 & 0 & 0  \\
     \hline
    1 & 0 & 1 & 0  \\
    0 & 0 & 0 & 0  \\
\end{array}\right).
\end{equation}
The resultant joint state space is a polytope characterised by the convex hull of these 16 states. We denote this state space as $\mathbb{H}^{[0]}_{(2,2)}$, where the superscript $[\npr]$ denotes that there are $\npr$ entangled extremal states (in this case $\npr=0$) and $(2,2)$ signifies $(\rm m=2,n=2)$. Any polytope can be characterised either by the convex hull of its extreme states (vertex description) or by a set of inequalities defining its facets (facet description) (see e.g., Section~2.2.4 of~\cite{boyd2004convex}). Assuming that the normalisation and no-signalling conditions hold (cf.~\eqref{Eq::NoSignalling}), the state space $\mathbb{H}^{[0]}_{(2,2)}$ can be characterised by 24 facets, of which 16 are positivity facets (corresponding to $\p(a,b|x,y) \geq 0\ \forall \ a,b,x,y$), and the remaining 8 are called CH facets. To list the CH facets, consider the following 4 vectors in $\mathbb{R}^{16}$:
\begin{equation}
e_{\rm CH_1}=\left(
\begin{array}{cc|cc}
 0 & 0 & 1 & 0 \\
 0 & 1 & 0 & 0 \\ \hline
 0 & -1 & 0 & 1 \\
 0 & 0 & 0 & 0 \\
\end{array}
\right), e_{\rm CH_2}=\left(
\begin{array}{cc|cc}
 0 & 0 & 0 & 0 \\
 0 & -1 & 0 & 1 \\ \hline
 0 & 0 & 1 & 0 \\
 0 & 1 & 0 & 0 \\
\end{array}
\right), e_{\rm CH_3}= \left(
\begin{array}{cc|cc}
 0 & 0 & 0 & 1 \\
 0 & 1 & 0 & 0 \\ \hline
 0 & -1 & 1 & 0 \\
 0 & 0 & 0 & 0 \\
\end{array}
\right),
e_{\rm CH_4}=\left(
\begin{array}{cc|cc}
 0 & 0 & 0 & 0 \\
 0 & -1 & 1 & 0 \\ \hline
 0 & 0 & 0 & 1 \\
 0 & 1 & 0 & 0 \\
\end{array}
\right).
\label{Eq::CHFacets}
\end{equation}
The 8 CH facets are $\{\left<e_{\rm CH_i},\mathbf{x}\right> \leq 1\}_{i=1}^4$ and $\{\left<e_{\rm CH_i},\mathbf{x}\right> \geq 0\}_{i=1}^4$ where $\mathbf{x} \in \mathbb{R}^{16}$ and the inner product is defined as the sum of element-wise products. The second set of 4 inequalities are given by $\{\left<u-e_{\rm CH_i},\mathbf{x}\right> \leq 1\}_{i=1}^4$, where $u$ is the unit effect. Each CH facet inequality is saturated by 8 local deterministic states. To find the corresponding maximal effect polytope, first recall that for a vector $e\in \mathbb{V}^*$ to be an effect, it must satisfy $\left<e, s\right> \in [0,1]$ for any state $s$ in the state space (see Def.~\ref{Def::GPT}). Since we defined state spaces to be convex and compact, it is sufficient to check whether $\left<e, s_{i}\right> \in [0,1]$ for every extreme state $s_i$ of the state space. We denote the extreme states of $\mathbb{H}^{[0]}_{(2,2)}$ as $\mathrm{Vert}\left[\mathbb{H}^{[0]}_{(2,2)}\right]$. The set of facet-defining inequalities of the effect polytope $\mathcal{E}_{\mathbb{H}^{[0]}_{(2,2)}}$ is then given by: 
\begin{equation}
    \mathrm{Facets}\left[ \mathcal{E}_{ \left[\mathbb{H}^{[0]}_{(2,2)}\right]} \right] \coloneqq \Biggl\{ \mathbf{x} . s_{\mathrm{vertex}} \geq 0 \ |\   s_{\mathrm{vertex}} \in \mathrm{Vert}\left[\mathbb{H}^{[0]}_{(2,2)}\right] \Biggr\} \bigcup \Biggl\{ \mathbf{x} . s_{\mathrm{vertex}} \leq 1 \ |\   s_{\mathrm{vertex}} \in \mathrm{Vert}\left[\mathbb{H}^{[0]}_{(2,2)}\right] \Biggr\}.
\label{eq::FacetDescription}
\end{equation}
\noindent
Finding the vertices of a polytope from its facets is called \textit{vertex enumeration}. For this work, we have used PANDA~\cite{Lorwald2015} to solve all vertex enumeration problems. In the present case, we find that the effect polytope has 90 extreme effects (see Appendix~\ref{Appendix::EffectPoly} for a full classification) of which 82 are separable effects and the remaining 8 are entangled effects of the form $ e_{\mathrm {CH}_i}$  and $u-e_{\mathrm{CH}_i}$ with $i=1,\ldots,4$. The 82 separable effects are the positivity effects and sums of them.

When $\rm m_\A=m_\B=3$ and $\rm n_\A=\rm n_\B=2$, the extreme (deterministic) states of the state space $\mathcal{G}_{3}^2$ are

$$
\left(
\begin{array}{c}
 1 \\
 0 \\
 \hline
 1 \\
 0 \\
 \hline
 1 \\
 0 \\
\end{array}
\right), \left(
\begin{array}{c}
 1 \\
 0 \\
 \hline
 1 \\
 0 \\
 \hline
 0 \\
 1 \\
\end{array}
\right), \left(
\begin{array}{c}
 1 \\
 0 \\
 \hline
 0 \\
 1 \\
 \hline
 1 \\
 0 \\
\end{array}
\right),\left(
\begin{array}{c}
 1 \\
 0 \\
 \hline
 0 \\
 1 \\
 \hline
 0 \\
 1 \\
\end{array}
\right),
\left(
\begin{array}{c}
 0 \\
 1 \\
 \hline
 1 \\
 0 \\
 \hline
 1 \\
 0 \\
\end{array}
\right),
\left(
\begin{array}{c}
 0 \\
 1 \\
 \hline
 1 \\
 0 \\
 \hline
 0 \\
 1 \\
\end{array}
\right),
\left(
\begin{array}{c}
 0 \\
 1 \\
 \hline
 0 \\
 1 \\
 \hline
 1 \\
 0 \\
\end{array}
\right),
\left(
\begin{array}{c}
 0 \\
 1 \\
 \hline
 0 \\
 1 \\
 \hline
 0 \\
 1 \\
\end{array}
\right).
$$
There are 64 local deterministic states of the state space polytope $\mathbb{H}^{[0]}_{(3,2)}$ formed by taking the min-tensor product of two $\mathcal{G}_3^2$ state spaces. Alternatively, $\mathbb{H}^{[0]}_{(3,2)}$ can be characterised by 36 positivity facets and 648 Bell facets. These Bell facets can be categorised into two equivalence classes: the first containing 72  CH facets and the second containing 576 $\mathrm{I}_{3322}$ facets~\cite{PhysRevLett.49.1220,PhysRevA.64.014102,Collins_2004}. In particular, consider the following two vectors in $\mathbb{R}^{36}$ in notation analogous to~\eqref{Eq::TsirelsonNotation}:
\begin{equation}\label{eq:FCH}
    {\rm F_{CH}}= \left(
\begin{array}{cc|cc|cc}
 0 & 0 & 1 & 0 & 0 & 0\\
 0 & 1 & 0 & 0 & 0 & 0\\ \hline
 0 & -1 & 0 & 1 & 0 & 0\\
 0 & 0 & 0 & 0 & 0 & 0\\ \hline
 0 & 0 & 0 & 0 & 0 & 0\\
 0 & 0 & 0 & 0 & 0 & 0\\
\end{array}
\right), \quad \quad {\rm F_{I_{3322}}}=\frac{1}{3} \left(
\begin{array}{cc|cc|cc}
 0 & 1 & 0 & 1 & 0 &
   0 \\
 0 & 0 & 0 & 0 & 0 & 1 \\ \hline
 0 & 0 & 0 & 0 & 0 & 0 \\ 
 0 & -1 & 0 & -1 & 1 & 0 \\ \hline
 0 & 0 & 0 & 1 & 0 & 0 \\
 0 & 1 & 0 & 0 & 0 & 0 \\
\end{array}
\right).
\end{equation}
The CH facets are given by $\left<{\rm F_{CH}},\mathbf{x}\right> \leq 1$ and the $\mathrm{I}_{3322}$ facets are given by $\left<{\rm F_{I_{3322}}},\mathbf{x}\right> \leq 1$ where $\mathbf{x} \in \mathbb{R}^{36}$. The remaining elements for each class can be found by applying all relabelling symmetries to $\mathrm{F_{CH}}$ and $\mathrm{I_{3322}}$ respectively and then discarding duplicates corresponding to different representations of the same effect. There are 32 extreme states that satisfy $\left<{\rm F_{CH}},s\right> = 1$ and 32 extreme states that satisfy $\left<{\rm F_{CH}},s\right> = 0$. On the other hand, there are 20 extreme states with $\left<{\rm F_{I_{3322}}},s \right> = 1$, 28 with $\left<{\rm F_{I_{3322}}},s \right> = 2/3$, 12 with $\left<{\rm F_{I_{3322}}},s \right> = 1/3$ and 4 with $\left<{\rm F_{I_{3322}}},s \right> = 0$. There are at most 18 extreme local states that simultaneously saturate facets from each class. For any pair of facets (one from each class), there can be at most 18 local deterministic states that simultaneously saturate both.

We perform a vertex enumeration similar to the previous example to find that the effect polytope $\mathcal{E}_{\mathbb{H}^{[0]}_{(3,2)}}$ is given by the convex hull of 27968 extreme effects. A classification of these effects is provided in Table~\ref{Table::LocalEffects3Fid} of Appendix~\ref{Appendix::Effects3Fid}.

\subsubsection{Box World}
\label{Subsection::ExampleBW}

 Any GPT of gbits is said to be nonlocal if it is not a subtheory of GLT. An example is the maximal tensor product of $\mathcal{G}_{\rm m_\A}^{\rm n_\A}$ and $\mathcal{G}_{\rm m_\B}^{\rm n_\B}$ also called \textit{box-world} (BW). For the case of $\rm m_\A=m_\B=n_\A=n_\B=2$, the extreme states include the 16 local deterministic states and 8 entangled (non-separable) states called PR boxes~\cite{Popescu2014,pr,PhysRevA.75.032304}. We denote this state space as $\mathbb{H}^{[8]}_{(2,2)}$ and list the 8 PR boxes: 
$$
\pr_1=\frac{1}{2}\left(
\begin{array}{cc|cc}
 1 & 0 & 1 & 0 \\
 0 & 1 & 0 & 1 \\
 \hline
 1 & 0 & 0 & 1 \\
 0 & 1 & 1 & 0 \\
\end{array}
\right), \pr_2=\frac{1}{2}\left(
\begin{array}{cc|cc}
 0 & 1 & 1 & 0 \\
 1 & 0 & 0 & 1 \\
 \hline
 1 & 0 & 1 & 0 \\
 0 & 1 & 0 & 1 \\
\end{array}
\right), \pr_3=\frac{1}{2}\left(
\begin{array}{cc|cc}
 1 & 0 & 0 & 1 \\
 0 & 1 & 1 & 0 \\
 \hline
 1 & 0 & 1 & 0 \\
 0 & 1 & 0 & 1 \\
\end{array}
\right), \pr_4=\frac{1}{2}\left(
\begin{array}{cc|cc}
 0 & 1 & 0 & 1 \\
 1 & 0 & 1 & 0 \\
 \hline
 1 & 0 & 0 & 1 \\
 0 & 1 & 1 & 0 \\
\end{array}
\right), 
$$

$$
\pr_1'=\frac{1}{2}\left(
\begin{array}{cc|cc}
 0 & 1 & 0 & 1 \\
 1 & 0 & 1 & 0 \\
 \hline
 0 & 1 & 1 & 0 \\
 1 & 0 & 0 & 1 \\
\end{array}
\right), \pr_2'=\frac{1}{2}\left(
\begin{array}{cc|cc}
 1 & 0 & 0 & 1 \\
 0 & 1 & 1 & 0 \\
 \hline
 0 & 1 & 0 & 1 \\
 1 & 0 & 1 & 0 \\
\end{array}
\right), \pr_3'=\frac{1}{2}\left(
\begin{array}{cc|cc}
 0 & 1 & 1 & 0 \\
 1 & 0 & 0 & 1 \\
 \hline
 0 & 1 & 0 & 1 \\
 1 & 0 & 1 & 0 \\
\end{array}
\right), \pr_4'=\frac{1}{2}\left(
\begin{array}{cc|cc}
 1 & 0 & 1 & 0 \\
 0 & 1 & 0 & 1 \\
 \hline
 0 & 1 & 1 & 0 \\
 1 & 0 & 0 & 1 \\
\end{array}
\right);
$$
The pairs $\pr_{{\rm i}}$ and $\pr_{\rm i}'$ are called \emph{isotropically opposite} since equal mixtures of them give the maximally mixed state. 
The probability representation we use for a state involves specifying the distributions of outcomes of a fixed set of fiducial measurements and a fixed labelling of their outcomes. In general, relabelling the measurements or their outcomes gives an alternative description of the same state. For instance, if both parties relabel the first and second measurements, then $\pr_1$ becomes $\pr_2$. For bipartite (in general multipartite) systems, we can consider local relabellings, i.e., relabelling the inputs and/or outputs for each subsystem, and global relabellings, i.e., relabelling the subsystems. We will discuss the second kind in more detail in Section~\ref{Sec::MinimalPreservability}.

The complete set of 8 PR boxes can be generated by taking any one of them and applying local relabellings, and similarly all the local deterministic states can be generated by applying relabellings to the state given in~\eqref{eq:LDS}. Hence, there are two classes of extreme states. The extreme points of the effect space for this state space are the 82 separable effects that occur in the generalised local theory with 2 inputs and 2 outputs per party discussed above.

The probability tables $\pr$ above have a direct correspondence to the variations of the CHSH game introduced in Section~\ref{subsection::achsh_game}. The 8 vectors $\{{\mathrm{C}_i \coloneqq 1/2 \pr_i}\}_{i=1}^4$ and $\{{\mathrm{C}_i' \coloneqq 1/2 \pr_i'}\}_{i=1}^4$ define the 8 CHSH games in the sense that the winning probability is given by $\langle\mathrm{C}_i^(\phantom{}'\phantom{}_{\phantom{i}}^),\p(A,B|X,Y)\rangle$. For instance, since the correlation table obtained after performing the fiducial measurements on $\pr_1$ coincides with the probability table of $\pr_1$, we have
\begin{equation}
  \mathrm{CHSH}_1[\p_{\pr_1}(A,B|X,Y)] \coloneqq \left<\mathrm{C}_1,\p_{\pr_1}(A,B|X,Y)\right> = 
  \left< \left(\begin{array}{cc|cc}
 1/4 & 0 & 1/4 & 0 \\
 0 & 1/4 & 0 & 1/4 \\ \hline
 1/4 & 0 & 0 & 1/4 \\
 0 & 1/4 & 1/4 & 0 \\
\end{array}
\right),\frac{1}{2}\left(
\begin{array}{cc|cc}
 1 & 0 & 1 & 0 \\
 0 & 1 & 0 & 1 \\
 \hline
 1 & 0 & 0 & 1 \\
 0 & 1 & 1 & 0 \\
\end{array}
\right)\right> = 1.
\label{Eq::CHSHGame}
\end{equation}
The vectors $\mathrm{C}_i$ and effects $e_{\mathrm{CH_i}}$ are related by the affine transformation\footnote{Note that the affine transformation does not directly generate the same vectors, but generates vectors that represent the same effect.} $\mathrm{C}_i=e_{\mathrm{CH_i}}/2+u/4$, and the CH facets of the state space $\mathbb{H}^{[0]}_{(2,2)}$ can be rewritten as $\{\left<{\mathrm{ \mathrm{C}_i}},\mathbf{x}\right> \leq 3/4\}_i$ together with $\{\left<{\mathrm{ \mathrm{C}_i}},\mathbf{x}\right> \geq 1/4\}_i$ or $\{\left<{\mathrm{ \mathrm{C}'_i}},\mathbf{x}\right> \leq 3/4\}_i$.

When $\rm m_\A=m_\B=3$ and $\rm n_\A=\rm n_\B=2$, the state space corresponding to the max-tensor product has 1408 extreme states, out of which 64 are local deterministic. We denote this polytope $\mathbb{H}^{[1344]}_{(3,2)}$, where 1344 denotes the number of entangled extreme states in the state space. These entangled states can be classified into 4 relabelling classes~\cite{PhysRevA.72.052312}. One representative from each class is as follows:
$$
\mathrm{N}_1=\frac{1}{2}\left(
\begin{array}{cc|cc|cc}
 1 & 0 & 1 & 0 & 0 & 1 \\
 0 & 1 & 0 & 1 & 0 & 1 \\
 \hline
 1 & 0 & 0 & 1 & 0 & 1 \\
 0 & 1 & 1 & 0 & 0 & 1 \\
 \hline
 0 & 0 & 0 & 0 & 0 & 0 \\
 1 & 1 & 1 & 1 & 0 & 2 \\
\end{array}
\right)\!,\,
\mathrm{N}_2=\frac{1}{2}\left(
\begin{array}{cc|cc|cc}
 1 & 0 & 1 & 0 & 0 & 1 \\
 0 & 1 & 0 & 1 & 1 & 0 \\
 \hline
 1 & 0 & 0 & 1 & 0 & 1 \\
 0 & 1 & 1 & 0 & 1 & 0 \\
 \hline
 0 & 1 & 0 & 1 & 0 & 1 \\
 1 & 0 & 1 & 0 & 1 & 0 \\
\end{array}
\right)\!,\,
\mathrm{N}_3=\frac{1}{2}\left(
\begin{array}{cc|cc|cc}
 1 & 0 & 1 & 0 & 1 & 0 \\
 0 & 1 & 0 & 1 & 0 & 1 \\
  \hline
 1 & 0 & 0 & 1 & 1 & 0 \\
 0 & 1 & 1 & 0 & 0 & 1 \\
  \hline
 1 & 0 & 1 & 0 & 1 & 0 \\
 0 & 1 & 0 & 1 & 0 & 1 \\
\end{array}
\right)\!,\, \mathrm{N}_4=\frac{1}{2}\left(
\begin{array}{cc|cc|cc}
 1 & 0 & 1 & 0 & 1 & 0 \\
 0 & 1 & 0 & 1 & 0 & 1 \\
 \hline
 1 & 0 & 0 & 1 & 0 & 1 \\
 0 & 1 & 1 & 0 & 1 & 0 \\
 \hline
 0 & 0 & 0 & 0 & 0 & 0 \\
 1 & 1 & 1 & 1 & 1 & 1 \\
\end{array}
\right)\!.
$$
Each entangled state violates multiple ${\rm F_{CH}}$-type and ${\rm F_{I_{3322}}}$-type facets and each ${\rm F_{CH}}$ or ${\rm F_{I_{3322}}}$ have multiple entangled states violating them as indicated in Tables~\ref{Tab:NLSClassvsFacets} and~\ref{Tab:FacetsvsNLSClasses}.

\begin{table}[H]
\centering
\begin{tabular}{ |c||c|c|c| }
\hline
  \hspace{5mm}Class\hspace{5mm} & \hspace{5mm}$\#$\hspace{5mm} & \hspace{5mm}$\# \rm F_{CH}$\hspace{5mm} & \hspace{5mm}$\# \rm F_{I_{3322}}$\hspace{5mm} \\
 \hline \hline
 $\rm N_1$   & 288   &  1  & 8\\ 
 $\rm N_2$   & 192   &  6  & 18\\
 $\rm N_3$   & 288   &  4  & 24\\
 $\rm N_4$   & 576   &  2  & 12\\ 
   \hline  
\end{tabular}
\caption{Table above summarises $``\#"$ the number of elements in each class, $``\# {\rm F_{CH}}"$ and $``\# {\rm F_{I_{3322}}}"$ the number of ${\rm F_{CH}}$-type and ${\rm F_{I_{3322}}}$-type facets violated by any element of the respective class.}
\label{Tab:NLSClassvsFacets}
\end{table}

\begin{table}[H]
\centering
\begin{tabular}{ |c||c|c|c|c|  }
\hline
  \hspace{5mm}Inequality\hspace{5mm} & \hspace{5mm}$\# \rm N_1$\hspace{5mm} & \hspace{5mm}$\# \rm N_2$\hspace{5mm} & \hspace{5mm}$\# \rm N_3$\hspace{5mm} & \hspace{5mm}$\# \rm N_4$\hspace{5mm} \\
 \hline \hline
 ${\rm F_{I_{3322}}}$   & 4 & 6 & 12 & 12 \\
  ${\rm F_{CH}}$       & 4 & 16 & 16 & 16\\ 
  \hline  
\end{tabular}
\caption{Table summarising the number of entangled states from each class violating a single facet.}
\label{Tab:FacetsvsNLSClasses}
\end{table}

The effect polytope of $\mathcal{E}_{\mathbb{H}^{[1344]}_{(3,2)}}$ has 248 extreme effects, all of which are separable~\cite{Short_2010} and all are also extreme effects of $\mathcal{E}_{\mathbb{H}^{[0]}_{(3,2)}}$. These 248 effects can be classified into 7 relabelling classes.  A classification of these effects can be found in~\cite{PhysRevA.108.062212} which we summarise in Table~\ref{Table::BWEffects3Fid} of Appendix~\ref{Appendix::Effects3Fid} for completeness.

\subsection{The Set of Correlations}
\label{SubSec::Correlations}
The probability state space is not necessarily in one-to-one correspondence with the set of correlations that can be generated from that state space. For instance, consider a theory in which the extremal points of the bipartite state space are all the local deterministic probability tables and that of $\pr_1$. With such a state space it is possible to generate the correlations of $\pr_1'$ by using the state $\pr_1$, and the fiducial measurements, by relabelling Alice's outcomes. Thus, although the probability table corresponding to $\pr_1'$ is not in the state space, the correlations it would give can be generated. This also means that given a set of correlations that a theory can produce, it is not always possible to uniquely identify the underlying state space.

\subsection{Entanglement Swapping in GPTs}
\label{subsec::EntSwap}
The standard entanglement swapping scenario in quantum theory involves Bob sharing one maximally entangled qubit pair with Alice and one with Charlie. Bob then performs a joint measurement on his two qubits in the Bell basis and announces his outcome. No matter which outcome occurs, conditioned on this outcome, Alice's qubit and Charlie's qubits are maximally entangled. More generally, we consider Bob's measurement to be entanglement swapping if at least one outcome the state between Alice and Charlie is entangled. This concept extends to GPTs, and, following~\cite{PhysRevLett.102.110402,Skrzypczyk_2009}, we define a coupler. 
\begin{definition}
Let $\mathcal{S}$ and $\mathcal{E}$ be a bipartite state space and compatible effect space respectively. 
An effect $e \in \mathcal{E}$ is said to be a \emph{coupler} if there exist states $s_{\A\B},\ s_{\B'\C}\in\mathcal{S}$ such that
\begin{equation}
    s_{\A\C|e} = \frac{\id_{\A} \otimes e \otimes \id_{\C}\left(s_{\A\B} \otimes s_{\B'\C}\right)}{\left<u,\id_{\A} \otimes e \otimes \id_{\C}\left(s_{\A\B} \otimes s_{\B'\C}\right)\right>}
\end{equation}
is an element of the state space $\mathcal{S}$ and is entangled.
\end{definition}
Note that an effect cannot be a coupler if it is separable, i.e., if it can be expressed as a sum of product effects. However, it can happen that a theory has entangled effects that are not couplers when there is no pair of entangled states $s_{\A\B}$ and $s_{\B'\C}$ in the state space that lead to an entangled output state on application of the effect. We will see examples of this in the following sections.

Entanglement swapping is a key ingredient in achieving a post-classical score in the ACHSH game as shown in~\cite{PhysRevLett.125.060406,PhysRevA.102.022203}. Theories like box-world can generate correlations that can perfectly win CHSH games but do not have any couplers (see~\cite{PhysRevA.73.012101} for the case when the single party state space is $\mathcal{G}^2_2$). In general, theories in which the state space is formed by the maximal tensor product (cf.\ Definition~\ref{def::MinAndMax}) have the smallest effect cone, with all effects being separable, and there is a trade-off between states and effects~\cite{Short_2010}. In order for a theory to have entanglement swapping, it needs to allow both entangled states and entangled effects, ruling out a state space formed by the min or max tensor products. In this regard, quantum theory lies in an intermediate spot in which for every quantum state $\rho$ and any number $t \in [0,1]$, $t\rho$ is an allowed effect (see Footnote~\ref{Foot::SelfDual}). In~\cite{PhysRevLett.102.110402} the authors considered the state space $\mathbb{H}^{[1]}_{(2,2)}[\pr_1]$, characterised by the convex-hull of $\mathbb{H}^{[0]}_{(2,2)}$ and $\pr_1$. When $\pr_1$ is added to the state space $\mathbb{H}^{[0]}_{(2,2)}$, $e_{\mathrm{CH}_1}$ (see~\eqref{Eq::CHFacets}) ceases to satisfy $\langle e_{\mathrm{CH}_1},\mathbf{x} \rangle \leq 1$ for all states $\mathbf{x}$; in particular, $\langle e_{\mathrm{CH}_1},\pr_1 \rangle = 3/2$. The facets of $\mathbb{H}^{[1]}_{(2,2)}[\pr_1]$ are otherwise the same as those of $\mathbb{H}^{[0]}_{(2,2)}$. Since $\pr_1$ is the only extremal state for which the inner product with $e_{\mathrm{CH}_1}$ is greater than 1, when scaled by $2/3$, the resultant vector, $2/3 e_{\mathrm{CH}_1}$, becomes a valid effect. In~\cite{PhysRevLett.102.110402} it was pointed out that this effect is a coupler. In addition, a correspondence was given between the facets of the state-space polytope and the extremal effects of its maximal effect space polytope. This correspondence only gives the extremal effects that lie on the extremal rays of the maximal effect space polytope, rather than the complete set of extremal effects, which we construct below.

We have seen in Subsection~\ref{SubSec::Correlations} that in the $[2,2]$ setting, the state space $\mathbb{H}^{[1]}_{(2,2)}[\pr_1]$ generates all non-signalling correlations. Therefore, $\mathbb{H}^{[1]}_{(2,2)}[\pr_1]$ is a potential example in which one might achieve a higher score in the ACHSH game compared to quantum theory. Here we revisit the example from~\cite{PhysRevLett.102.110402} but using $\pr_2$ instead of $\pr_1$ (because these two states are the same up to local relabellings, the fact that we treat $\pr_2$ instead makes no essential difference).

\begin{figure}[h]
  \begin{center}
    \includegraphics[width=0.45\textwidth]{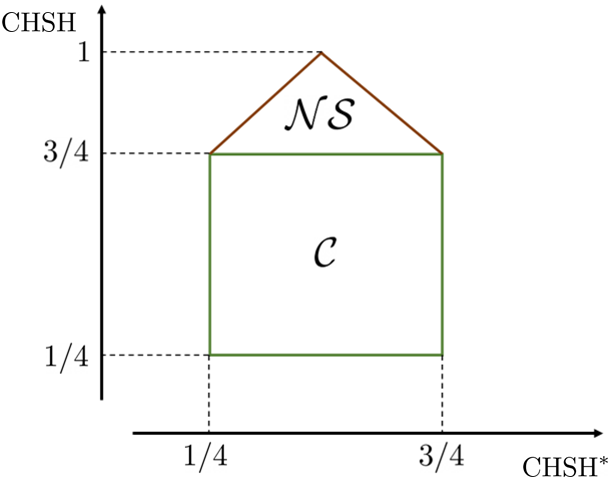}
  \end{center}
  \caption{A two dimensional slice of the set of correlations generated when fiducial measurements are performed on the states of the bipartite state space characterised by the 16 local deterministic states and one PR box~\cite{PhysRevLett.102.110402}. The vertical axes represent a CHSH inequality and the horizontal axes represent one of its symmetries obtained by relabelling the inputs. Local correlations, denoted by the square $\mathcal{C}$,  satisfy $ 1/4 \leq \mathrm{CHSH}[\p] \leq 3/4$ and $1/4 \leq \mathrm{CHSH^*}[\p] \leq 3/4$. } 
  \label{fig:House1} 
\end{figure}

We consider the state space $\mathbb{H}^{[1]}_{(2,2)}[\pr_2]$.  Solving the vertex enumeration problem for the effect polytope, we found that $\mathcal{E}_{\mathbb{H}^{[1]}_{(2,2)}}$ has 106 extreme effects, of which 82  are the extreme effects of $\mathbb{H}^{[8]}_{(2,2)}$.  We call these 82 effects the boxworld (BW) effects. Of the 24 non-BW effects, 9 are couplers. When applied by Bob to two halves of the state $\pr_2$ (one shared with Alice and one with Charlie), for one of these effects, $e_{\mathrm{pure}}$, the resultant state on Alice and Charlie is extremal, while for the other 8 collected in the set $E_{\mathrm{noisy}}$ the resultant state is not, where
\begin{equation}
e_{\mathrm{pure}} = \left(
\begin{array}{cc|cc}
 0 & 0 & 0 & 0 \\
 0 & -2/3 & 0 & 2/3 \\ \hline
 0 & 0 & 2/3 & 0 \\
 0 & 2/3 & 0 & 0 \\
\end{array}
\right) =  2e_{\mathrm{CH}_2}/3,
\label{Ex::CouplerSymm}
\end{equation}
and 
\begin{equation}
    E_{\mathrm{noisy}} = \left\{ \begin{array}{l}
        \!\left(
\begin{array}{cc|cc}
 0 & 1/2 & 0 & 0 \\
 0 & -1/2 & 0 & 1/2 \\ \hline
 0 & 0 & 1/2 & 0 \\
 0 & 1/2 & 0 & 0 \\
\end{array}
\right)\!,
\left(
\begin{array}{cc|cc}
 0 & 0 & 0 & 0 \\
 1/2 & -1/2 & 0 & 1/2 \\ \hline
 0 & 0 & 1/2 & 0 \\
 0 & 1/2 & 0 & 0 \\
\end{array}
\right)\!,
\left(
\begin{array}{cc|cc}
 0 & 0 & 1/2 & 0 \\
 0 & -1/2 & 0 & 1/2 \\ \hline
 0 & 0 & 1/2 & 0 \\
 0 & 1/2 & 0 & 0 \\
\end{array}
\right)\!,
\left(
\begin{array}{cc|cc}
 0 & 0 & 0 & 0 \\
 0 & -1/2 & 0 & 1 \\ \hline
 0 & 0 & 1/2 & 0 \\
 0 & 1/2 & 0 & 0 \\
\end{array}
\right)\!, \\
\!\left(
\begin{array}{cc|cc}
 0 & 0 & 0 & 0 \\
 0 & -1/2 & 0 & 1/2 \\ \hline
 1/2 & 0 & 1/2 & 0 \\
 0 & 1/2 & 0 & 0 \\
\end{array}
\right)\!,
\left(
\begin{array}{cc|cc}
 0 & 0 & 0 & 0 \\
 0 & -1/2 & 0 & 1/2 \\ \hline
 0 & 0 & 1/2 & 0 \\
 0 & 1 & 0 & 0 \\
\end{array}
\right)\!,
\left(
\begin{array}{cc|cc}
 0 & 0 & 0 & 0 \\
 0 & -1/2 & 0 & 1/2 \\ \hline
 0 & 0 & 1 & 0 \\
 0 & 1/2 & 0 & 0 \\
\end{array}
\right)\!,
\left(
\begin{array}{cc|cc}
 0 & 0 & 0 & 0 \\
 0 & -1/2 & 0 & 1/2 \\ \hline
 0 & 0 & 1/2 & 0 \\
 0 & 1/2 & 0 & 1/2 \\
\end{array}
\right)\!
    \end{array}\!\right\}.\label{Eq::MinCoupler}
\end{equation}
If Bob performs the joint measurement $\{e_{\mathrm{pure}},u-e_{\mathrm{pure}}\}$, then with probability $1/3$ the outcome corresponding to $e_{\rm pure}$ occurs and the post-measurement state is
\begin{equation}
    \frac{\id_{\A} \otimes e_{\mathrm{pure}} \otimes \id_{\C}\left(\left(\pr_2\right)_{\A\B} \otimes \left(\pr_2\right)_{\B'\C}\right)}{\left<u,\id_{\A} \otimes  e_{\mathrm{pure}} \otimes \id_{\C}\left(\left(\pr_2\right)_{\A\B} \otimes\left(\pr_2\right)_{\B'\C}\right)\right>} = \left(\pr_2\right)_{\A\C}.
\end{equation}
Likewise, if Bob measures $\{e_{\mathrm{noisy}},u-e_{\mathrm{noisy}}\}$ instead, where $e_{\rm noisy}$ is the first element of $E_{\rm noisy}$, we find that with probability $3/8$, the outcome corresponding to $e_{\rm noisy}$ occurs and the post-measurement state is
\begin{equation}\label{eq:15}
    \frac{\id_{\A} \otimes e_{\mathrm{noisy}} \otimes \id_{\C}\left(\left(\pr_2\right)_{\A\B} \otimes \left(\pr_2\right)_{\B'\C}\right)}{\left<u,\id_{\A} \otimes  e_{\mathrm{noisy}} \otimes \id_{\C}\left(\left(\pr_2\right)_{\A\B} \otimes\left(\pr_2\right)_{\B'\C}\right)\right>} = \frac{2}{3}(\pr_2)_{\A\C} + \frac{1}{3} \left(
\begin{array}{cc|cc}
 0 & 0 & 0 & 0 \\
 1 & 0 & 0 & 1 \\
 \hline
 1 & 0 & 0 & 1 \\
 0 & 0 & 0 & 0 \\
\end{array}
\right)_{\A\C}.
\end{equation}
The local deterministic state at the end of this equation has value $3/4$ for $\mathrm{CHSH}_2$, and hence the state on the right-hand-side of~\eqref{eq:15} is entangled. When another effect from the set $E_{\rm noisy}$ is used instead, we get a similar decomposition with the corresponding local deterministic state also having value $3/4$ for $\mathrm{CHSH_2}$. This implies that all of the effects in~\eqref{Eq::MinCoupler} are couplers.

Additionally, note that among the extremal effects that are couplers, only $e_{\mathrm{pure}}$ is ray-extremal. The couplers in the set $E_{\mathrm{noisy}}$ are not, which is why they were not found in~\cite{PhysRevLett.102.110402}. 

\section{Bipartite Compositions of \texorpdfstring{$\mathcal{G}_2^2$ }{}and their Effect Polytopes}
\label{Section::EffectPolytope}
We have shown in Subsection~\ref{SubSec::Correlations}, that both the bipartite gbit state spaces $\mathbb{H}^{[1]}_{(2,2)}$ and $\mathbb{H}^{[8]}_{(2,2)}$ generate all no-signalling correlations in the $[2,2]$ setting. We have seen that $\mathbb{H}^{[8]}_{(2,2)}$ has no couplers~\cite{PhysRevA.73.012101}, which also implies that the ACHSH game cannot be won with a success probability higher than that of GLT. The same argument does not hold for $\mathbb{H}^{[1]}_{(2,2)}$.
Here, we proceed to study more generally whether there are other bipartite gbit theories that support couplers and whether such theories outperform quantum theory in the ACHSH game.
 
To do so we consider a bipartite gbit state space characterised by the convex hull of $\mathbb{H}^{[0]}_{(2,2)}$ and $\npr$ noisy PR boxes of the form
\begin{equation}
    \pr_{i,\alpha_i} \coloneqq \alpha_i \pr_i + (1-\alpha_i)\frac{\mathbb{I}}{4},
    \label{Eq::NoiseModel}
\end{equation}
where $\alpha_i \in [1/2,1]$. In this work, we only consider the scenario where the amount of noise is the same on all of the PR boxes i.e., $\alpha_i=\alpha$ and denote such a state space by $\mathbb{H}^{[\npr]}_{\alpha (2,2)}$. Note that $\pr_{i,1/2}$ is local, so $\mathbb{H}^{[\npr]}_{1/2 (2,2)}=\mathbb{H}^{[0]}_{(2,2)}$.  Therefore in the following, we restrict the range of $\alpha$ to the interval $(1/2,1]$, unless specified otherwise. To check whether a state space of this form supports couplers, we find the corresponding extreme effects of the maximal effect space. Extreme effects are useful because if couplers are present in the effect polytope, at least one of the extreme effects must be a coupler.

In the following, we first describe the effect polytope for state spaces with one noisy PR box and then show how to generalise to two or more noisy PR boxes. Note that $\mathbb{H}^{[\npr]}_{\alpha(2,2)}$ is not a unique state space for a fixed $\npr,\alpha$ (it depends on the choice of $\npr$ noisy PR boxes).
 
\subsection{State spaces with 1 noisy PR-box extremal state} 
\label{Subsection::1Roof}

The 8 PR boxes are equivalent up to relabelling symmetries (see Section~\ref{Subsection::ExampleBW}) and so, without loss of generality, we consider $\pr_2$ here. The state space $\mathbb{H}^{[1]}_{\alpha(2,2)}[\pr_2]$ is characterised by 23 facets of which 16 are positivity facets and 7 are CH facets. These are the same as the facets of $\mathbb{H}^{[0]}_{(2,2)}$, but with $\left<e_{\rm CH_2},\mathbf{x}\right> \leq 1$ removed\footnote{Note that $\left<e_{\rm CH_2},\pr_{2,\alpha}\right> > 1$ when $\alpha \in (1/2,1]$.}. Furthermore, as discussed in Section~\ref{Subsection::ExampleGLT}, the maximal effect polytope of $\mathbb{H}^{[0]}_{(2,2)}$ has 90 extreme effects, which includes 82 BW effects~\cite{PhysRevA.73.012101} and 8 entangled effects $\{e_{\rm CH_i}\}_{i=1}^4$ and $\{u-e_{\rm CH_i}\}_{i=1}^4$. The maximal effect polytope of $\mathbb{H}^{[1]}_{\alpha(2,2)}[\pr_2]$ is the subset of the maximal effect polytope of $\mathbb{H}^{[0]}_{(2,2)}$ that is contained in the intersection of the half-spaces satisfying $\left< \mathbf{x}, \mathrm{PR}_{2,\alpha} \right> \leq 1$ and $\left< \mathbf{x}, \mathrm{PR}_{2,\alpha} \right> \geq 0$, where $\mathbf{x} \in \mathbb{R}^{16}$.  Since $\left<e_{\rm CH_2},\pr_{2,\alpha}\right> > 1$,  $e_{\rm CH_2}$ and $u-e_{\rm CH_2}$ cease to be valid effects of $\mathbb{H}^{[1]}_{\alpha (2,2)}[\pr_2]$. The remaining 88 extreme effects of $\mathbb{H}^{[0]}_{(2,2)}$ are valid effects for $\mathbb{H}^{[1]}_{\alpha (2,2)}[\pr_2]$ and remain extreme. 

To find the new extremal effects for this state space we take each of the effects of $\mathbb{H}^{[0]}_{(2,2)}$ that are not valid for $\mathbb{H}^{[1]}_{\alpha (2,2)}[\pr_2]$ and form line segments from them to each of the other extremal effects, identifying the point on the line where the vector becomes a valid effect. The set formed in this way then needs to be reduced to its extreme elements (see Appendix~\ref{Appendix::EffectPoly} for more details).
Using this technique we find that all additional extremal effects for this state space are vectors $\mathbf{x} \in \mathbb{R}^{16}$ that satisfy $\left< \mathbf{x}, \mathrm{PR}_{2,\alpha} \right> = 0$ or $\left< \mathbf{x}, \mathrm{PR}_{2,\alpha} \right> = 1$. We found that these come in 4 types up to relabelling. A candidate effect of each type lying on the hyperplane $\left< \mathbf{x}, \mathrm{PR}_{2,\alpha} \right> = 0$ is as follows: 
\begin{equation}
    \begin{split}
        \text{Type 1}&: \quad \frac{1-\alpha}{\alpha} e_{\rm CH_2} +\left(1-\frac{1-\alpha}{\alpha}\right)\left(
\begin{array}{cc|cc}
 0 & 1 & 0 & 0 \\
 1 & 0 & 0 & 0 \\ \hline
 0 & 0 & 0 & 0 \\
 0 & 0 & 0 & 0 \\
\end{array}
\right), \\
\text{Type 2}&: \quad  \frac{1-\alpha}{3\alpha-1}e_{\rm CH_2} + \left(1- \frac{1-\alpha}{3\alpha-1}\right) \left(
\begin{array}{cc|cc}
 1 & 1 & 0 & 0 \\
 1 & 0 & 0 & 0 \\ \hline
 0 & 0 & 0 & 0 \\
 0 & 0 & 0 & 0 \\
\end{array}
\right),\\
\text{Type 3}&: \quad \frac{3-\alpha}{3\alpha+1} e_{\rm CH_2} + \left(1-\frac{3-\alpha}{3\alpha+1}\right)\left(
\begin{array}{cc|cc}
 0 & 1 & 0 & 0 \\
 0 & 0 & 0 & 0 \\ \hline
 0 & 0 & 0 & 0 \\
 0 & 0 & 0 & 0 \\
\end{array}
\right) \eqqcolon e_{\rm m,\alpha},\\
\text{Type 4}&: \quad \frac{2}{3}e_{\rm CH_2} \eqqcolon e_{\rm p,\alpha}.
    \end{split}
\end{equation}
On the hyperplane $\left< \mathbf{x}, \mathrm{PR}_{2,\alpha} \right> = 0$, there are 12 effects of Type 1, 8 of Type 2, 8 of Type 3 and 1 of Type 4. Their complementary effects (the effects formed by subtracting them from the unit effect) are also extreme effects and lie on the hyperplane $\left< \mathbf{x}, \mathrm{PR}_{2,\alpha} \right> = 1$.

Collecting all of these we find that the maximal effect polytope of $\mathbb{H}^{[1]}_{\alpha  (2,2)}[\pr_2]$ is the convex hull of 146 extreme effects. These include 82 BW effects, 6 CH type effects, 29 effects satisfying $\left< \Tilde{e} , \mathrm{PR}_{2,\alpha} \right>=1$ and 29 effects satisfying $\left< \Tilde{e} , \mathrm{PR}_{2,\alpha} \right>=0$. Note that when $\alpha \to 1$, all the effects satisfying $\left< \Tilde{e} , \mathrm{PR}_{2,\alpha} \right>=1$ from the first two types converge to deterministic effects and their complementary effects ($u-\Tilde{e}$) converge to the complementary deterministic effects. This leaves 106 extremal effects of $\mathbb{H}^{[1]}_{(2,2)}[\pr_2]$ in agreement with the example from Section~\ref{subsec::EntSwap}.

We next consider which of these extreme effects are couplers. A short calculation shows that 
\begin{align}    
       \chsh_2 \left[ \frac{\id_{\A} \otimes e_{\B_1\B_2} \otimes \id_{\C}\left(\left(\pr_{2,\alpha}\right)_{\A\B_1} \otimes \left(\pr_{2,\alpha}\right)_{\B_2\C} \right)}{\left<u,\id_{\A} \otimes e_{\B_1\B_2} \otimes \id_{\C}\left(\left(\pr_{2,\alpha}\right)_{\A\B_1} \otimes \left(\pr_{2,\alpha}\right)_{\B_2\C} \right)\right>} \right]   = 
       \begin{cases}     
       \hfill \frac{\alpha +2}{4} \quad &\text{if } e \in \mathrm{Type\ 1} \\
        \hfill \frac{\alpha  (\alpha +10)-4}{20 \alpha -8} \quad &\text{if } e \in \mathrm{Type\ 2} \\
        \hfill \frac{5\alpha^2+2\alpha+4}{4(\alpha+2)} \quad &\text{if } e \in \mathrm{Type\ 3} \\
       \hfill \frac{\alpha ^2+1}{2}   \quad &\text{if } e \in \mathrm{Type\ 4} \\
        \end{cases} ,
        \label{Eq::PostStateCHSH}
\end{align}
which are shown in Fig.~\ref{Fig::CHSHScoresExtreme}. Effects in Type 3 are couplers in the range $(1+\sqrt{41})/10 < \alpha \leq 1$ and the effect in Type 4 is a coupler for $1/\sqrt{2} < \alpha \leq 1$. Note that $\alpha > 1/\sqrt{2}$ corresponds to the state spaces having nonlocality strictly larger than Tsirelson's bound. Additionally, for the couplers of Type 3, the post measurement state will have a CHSH value larger than Tsirelson's bound when $\alpha > 1/10(\sqrt{2}+\sqrt{2(1+20\sqrt{2})})$. Similarly, for the coupler from Type 4, this corresponds to $\alpha > 1/\root 4\of 2$.  These are the only extremal effects that are couplers.

\begin{figure}
    \centering
    \includegraphics[width=0.65\textwidth]{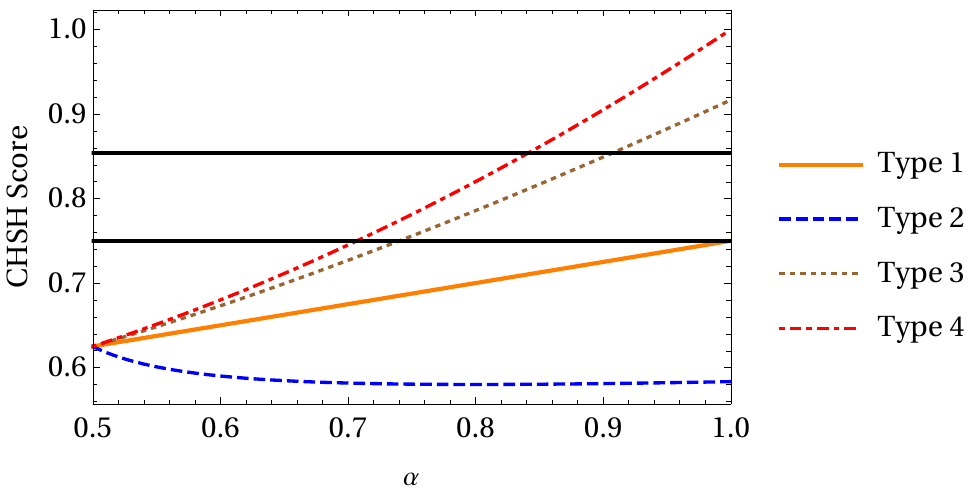}
    \caption{ A plot of the CHSH scores of the renormalised states obtained when an effect $\Tilde{e}$ from each of the four types is applied in the middle half of the 4-partite state $\pr_{2,\alpha} \otimes \pr_{2,\alpha}$. The red line is obtained when $\Tilde{e} \in \mathrm{Type\ 4}$.  The brown line is obtained for any $\Tilde{e} \in \mathrm{Type\ 3}$.   The yellow and blue lines are obtained when any $\Tilde{e}$ is taken from Type 1 and Type 2 respectively. The straight horizontal black lines represent the classical score $3/4$ and Tsirelson's bound.}
    \label{Fig::CHSHScoresExtreme}
\end{figure}

We now consider the probability of successful coupling by using these coupling effects within measurements $\{e_{\mathrm{p},\alpha},u-e_{\mathrm{p},\alpha}\}$ and $\{e_{\mathrm{m},\alpha},u-e_{\mathrm{m},\alpha}\}$. If Bob shares two copies of the state $\rm{PR}_{\alpha}$, one with Alice and another with Charlie, then the probability of successful entanglement swapping can be expressed in terms of $\alpha$ as
\begin{equation}
    p_{\mathrm{success}}=\left<u,\id \otimes e \otimes  \id\left( \rm{PR}_{\alpha} \otimes  \rm{PR}_{\alpha}\right)\right> =\begin{cases}
        \frac{1}{1+2\alpha} \quad \text{if } e = e_{\mathrm{p}}\\
        \frac{2+\alpha}{2+6\alpha} \quad \text{if } e = e_{\mathrm{m}}
    \end{cases}.
    \label{eq:p_perfect_coupling}
\end{equation}
  Note that when $\alpha=1$, $e_{p,\alpha=1}=e_{\rm pure}$ and $e_{m,\alpha=1} \in E_{\mathrm{noisy}}$ with the success probability being $1/3$ for $e_{\rm p}$ and $3/8$ for $e_{\rm m}$ which matches with the example discussed in Section~\ref{subsec::EntSwap}.

\subsection{State spaces with 2 noisy PR-box extremal states}
\label{Subsection::2Roofs}

 In this section, we consider state spaces with 2 noisy PR boxes and perform a similar analysis. There are $\binom{8}{2}=28$ pairs of PR boxes. A pair of PR boxes $(\pr_i,\pr_j)$ is said to be equivalent to another pair $(\pr_k,\pr_l)$ if there exists a local relabelling operation $R$ such that $R[\pr_i]=\pr_k$ and $R[\pr_j]=\pr_l$. We found that there are two classes of pairs of PR boxes. $(\pr_1,\pr_2)$ are an instance of the first class and $(\pr_2,\pr_2')$ are an instance of the second (cf.\ the definition of $\{\pr_i\}$ in Section~\ref{Subsection::ExampleBW}). Therefore, analysis of these two cases covers all possibilities with 2 PR boxes up to symmetry. We will first look into state spaces where the pair is isotropically opposite and then investigate the other case.

We use $\mathbb{H}^{[2]}_{\alpha (2,2)}[\pr_{2,2'}]$ to denote the state space characterised by the convex hull of $\mathbb{H}^{[1]}_{\alpha (2,2)}[\pr_2]$ and the noisy PR box $\pr_{2,\alpha}'$. This state space is characterised by 16 positivity facets and 6 Bell facets. In particular, $\left\<u-e_{\rm CH_2},\mathbf{x}\right> \leq 1$ which is a facet of $\mathbb{H}^{[1]}_{\alpha (2,2)}[\pr_2]$ is no longer a facet of $\mathbb{H}^{[2]}_{\alpha (2,2)}[\pr_{2,2'}]$\footnote{Note that $\left\<u-e_{\rm CH_2},\pr_{2,\alpha}'\right> > 1$.}. The maximal effect space of $\mathbb{H}^{[2]}_{\alpha (2,2)}[\pr_{2,2'}]$ is the subset of the maximal effect space of $\mathbb{H}^{[1]}_{\alpha (2,2)}[\pr_{2}]$ that is contained in the intersection of the half spaces satisfying $\left< \mathbf{x}, \mathrm{PR}_{2',\alpha} \right> \leq 1$ and $\left< \mathbf{x}, \mathrm{PR}_{2',\alpha} \right> \geq 0$. Using the same technique as in the case of one entangled state we calculated all the extreme effects of the maximal effect space and found that there are no couplers. The extreme effects include 82 BW effects, 6 CH type effects, 12 effects shared by the hyperplanes $\left< \mathbf{x}, \mathrm{PR}_{2,\alpha} \right>=0$ and $\left< \mathbf{x}, \mathrm{PR}_{2,\alpha}' \right>=1$, 12 effects shared by the hyperplanes $\left< \mathbf{x}, \mathrm{PR}_{2,\alpha} \right>=1$  and $\left< \mathbf{x}, \mathrm{PR}_{2,\alpha}' \right>=0$ and 8 Type 2 effects lying on each of these four hyperplanes, making it a total of 144 extreme effects.  We refer the reader to Appendix~\ref{Appendix::EffectPoly22p} for more details. 

Note that a related case with two noisy PR-boxes was considered in~\cite{Skrzypczyk_2009}, but with different amounts of noise on each of the added PR-boxes. A special case of this is when the noise is the same for both, and~\cite{Skrzypczyk_2009} found that no couplers occur here, which our analysis above agrees with.
 
Next, we consider the second state space $\mathbb{H}^{[2]}_{\alpha (2,2)}[\pr_{1,2}]$ where the two noisy PR boxes are not isotropically opposite to each other. We found that the only extreme effects of $\mathbb{H}^{[1]}_{\alpha (2,2)}[\pr_{2}]$ that cease to be valid now are $e_{\rm CH_1}$ and $e_{\rm CH_{1'}}$. Following the same construction as above, we found that the extreme effects lying on the hyperplane $\left< \mathbf{x}, \mathrm{PR}_{1,\alpha} \right>=0$
are exactly of the form of Types 1, 2, 3 and 4 and hence their complementary effects lie on the hyperplane $\left< \mathbf{x}, \mathrm{PR}_{1,\alpha} \right>=1$ and are extreme. These constitute the new extreme effects. Then, extreme points of the maximal effect space comprise 82 BW effects, 4 CH type effects and 29 effects from each of the 4 hyperplanes, giving a total of 202 extreme effects (see Appendix~\ref{Appendix::EffectPoly12} for more details). Since all the extreme effects of $\mathbb{H}_{\alpha (2,2)}^{[1]}[\pr_2]$, except $e_{\rm CH_1}$ and $e_{\rm CH_{1'}}$, are still extreme effects here, this state space does have couplers, in particular the Type 3 and Type 4 effects.  In Section~\ref{Sec::MinimalPreservability}, we will use a compositional consistency criterion to show that all these couplers are in fact inconsistent, and hence entanglement swapping is not possible for these state spaces with 2 entangled states.

\subsection{General Algorithm for state spaces with \texorpdfstring{$\npr$}{\npr} noisy PR-box extremal states}
\label{Subsection::mRoofs}
Finally, let us focus on the general case, $\mathbb{H}^{[\npr]}_{\alpha (2,2)}$, which is the state space characterised by the convex hull of 16 local deterministic boxes and $\npr$ noisy PR boxes. For convenience in this discussion consider the case where the $\npr$ PR boxes are $\pr_{i,\alpha}$ for $i \in \{1,\ldots,\npr\}$. We find the extreme points of the effect polytope of $\mathbb{H}^{[\npr]}_{\alpha (2,2)}$ using the following steps.
\begin{itemize}
\item[] \textbf{Step 1.} For each $i$, find the extremal effects of the state space $\mathbb{H}^{[1]}_{\alpha(2,2)}[\pr_i]$ by applying a local relabelling symmetry to the extremal effects of $\mathcal{E}_{\mathbb{H}^{[1]}_{\alpha (2,2)}[\pr_{2}]}$ found previously in Section~\ref{Subsection::1Roof}. [More precisely, if $R_i$ is a relabelling operation such that $R_i[\pr_2]=\pr_i$, then the extreme effects of $\mathbb{H}^{[1]}_{\alpha (2,2)}[\pr_i]$ are
    $$
    \mathrm{Extreme}\left[ \mathcal{E}_{\mathbb{H}^{[1]}_{\alpha (2,2)}[\pr_i]} \right] = \left\{R_i[e]\ |\ e \in    \mathrm{Extreme}\left[ \mathcal{E}_{\mathbb{H}^{[1]}_{\alpha (2,2)}[\pr_2]} \right]  \right\}.
    $$    
    \item[] \textbf{Step 2.} Define $V_E$ as the union of all the extreme effects found in each case, i.e., 
    $$ V_E \coloneqq \bigcup_{i=1}^m \mathrm{Extreme}\left[\mathcal{E}_{\mathbb{H}^{[1]}_{\alpha (2,2)}[\pr_{i}]}\right].$$
  \item[]  \textbf{Step 3.} Take each of the elements of $V_E$ and compute the inner product with each $\pr_i$, discarding any whose inner product is outside $[0,1]$.
\end{itemize}
In Step 3, discarding suffices since all new effects that might arise are captured when computing the extremal effects of $\mathcal{E}_{\mathbb{H}^{[1]}_{\alpha (2,2)}[\pr_{2}]}$ in Step 1 (see Appendix~\ref{Appendix::EffectPoly} for details).

The number of extreme effects can be counted as follows. As discussed earlier, there are 90 extreme effects from $\mathbb{H}^{[0]}_{(2,2)}$, including 82 BW effects and 8 CH type effects. Suppose that in the state space $\mathbb{H}^{[\npr]}_{\alpha (2,2)}$ with $1/2 < \alpha < 1$ there are $t$ pairs of noisy PR boxes that are isotropically opposite to each other. From Section~\ref{Subsection::1Roof}, the addition of $(\npr-2t)$ noisy PR boxes to the local effect polytope introduces $58(\npr-2t)$ new extreme effects and eliminates $2(\npr-2t)$ CH type effects. On the other hand, from Section~\ref{Subsection::2Roofs}, $t$ pairs of isotropically opposite PR boxes introduce to the local effect polytope $56t$ new extreme effects and eliminate $2t$ CH type effects. Putting all these together one gets $90+56\npr-58t$ extreme effects of the effect polytope. A similar analysis can be done when $\alpha = 1$ and leads us to the total number of extreme effects of the effect polytope as follows:
\begin{equation}
     \Big|\mathrm{Extreme}\left[\mathcal{E}_{\mathbb{H}^{[\npr]}_{\alpha (2,2)}}\right]\Big| = \begin{cases}
        90         & \text{if } \alpha = 1/2\\
        90+56\npr-58t & \text{if } 1/2 < \alpha < 1\\
        90+16\npr-34t & \text{if } \alpha = 1.
    \end{cases}
 \end{equation}
The techniques in this section can be extended to cases where the PR boxes have different amounts of noise on them, but we do not do so here for simplicity.

\section{Bipartite Compositions of \texorpdfstring{$\mathcal{G}_{3}^2$}{} and their Effect Polytopes}
\label{Section::EffectPolytope3Fid}

In this section we study bipartite state spaces of $\mathcal{G}_3^2$, the gbit system with three fiducial measurements, and investigate the presence of couplers in them. Because the smallest quantum system is a qubit, which needs three fiducial measurements to be characterised, this is arguably closer to the quantum case than $\mathcal{G}_2^2$. We have looked at state spaces constructed from the convex hull of 64 local deterministic states and 1 extremal entangled state, denoted $\mathbb{H}^{[1]}_{(3,2)}$, Since there are 4 classes of extremal entangled states that allow for this construction, we consider them separately. In Table~\ref{Table::1Roof3Fid} we summarise the different facets of each state space and the number of extreme effects of their respective effect polytopes and we discuss the existence of couplers for such state spaces in Section~\ref{Subsection::Couplers3Fid}. 

\begin{table}
    \centering
    \begin{tabular}{|c|c|c|c|}
    \hline
        $\phantom{\Big(} \rm Class \phantom{\Big(}$ & $\#\rm\ I_{CH}\ Facets $ & $\#\ \rm I_{3322}\ Facets $ & $\#{\rm Extreme\  Effects}$\\
        \hline
        $\phantom{\Big(} \mathbb{H}^{[1]}_{(3,2)}[{\rm N_1}] \phantom{\Big(}$ & 71 & 568 & 29486  \\
         \hline
        $\phantom{\Big(} \mathbb{H}^{[1]}_{(3,2)}[{\rm N_2}] \phantom{\Big(}$ & 66 & 558 & 41888  \\
         \hline
        $\phantom{\Big(} \mathbb{H}^{[1]}_{(3,2)}[{\rm N_3}] \phantom{\Big(} $ & 68 & 552 & 37376 \\
         \hline
        $\phantom{\Big(} \mathbb{H}^{[1]}_{(3,2)}[{\rm N_4}] \phantom{\Big(}$ & 70 & 564 & 32384 \\
         \hline        
    \end{tabular}
    \caption{Summary of the number of CH facets, ${\rm I_{3322}}$ facets and extreme effects for the state space $\mathbb{H}^{[1]}_{(3,2)}$. $\rm N_1,N_2,N_3$ and $\rm N_4$ are as defined in Subsection~\ref{Subsection::ExampleBW}.}
    \label{Table::1Roof3Fid}
\end{table}

\section{Minimal \texorpdfstring{$k$}{k}-Preservability Criterion}
\label{Sec::MinimalPreservability}
So far, the only consistency condition we have put on state and effect spaces is that they result in valid probabilities, i.e., the inner product between any state and any effect should be between 0 and 1. In order to build a full theory, we also need to consider composability of systems. Expanding on ideas from~\cite{PhysRevLett.102.110402,PhysRevA.73.012101}, effects corresponding to a given number of systems should also be compatible with states of larger systems, in the sense that for two state spaces $\mathcal{S}^{\boxtimes m}$ and $\mathcal{S}^{\boxtimes k}$ and effect space $\mathcal{E}^{\boxtimes m} \subseteq \mathcal{E}_{\mathcal{S}^{\boxtimes m}}$, every effect $\Tilde{e} \in \mathcal{E}^{\boxtimes m}$ must have the property that
\begin{equation}
   \id^{\otimes k_1} \otimes \Tilde{e} \otimes \id^{\otimes k_2} \left(\varsigma \right) \in \mathcal{S}^{\boxtimes k}_{\leq},
   \label{Def::KPreserving}
\end{equation}
for all $\varsigma \in \mathcal{S}^{\boxtimes(m+k)}$ and $k_1,k_2\in\mathbb{N}_0$ with $k_1+k_2=k$. In other words, $\Tilde{e}$ acts on $m$ of the $m+k$ systems in state $\varsigma$, and the requirement ensures that $\Tilde{e}$ respects the composition $\boxtimes$ by preserving the state space structure in the presence of $k$ extra subsystems\footnote{Because the property must hold for all states, we can fix the positions of the identities in~\eqref{Def::KPreserving}.}. However, in this work we propose the inclusion of an additional condition for an effect $\tilde{e}$ to be valid, namely that there exists a measurement with $\tilde{e}$ as one of its effects, and in which all the effects in the measurement satisfy~\eqref{Def::KPreserving}. If this holds we say that $\tilde{e}$ is a $k$-\textit{preserving} effect. Further, if $\tilde{e}$ is $k$-\textit{preserving} for all $k \geq 0$, we say that $\tilde{e}$ is a \textit{completely preserving} effect. Such a concept is not needed in quantum theory where all elements of the maximal effect space are POVM elements (which are completely positive).

For an arbitrary state space $\mathcal S^{\boxtimes m} \in \mathbb V^m$, in general there is no limit to the number of composition rules $\boxtimes$ that identify composite state spaces of $m+k$ systems in $\mathbb V^m \otimes \mathbb V^k$, such that appropriate marginalisation gives a state in $\mathcal{S}^{\boxtimes m}$. In addition, depending on how it is specified, given a specific rule $\boxtimes$, completely characterising every state in $\mathcal{S}^{\boxtimes (m+k)}$ may not be straightforward, and hence it may be difficult to show that an effect $\tilde{e}$ is $k$-preserving. [The question of whether there is a simple set of sufficient conditions to test $k$-preservability (or even complete preservability) remains open as far as we are aware.]

In this work, we focus on a weaker, necessary condition for $\tilde{e}$ to be $k$-preserving that can be checked when all states in $\mathcal{S}^{\boxtimes m}$ and  $\mathcal{S}^{\boxtimes k}$ are known, but without the need for a full list of states in $\mathcal{S}^{\boxtimes (m+k)}$. From Def.~\ref{def::MinAndMax}, for two state spaces $\mathcal{S}^{\boxtimes m}$ and $\mathcal{S}^{\boxtimes k}$, we always have $\mathcal{S}^{\boxtimes m} \ntens \mathcal{S}^{\boxtimes k} \subseteq\mathcal{S}^{\boxtimes (m+k)}$, and so for any two states $r\in\mathcal{S}^{\boxtimes m}$ and $s\in\mathcal{S}^{\boxtimes k}$, the product state $r \otimes s$ is an element of $\mathcal{S}^{\boxtimes (m+k)}$. Therefore if $e$ has to satisfy~\eqref{Def::KPreserving}, it must at least act consistently on any $m$ subsystems of $r \otimes s$. We call this \textit{weak minimal $k$-preservability}.  Similarly to the argument above, for $e$ to be a valid effect we require that it is part of a measurement in which every effect is weakly minimally $k$-preserving. We call this necessary condition for $k$-preservability \textit{minimal k-preservability} and formally define it below. 

\begin{definition}\textbf{(Minimal k-preservability)}
     Let $\mathcal{S}^{\boxtimes m}$ and $\mathcal{S}^{\boxtimes k}$ be $m$- and $k$-partite state spaces and  $r \otimes s$ be a state describing $m+k$ systems where $r \in \mathcal{S}^{\boxtimes m}$ describes a composite system with subsystems labelled $1,2,\ldots,m$ and $s \in \mathcal{S}^{\boxtimes k}$ describes the composite system with subsystems labelled $m+1,m+2,\ldots,m+k$. Let 
     $X_{r} \coloneqq \{1,2,\ldots,m\}$, $X_{s} \coloneqq \{m+1,m+2,\ldots,m+k\}$, $B_{r} \subseteq X_{r}$, $B_{s} \subseteq X_{s}$, $A_{r} = X_{r} \setminus B_{r}$ and $C_{s} = X_{s} \setminus B_{s}$.  An effect $e \in \mathcal{E}_{\mathcal{S}^{\boxtimes m}}$ is said to be \emph{weakly minimally $k$-preserving} if
     $$
    \id^{A_r} \otimes e^{B_{r} \cup B_{s}} \otimes \id^{C_{s}} \left(r \otimes s \right) \in \mathcal{S}^{\boxtimes k}_{\leq}
     $$
     for any $r \in \mathcal{S}^{\boxtimes m}$ and $s \in \mathcal{S}^{\boxtimes k}$ and any $B_{r},B_{s}$ such that $|B_{r}|+|B_{s}|=m$, where $e^{B_{r} \cup B_{s}}$ denotes the action of $e$ on the $m$ subsystems in $B_{r} \cup B_{s}$. 

     If, additionally, there exists a measurement $\{e_i\}_i$ such that $e_i$ is weakly minimally $k$-preserving for all $i$, then $e_i$ is \emph{minimally $k$-preserving}.
     \label{eq:MinPreserving}
\end{definition}
It turns out to be sufficient to consider the measurement $\{e,u-e\}$ to confirm whether or not $e$ is minimally $k$-preserving, as shown in the following lemma.
\begin{lemma} \label{lemma:complementary}
    Let $e\in\mathcal{E}_{\mathcal{S}^{\boxtimes m}}$ be a weakly minimally $k$-preserving effect. Then $e$ is minimally $k$-preserving if and only if $u-e$ is weakly minimally $k$-preserving.
\end{lemma}
\begin{proof}
If $u-e$ is weakly minimally $k$-preserving, then $\{e,u-e\}$ is a measurement with weakly minimally $k$-preserving effects, and hence $e$ (and $u-e$) is minimally $k$-preserving.

For the reverse direction, let $\{e_i\}_i$ be a set of weakly minimally $k$-preserving effects such that $\sum_i e_i = u-e$. Then, for any pair of states, $r\in\mathcal{S}^{\boxtimes m}$ and $s\in\mathcal{S}^{\boxtimes k}$,
$$
\id^{A_r} \otimes e_i^{B_r \cup B_s} \otimes \id^{C_s} (r \otimes s)
$$
is a valid sub-normalised state of the state space, i.e., is in $\mathcal{S}^{\boxtimes k}_{\leq}$, where $A_r$, $B_r$, $B_s$ and $C_s$ are any sets satisfying the conditions given in Definition~\ref{eq:MinPreserving}. For brevity, in the following we omit these superscripts. Next, let 
$p_{u-e} \coloneqq \langle u, \id \otimes (u-e)\otimes \id (r \otimes s) \rangle $ and note that $p_{u-e} \in [0,1]$. First, if $p_{u-e}\neq0$, then
$$
\frac{1}{p_{u-e}}\left ( \id \otimes (u-e) \otimes \id (r \otimes s) \right)=\frac{1}{p_{u-e}} \left[ \sum_i \left ( \id \otimes e_i \otimes \id (r \otimes s) \right) \right]\in\mathcal{S}^{\boxtimes k},
$$
where the latter inclusion follows because $\{e_i\}$ are weakly minimally $k$-preserving and because the state on the left is normalised. Second, if $p_{u-e} = 0$, then $u-e$ maps the pair of states $r$ and $s$ to the zero state, which is an element of $\mathcal{S}^{\boxtimes k}_{\leq}$.

The conditions for $u-e$ to be weakly minimally $k$-preserving are hence satisfied.
\end{proof}

The framework of GPTs can be used to describe theories with different types of systems. For instance, in a theory with two system types such that the first requires two fiducial measurements for state tomography but the second requires three, then bipartite state spaces consisting of states formed by the composition of these two types of systems are needed. In other words, we need a composition rule for each pair of system types to describe general bipartite systems. If a four-partite extension of this is considered, then additional composition rules are needed to allow for all the possible combinations. An example of this has been considered in~\cite{barnum2008teleportation} where the authors used both min- and max-tensor products to describe states with four subsystems.

The fact that $B_r$ and $B_s$ are arbitrary subsets of $X_r$ and $X_s$, means that we are considering state spaces that are symmetric under any permutation of the $m+k$ systems. However, the state space need not be symmetric for all GPTs. Symmetry under the exchange of parties can be absent, for instance, because we consider different system types, but also in cases where the local state space of each sub-system is the same. An example of this is $\mathbb{H}^{[1]}_{\alpha (2,2)}[\pr_3]$ with effect space $\mathcal{E}_{\mathbb{H}^{[1]}_{\alpha (2,2)}[\pr_3]}$. Comparing this to $\mathbb{H}^{[1]}_{\alpha (2,2)}[\pr_1]$ with effect space $\mathcal{E}_{\mathbb{H}^{[1]}_{\alpha (2,2)}[\pr_1]}$, we see that since the state $\pr_1$ is party symmetric, i.e., the probability table obtained after relabelling parties is also $\pr_1$, for every bipartite effect in $\mathcal{E}_{\mathbb{H}^{[1]}_{\alpha (2,2)}[\pr_1]}$ applying the party swap operation ($T_{\mathrm{SWAP}}$) to that effect gives a vector whose inner product with $\pr_{1,\alpha}$ is between 0 and 1. However, the analogous property does not hold for $\mathbb{H}^{[1]}_{\alpha (2,2)}[\pr_3]$ and $\mathcal{E}_{\mathbb{H}^{[1]}_{\alpha (2,2)}[\pr_3]}$. In particular, for all $\alpha\in(1/2,1]$,
$$\left<T_{\mathrm{SWAP}}\left[e_{\rm CH_4'}\right],\pr_{3,\alpha}\right> \notin [0,1].$$
Because $\mathbb{H}^{[1]}_{\alpha (2,2)}[\pr_3]$ and $\mathcal{E}_{\mathbb{H}^{[1]}_{\alpha (2,2)}[\pr_3]}$ can be mapped to $\mathbb{H}^{[1]}_{\alpha (2,2)}[\pr_1]$ and $\mathcal{E}_{\mathbb{H}^{[1]}_{\alpha (2,2)}[\pr_1]}$ by local relabellings, this particular example is equivalent to a party symmetric one.

Consider a party symmetric GPT $\left(\mathcal{S}, \mathcal{E}, \boxtimes \right)$.  Requiring minimal $2$-preservability puts constraints on the effect space, in the sense that not all elements of $\mathcal{E}_{\mathcal{S}^{\boxtimes 2}}$ correspond to valid effects. Let $e \in \mathcal{E}_{\mathcal{S}^{\boxtimes 2}}$ be an effect and $r$ and $s$ be two arbitrary states of $\mathcal{S}^{\boxtimes 2}$. Further, suppose that $r$ is composed of 2 subsystems labelled by $\{1,2\}$ and $s$ is composed of 2 subsystems labelled by $\{3,4\}$. Since $e$ is a bipartite effect, and since we impose weak minimal $2$-preservability, for a party symmetric state space, it is sufficient to consider two different maps arising from $e$ depending on which subsystems of the state $r\otimes s$ it acts on:
$$
e^{(1,2)} \otimes \id^{(3,4)} \left(r\otimes s \right) \quad \text{and}\quad \id^{(1)} \otimes e^{(2,3)} \otimes \id^{(4)}  \left(r\otimes s \right)=:\Phi_{e}^{(23)}\left(r,s\right).
$$
For $e$ to be weakly minimally 2-preserving, we require that these two vectors represent subnormalised states (i.e., are elements of $\mathcal{S}^{\boxtimes 2}_{\leq}$) for any choice of $r$ and $s$. From the definition of an effect space (see Def.~\ref{Def::GPT}), the first state is always an element of $\mathcal{S}^{\boxtimes 2}_{\leq}$, but the second is not necessarily. Weak minimal $2$-preservability in this scenario hence corresponds to $\Phi_{e}^{(23)}\left(r,s\right)$ being an element of $\mathcal{S}^{\boxtimes 2}_{\leq}$ for all $r,s\in\mathcal{S}^{\boxtimes 2}$. This discussion can be extended to any scenario in which a composite state space is provided but a way to extend it to larger systems is not.

By applying the condition for every pair of extremal states, we find that the couplers from $\mathbb{H}^{[1]}_{\alpha (2,2)}[\pr_2]$ presented in Subsection~\ref{Subsection::1Roof} (and from $\mathbb{H}^{[1]}_{\alpha=1 (2,2)}[\pr_1]$ presented in~\cite{PhysRevLett.102.110402}) are weakly minimally $2$-preserving. However, some of the effects we found are not---see the next two subsections for examples. For all models considered in this work in the $(2,2)$ case, effects that are weakly minimally $2$-preserving are also minimally $2$-preserving. This fails to hold in the $(3,2)$ case, as shown in Example~\ref{Subsec::Min2PresMeas}, illuminating the significance of this criterion. Note that notions similar to \emph{weak} minimal 2-preservibility have appeared in the literature~\cite{PhysRevLett.102.110402, PhysRevA.73.012101, Selby2023correlations,PhysRevLett.119.020401,dallarno2023signaling}, for a comparison to~\cite{PhysRevLett.119.020401,dallarno2023signaling}, we refer to Section~\ref{subsection::RelatedMin2Pres} below. 

\subsection{CH Type Effects of \texorpdfstring{$\mathbb{H}^{[1]}_{\alpha (2,2)}[\pr_2]$}{}}\label{sec:CH1R}
As introduced in Section~\ref{Subsection::1Roof}, the CH-type effects of the effect polytope $\mathcal{E}_{\mathbb{H}^{[1]}_{\alpha (2,2)}}$ for $\alpha \in (1/2,1]$ are $e_{\mathrm{CH_1}},$ $e_{\mathrm{CH_1}}',$ $e_{\mathrm{CH_3}},$ $e_{\mathrm{CH_3}}',$ $e_{\mathrm{CH_4}}$ and $e_{\mathrm{CH_4}}'$, where $e_{\rm CH_i}' = u-e_{\rm CH_i}$. However, these effects are not weakly minimally $2$-preserving when $\alpha > 1/\sqrt{2}$. A direct calculation shows that 
\begin{equation}
    \begin{split}
        {\rm CHSH_{1'}} \left[  \frac{\Phi^{(23)} _{e_{\rm CH_1}} \left(\pr_{2,\alpha},\pr_{2,\alpha}  \right)}{\left<u,\Phi^{(23)} _{e_{\rm CH_1}} \left(\pr_{2,\alpha},\pr_{2,\alpha}  \right)\right>}  \right] = \frac{1}{2}(\alpha^2+1),\quad  {\rm CHSH_{1}} \left[  \frac{\Phi^{(23)} _{e_{\rm CH_1'}} \left(\pr_{2,\alpha},\pr_{2,\alpha}  \right)}{\left<u,\Phi^{(23)} _{e_{\rm CH_1'}} \left(\pr_{2,\alpha},\pr_{2,\alpha}  \right)\right>}  \right] =\frac{1}{2}(\alpha^2+1) ,\\
        {\rm CHSH_{4'}} \left[  \frac{\Phi^{(23)} _{e_{\rm CH_3}} \left(\pr_{2,\alpha},\pr_{2,\alpha}  \right)}{\left<u,\Phi^{(23)} _{e_{\rm CH_3}} \left(\pr_{2,\alpha},\pr_{2,\alpha}  \right)\right>}  \right] = \frac{1}{2}(\alpha^2+1),\quad  {\rm CHSH_{4}} \left[  \frac{\Phi^{(23)} _{e_{\rm CH_3'}} \left(\pr_{2,\alpha},\pr_{2,\alpha}  \right)}{\left<u,\Phi^{(23)} _{e_{\rm CH_3'}} \left(\pr_{2,\alpha},\pr_{2,\alpha}  \right)\right>}  \right] =\frac{1}{2}(\alpha^2+1) ,\\
        {\rm CHSH_{3'}} \left[  \frac{\Phi^{(23)} _{e_{\rm CH_4}} \left(\pr_{2,\alpha},\pr_{2,\alpha}  \right)}{\left<u,\Phi^{(23)} _{e_{\rm CH_4}} \left(\pr_{2,\alpha},\pr_{2,\alpha}  \right)\right>}  \right] = \frac{1}{2}(\alpha^2+1),\quad  {\rm CHSH_{3}} \left[  \frac{\Phi^{(23)} _{e_{\rm CH_4'}} \left(\pr_{2,\alpha},\pr_{2,\alpha}  \right)}{\left<u,\Phi^{(23)} _{e_{\rm CH_4'}} \left(\pr_{2,\alpha},\pr_{2,\alpha}  \right)\right>}  \right] = \frac{1}{2}(\alpha^2+1).
    \end{split}
\end{equation}
Since $(\alpha^2+1)/2 > 3/4$ when $\alpha > 1/\sqrt{2}$, no CH-type effect is weakly minimally $2$-preserving for $\alpha > 1/\sqrt{2}$ since in $\mathbb{H}^{[1]}_{\alpha(2,2)}[\pr_2]$ only $\rm CHSH_2[s]\leq3/4$ can be violated. An alternative way to check whether the state on systems 1 and 4 is an element of the state space, is to verify whether the list of inner products it generates with all the extreme effects are in the interval $[0,1]$. For the current example it is sufficient to only compute the inner products for the CH-type effects since these are the non-trivial facets of the state space polytope. 

\subsection{Couplers of \texorpdfstring{$\mathbb{H}^{[2]}_{\alpha (2,2)}[\pr_{1,2}]$}{}}
\label{Example::PR1PR2}

In Section~\ref{Subsection::2Roofs} we found that the state space $\mathbb{H}^{[2]}_{\alpha (2,2)}[\pr_{1,2}]$ contains couplers. These couplers turn out not to be weakly minimally $2$-preserving as the following calculations show. 
\begin{equation}
\begin{split}
    &\mathrm{CHSH}_{2'} \left[ \frac{\Phi^{(23)}_{f_2(\alpha)}\left(\pr_{1,\alpha},\pr_{2,\alpha}\right)}{\left<u,\Phi^{(23)}_{f_2(\alpha)}\left(\pr_{1,\alpha},\pr_{2,\alpha}\right)\right>}\right] = \frac{1}{2} \left(\alpha ^2+1\right), \ \  \mathrm{CHSH}_{2'} \left[ \frac{\Phi^{(23)}_{\Tilde{e}(\alpha)}\left(\pr_{1,\alpha},\pr_{2,\alpha}\right)}{\left<u,\Phi^{(23)}_{\Tilde{e}(\alpha)}\left(\pr_{1,\alpha},\pr_{2,\alpha}\right)\right>}\right] =  \frac{5\alpha^2+2\alpha+4}{4(\alpha+2)},\\
    &\mathrm{CHSH}_{1'} \left[ \frac{\Phi^{(23)}_{f_1(\alpha)}\left(\pr_{2,\alpha},\pr_{2,\alpha}\right)}{\left<u,\Phi^{(23)}_{f_1(\alpha)}\left(\pr_{2,\alpha},\pr_{2,\alpha}\right)\right>}\right] = \frac{1}{2} \left(\alpha ^2+1\right),\ \ \mathrm{CHSH}_{1'} \left[ \frac{\Phi^{(23)}_{\Tilde{g}(\alpha)}\left(\pr_{2,\alpha},\pr_{2,\alpha}\right)}{\left<u,\Phi^{(23)}_{\Tilde{g}(\alpha)}\left(\pr_{2,\alpha},\pr_{2,\alpha}\right)\right>}\right] =  \frac{5\alpha^2+2\alpha+4}{4(\alpha+2)},
\end{split}  \label{Equation::CHSHScores2Roofs}
\end{equation}
where $f_1(\alpha)$ and $\Tilde{g}(\alpha)$ are the Type 4 effect for $e_{\mathrm{CH}_1}$ and any Type 3 effect lying on the facet $\left<\mathbf{x},\pr_{1,\alpha}\right>=1$ respectively; $f_2(\alpha)$ and $\Tilde{e}(\alpha)$ are the Type 4 effect for $e_{\mathrm{CH}_2}$ and any Type 3 effect lying on the facet $\left<\mathbf{x},\pr_{2,\alpha}\right>=1$, respectively. The state space $\mathbb{H}^{[2]}_{\alpha (2,2)}[\pr_{1,2}]$ has the property that $\forall s \in \mathbb{H}^{[2]}_{\alpha (2,2)}[\pr_{1,2}]$  $\rm CHSH_{1'}[s]\leq3/4$ and $\rm CHSH_{2'}[s]\leq3/4$. Hence, for $\alpha>1/\sqrt{2}$ the Type 4 effects cease to be weakly minimally $2$-preserving and for $\alpha>(1+\sqrt{41})/10$ the Type 3 effects cease to be weakly minimally $2$-preserving. Note that these ranges of $\alpha$ coincide with the values of $\alpha$ for which these effects are couplers (cf.\ Section~\ref{Subsection::1Roof}). Thus there are no weakly minimally $2$-preserving extreme effects for this state space that are couplers. By a similar argument to that in Section~\ref{sec:CH1R}, the CH-type effects are also not weakly minimally $2$-preserving, for $\alpha>1/\sqrt{2}$ [this is because the maximum score in a $\rm CHSH_{i \neq 1,2}$ game on systems $1$ and $4$ after a CH effect is applied to systems $2$ and $3$ for the tensor product of two allowed PR boxes is $(\alpha^2+1)/2$].

\section{Minimally \texorpdfstring{$2$}{2}-Preserving Couplers of Party Symmetric State Spaces}
\label{Section::Min2PresCouplers}
\subsection{Party Symmetric State Spaces with Restricted Relabelling}
\label{Subsection::StateSpacesConsidered}

In this work we will focus on bipartite
compositions of $\mathcal{G}_{2}^2$ and $\mathcal{G}_{3}^2$ that are party symmetric. We divide these into cases based on the number of maximally non-local states present in the state space. For compositions of $\mathcal{G}_{2}^2$, we will consider this varying from one to all eight noisy PR boxes.   

For a given number of maximally non-local states there are in general equivalence classes of the various state spaces under the relabelling symmetries. We say that two state spaces $\mathbb{H}^{[\npr]}_{\alpha (2,2)}$ and $\mathbb{H}'^{[\npr]}_{\alpha (2,2)}$ are equivalent if there exists a local relabelling $R$ that maps $\mathbb{H}^{[\npr]}_{\alpha (2,2)}$ to $\mathbb{H}'^{[\npr]}_{\alpha (2,2)}$, i.e., for every state $s' \in \mathbb{H}'^{[\npr]}_{\alpha (2,2)}$, there exists a state $s$ such that $R[s]=s'$. In the table below we state the number of classes of party symmetric state spaces for $2 \leq \npr \leq 7$ (see Appendix~\ref{Appendix::EquivClasses} for a full classification).

\begin{table}[H]
    \centering
    \begin{tabular}{|c|c|c|c|c|c|c|}
\hline 
 $\phantom{\Big(}$ $\npr$ $\phantom{\Big(}$ & $\phantom{\Big(} 2 \phantom{\Big(}$   & $ \phantom{\Big(}3 \phantom{\Big(}$  & $\phantom{\Big(} 4 \phantom{\Big(}$ & $\phantom{\Big(} 5 \phantom{\Big(}$   & $\phantom{\Big(} 6 \phantom{\Big(}$  & $\phantom{\Big(} 7 \phantom{\Big(}$ \\
   \hline 
  $\phantom{\Big(}$ \# Classes $\phantom{\Big(}$ & $\phantom{\Big(} 2 \phantom{\Big(}$  & $\phantom{\Big(} 2 \phantom{\Big(}$  & $\phantom{\Big(} 4\phantom{\Big(} $  & $\phantom{\Big(} 2 \phantom{\Big(}$  & $\phantom{\Big(} 2\phantom{\Big(} $  & $\phantom{\Big(} 1\phantom{\Big(} $ \\
    \hline      
\end{tabular}
    \caption{The number of equivalence classes for party symmetric states spaces with $\npr$ PR-boxes in the 2 input, 2 output case.}
    \label{Tab::PartySwapCount}
\end{table}

When single systems are described by $\mathcal{G}_3^2$, party symmetric state spaces with one maximally non-local state are characterised by the convex hull of 64 local deterministic states and one of the entangled states $\rm N_1,$ $\rm N_2$ or $\rm N_3$ from Section~\ref{Subsection::ExampleBW} (note that no relabelling of $\rm N_4$ is symmetric under party swap so we do not consider this case).  We will restrict to only one non-local state since the number of classes grows significantly with more. Furthermore, we have only considered compositions of state spaces without noise for this case, i.e., state spaces of the form $\mathbb{H}^{[1]}_{\alpha=1(3,2)}[\rm N_i]$. 

In the following, we calculate effect polytopes for $\mathbb{H}^{[\npr]}_{\alpha (2,2)}$ and $\mathbb{H}^{[1]}_{(3,2)}$ and search for effects that are minimally $2$-preserving couplers.

\subsection{Bipartite systems \texorpdfstring{$\mathbb{H}^{[\npr]}_{\alpha (2,2)}$}{} }

In Section~\ref{Sec::MinimalPreservability} we found that the state space $\mathbb{H}^{[1]}_{\alpha=1(2,2)}[\pr_2]$ allows couplers that are minimally $2$-preserving~\cite{PhysRevLett.102.110402} and that there are no such couplers for the state space $\mathbb{H}^{[2]}_{\alpha}[\pr_{1,2}]$. Using the ideas from the previous sections, we find that none of the extremal effects of $\mathbb{H}^{[\npr]}_{\alpha (2,2)}$ with $2\leq \npr \leq 8$  are minimally $2$-preserving couplers. However, this does not rule out that there might be non-extreme minimally $2$-preserving effects that are couplers, i.e., whether convex mixtures of extreme effects can be both minimally $2$-preserving and couplers. Our next theorem says that these also do not exist.

\begin{theorem}\label{Theorem::NoCouplers}
    Let $\Tilde{e} \in \mathcal{E}_{\mathbb{H}^{[\npr]}_{\alpha (2,2)}}$ be a candidate effect for a party symmetric bipartite state space $\mathbb{H}^{[\npr]}_{\alpha (2,2)}$ where $2\leq \npr \leq 8$. Then $\Tilde{e}$ cannot both (i) be a coupler and (ii) be minimally $2$-preserving.
\end{theorem}
\noindent The proof of this theorem is given in Appendix~\ref{Appendix::NoCouplers}.

We have also considered bipartite state spaces that are not party symmetric. For these we have considered state spaces of the form $\mathbb{H}^{[\npr]}_{\alpha (2,2)}$, for a discrete set of values of $\alpha$ between 1/2 and 1 with a step size of 1/30 and $2\leq \npr\leq 8$. In none of these cases were minimally $2$-preserving couplers found, leading us to conjecture that of the state spaces of the form $\mathbb{H}^{[\npr]}_{\alpha (2,2)}$, minimally $2$-preserving couplers are only present for $\npr=1$.  

\subsection{Bipartite systems \texorpdfstring{$\mathbb{H}^{[1]}_{(3,2)}$}{}}

In this subsection we first present an instance of an effect that is weakly but not minimally $2$-preserving, a feature that was absent in the $(2,2)$ case so far. Then we present a classification of the extremal effects of party swap symmetric states spaces of the form $\mathbb{H}^{[1]}_{(3,2)}$ based on compositional consistency and the ability to couple.

\subsubsection{An effect that is weakly minimally \texorpdfstring{$2$}{2}-preserving but not minimally \texorpdfstring{$2$}{2}-preserving}
\label{Subsec::Min2PresMeas}
Consider the state space $\mathbb{H}^{[1]}_{(3,2)}[\rm N_2]$ and take the extremal effect
$$ e \coloneqq 
\frac{2}{3}\left(
\begin{array}{cc|cc|cc}
 -1 & 0 & 0 & 0 & 0 & 0 \\
 0 & 0 & 0 & -1 & 0 & 0\\ \hline
 0 & 0 & 0 & 0 & 0 & 0\\
 0 & -1 & 0 & 1 & 0 & 0\\ \hline
 0 & 0 & 0 & 0 & 3/2 & 3/2\\
 0 & 0 & 0 & 0 & 3/2 & 3/2\\ 
\end{array}
\right).
$$
With this, one has $\Phi^{(23)}_{e}(\mathrm{ N_2,N_2})/\left<u,\Phi^{(23)}_{e}(\mathrm{ N_2,N_2})\right>=$
\begin{align*}
    \begin{split}
        &\frac{1}{8}\left(
\begin{array}{cc|cc|cc}
 0 & 1 & 1 & 0 & 1 & 0 \\
 0 & 0 & 0 & 0 & 0 & 0 \\ \hline
 0 & 1 & 1 & 0 & 1 & 0 \\
 0 & 0 & 0 & 0 & 0 & 0 \\ \hline
 0 & 0 & 0 & 0 & 0 & 0 \\ 
 0 & 1 & 1 & 0 & 1 & 0 \\
\end{array}
\right)+\frac{1}{8}\left(
\begin{array}{cc|cc|cc}
 0 & 1 & 0 & 1 & 1 & 0 \\
 0 & 0 & 0 & 0 & 0 & 0 \\ \hline
 0 & 1 & 0 & 1 & 1 & 0 \\
 0 & 0 & 0 & 0 & 0 & 0 \\ \hline
 0 & 0 & 0 & 0 & 0 & 0 \\ 
 0 & 1 & 0 & 1 & 1 & 0 \\
\end{array}
\right)+\frac{1}{8}\left(
\begin{array}{cc|cc|cc}
 1 & 0 & 0 & 1 & 0 & 1 \\
 0 & 0 & 0 & 0 & 0 & 0 \\ \hline
 0 & 0 & 0 & 0 & 0 & 0 \\
 1 & 0 & 0 & 1 & 0 & 1 \\ \hline
 0 & 0 & 0 & 0 & 0 & 0 \\
 1 & 0 & 0 & 1 & 0 & 1 \\
\end{array}
\right)+\frac{1}{8}\left(
\begin{array}{cc|cc|cc}
 0 & 1 & 0 & 1 & 1 & 0 \\
 0 & 0 & 0 & 0 & 0 & 0 \\ \hline
 0 & 0 & 0 & 0 & 0 & 0 \\
 0 & 1 & 0 & 1 & 1 & 0 \\ \hline
 0 & 0 & 0 & 0 & 0 & 0 \\
 0 & 1 & 0 & 1 & 1 & 0 \\
\end{array}
\right)\\
&\frac{1}{8}\left(
\begin{array}{cc|cc|cc}
 0 & 0 & 0 & 0 & 0 & 0 \\
 1 & 0 & 1 & 0 & 1 & 0 \\ \hline
 1 & 0 & 1 & 0 & 1 & 0 \\
 0 & 0 & 0 & 0 & 0 & 0 \\ \hline
 1 & 0 & 1 & 0 & 1 & 0 \\
 0 & 0 & 0 & 0 & 0 & 0 \\
\end{array}
\right)+\frac{1}{8}\left(
\begin{array}{cc|cc|cc}
 0 & 0 & 0 & 0 & 0 & 0 \\
 0 & 1 & 1 & 0 & 0 & 1 \\  \hline
 0 & 1 & 1 & 0 & 0 & 1 \\
 0 & 0 & 0 & 0 & 0 & 0 \\ \hline
 0 & 1 & 1 & 0 & 0 & 1 \\
 0 & 0 & 0 & 0 & 0 & 0 \\
\end{array}
\right)+\frac{1}{8}\left(
\begin{array}{cc|cc|cc}
 0 & 0 & 0 & 0 & 0 & 0 \\
 1 & 0 & 1 & 0 & 0 & 1 \\ \hline
 0 & 0 & 0 & 0 & 0 & 0 \\
 1 & 0 & 1 & 0 & 0 & 1 \\ \hline
 1 & 0 & 1 & 0 & 0 & 1 \\
 0 & 0 & 0 & 0 & 0 & 0 \\
\end{array}
\right)+\frac{1}{8}\left(
\begin{array}{cc|cc|cc}
 0 & 0 & 0 & 0 & 0 & 0 \\
 1 & 0 & 0 & 1 & 0 & 1 \\ \hline
 0 & 0 & 0 & 0 & 0 & 0 \\
 1 & 0 & 0 & 1 & 0 & 1 \\ \hline
 1 & 0 & 0 & 1 & 0 & 1 \\
 0 & 0 & 0 & 0 & 0 & 0 \\
\end{array}
\right),\\
    \end{split}
\end{align*}
which is local. Now consider the complementary effect 
$$
u-e= \frac{2}{3}\left(
\begin{array}{cc|cc|cc}
 1 & 0 & 0 & 0 & 0 & 0 \\
 0 & 0 & 0 & 1 & 0 & 0\\ \hline
 0 & 0 & 0 & 0 & 0 & 0\\
 0 & 1 & 0 & -1 & 0 & 0\\ \hline
 0 & 0 & 0 & 0 & 0 & 0\\
 0 & 0 & 0 & 0 & 0 & 0\\ 
\end{array}
\right).
$$
This is ray-extremal in $\mathcal{E}_{\mathbb{H}^{[1]}_{(3,2)}[\rm N_2]}$ (the output of PANDA identifies the ray-extremal effects directly). This effect is not weakly minimally $2$-preserving since
$$
\frac{\id \otimes (u-e) \otimes \id \left(\mathrm{N_2} \otimes \mathrm{N_2}\right)}{\left<u,\id \otimes (u-e) \otimes \id \left(\mathrm{N_2} \otimes \mathrm{N_2}\right)\right>} = \frac{1}{2} \left(
\begin{array}{cc|cc|cc}
 1 & 0 & 1 & 0 & 0 & 1 \\
 0 & 1 & 0 & 1 & 1 & 0 \\ \hline
 1 & 0 & 0 & 1 & 0 & 1 \\
 0 & 1 & 1 & 0 & 1 & 0 \\ \hline
 0 & 1 & 0 & 1 & 1 & 0 \\
 1 & 0 & 1 & 0 & 0 & 1 \\
\end{array}
\right) \notin \mathbb{H}^{[1]}_{(3,2)}[\rm N_2].
$$
It hence follows from Lemma~\ref{lemma:complementary} that $e$ is not minimally $2$-preserving and hence it is not a valid effect of $\mathbb{H}^{[1]}_{(3,2)}[\rm N_2]$.

\subsubsection{Classifications of Extremal Effects}
\label{Subsection::Couplers3Fid}

For the state spaces $\mathbb{H}^{[1]}_{(3,2)}[\mathrm{N_1}]$, $\mathbb{H}^{[1]}_{(3,2)}[\mathrm{N_2}]$ and $\mathbb{H}^{[1]}_{(3,2)}[\mathrm{N_3}]$, we found the maximal set of extremal effects using PANDA~\cite{Lorwald2015}. For each case, we then computed the subset of these effects that are weakly minimally $2$-preserving. In order to check how many of these effects are minimally $2$-preserving we consider the complementary effects of each (see Lemma~\ref{lemma:complementary}). 

For the state space $\mathbb{H}^{[1]}_{(3,2)}[\mathrm{N_1}]$ there are 768 effects that are not minimally $2$-preserving (but are weakly minimally $2$-preserving), and analogous counts for $\mathbb{H}^{[1]}_{(3,2)}[\mathrm{N_2}]$ and $\mathbb{H}^{[1]}_{(3,2)}[\mathrm{N_3}]$ are given in Table~\ref{Tab:WeakPresClassification}.

\begin{table}
    \centering
    \begin{tabular}{|c|c|c|c|}
    \hline
        state space $\mathbb{H}$& \# Extremal in $\mathcal{E}_{\mathbb{H}}$& \# w-min $2$-pres (I)  & \shortstack{\# I$'$}  \\ \hline \hline
       $\mathbb{H}^{[1]}_{(3,2)}[\mathrm{N_1}]$ & 29486 & 28689 & 768 \\ \hline
       $\mathbb{H}^{[1]}_{(3,2)}[\mathrm{N_2}]$ & 41888 & 19222 & 9030  \\ \hline
       $\mathbb{H}^{[1]}_{(3,2)}[\mathrm{N_3}]$ & 37376 & 35504 & 1536  \\ \hline 
    \end{tabular}
    \caption{Classification of extremal effects based on weak minimal (w-min) $2$-preservability. $\#$ denotes the size of a set. I denotes the set of extremal effects that are weakly minimally $2$-preserving. I$'$ denotes the set of effects in I that are not minimally $2$-preserving.} 
     \label{Tab:WeakPresClassification}
\end{table}

We found the maximal set of extremal effects for the state spaces $\mathbb{H}^{[1]}_{(3,2)}[\mathrm{N_1}]$, $\mathbb{H}^{[1]}_{(3,2)}[\mathrm{N_2}]$ and $\mathbb{H}^{[1]}_{(3,2)}[\mathrm{N_3}]$ using PANDA~\cite{Lorwald2015}. From each of these maximal effect spaces, we removed the candidate extremal effects that are not weakly minimally $2$-preserving. We then classified the extremal effects present in the remaining sets into classes of effects equivalent up to relabelling. For each class, we calculated the maximal violations of $\langle\mathrm{F_{CH}},\mathbf{x}\rangle\leq1$ and $\langle\mathrm{F_{I_{3322}}},\mathbf{x}\rangle\leq1$ (cf.~\eqref{eq:FCH}), and the probabilities of observing the corresponding outcomes.

Considering $\mathbb{H}^{[1]}_{(3,2)}[{\mathrm{N_1}}]$, of the 28689 effects that are weakly minimally $2$-preserving, there are 856 couplers, which can be classified into 61 relabelling classes. 88 of these couplers are pure (in the sense that after Bob applies the coupler the resultant state on Alice and Charlie is extremal, i.e., $\mathrm{N_1}$) and fall into 15 of the 61 classes. We present a member of each of these 15 classes in Appendix~\ref{Appendix:PureCouplersN13322}. We found that the maximum product of the probability of successful swapping and the CHSH score of the post-selected state is $1/3$. One of the couplers that achieves this value is
$$f_{\mathrm{CH_1}} \coloneqq 
\frac{2}{3}\left(
\begin{array}{cc|cc|cc}
 0 & 0 & 1 & 0 & 0 & 0 \\
 0 & 1 & 0 & 0 & 0 & 0\\ \hline
 0 & -1 & 0 & 1 & 0 & 0\\
 0 & 0 & 0 & 0 & 0 & 0\\ \hline
 0 & 0 & 0 & 0 & 0 & 0\\
 0 & 0 & 0 & 0 & 0 & 0\\ 
\end{array}
\right).
$$
Notice that the top left $4\times 4$ block of $f_{\mathrm{CH_1}}$ is $e_{\mathrm{CH_1}}$. With the measurement $\{f_{\mathrm{CH_1}},u-f_{\mathrm{CH_1}}\}$, the probability of a successful swap is 1/3 and the CHSH value of the post-selected normalised state is 1 (corresponding to a $\mathrm{F_{CH}}$ value of $2/3$).

In the case of $\mathbb{H}^{[1]}_{(3,2)}[{\mathrm{N_2}}]$, the effect $f_{\mathrm{CH_1}}$, defined above, although extremal, turns out to be not weakly minimally $2$-preserving since the normalised state obtained after swapping is $\mathrm{N_3}$. In this case, there are 19222 weakly minimally $2$-preserving effects of which 390 are couplers. The couplers can be grouped into 15 classes. None of these couplers are pure. We calculated maximum product of the probability of successful swap and highest achievable CHSH score, similarly to the case above. This maximum corresponds to a successful swap probability of 13/24 and the CHSH score of the normalised post-selected state is 41/52. The coupler $u-f$ for
$$
f=\frac{1}{6}\left(
\begin{array}{cc|cc|cc}
 0 & 1 & 0 & 0 & 0 & 0 \\
 0 & 0 & 2 & 0 & 0 & 1 \\ \hline
 0 & 0 & 0 & 0 & 0 & 0 \\
 1 & 0 & 0 & 2 & 0 & 0 \\ \hline
 0 & 0 & 1 & 0 & 0 & 1 \\
 0 & 3 & 0 & 0 & 1 & 0 
\end{array}
\right)
$$
can achieve this.

For $\mathbb{H}^{[1]}_{(3,2)}[{\mathrm{N_3}}]$, 35504 of the extremal effects in the maximal effect space are weakly minimally $2$-preserving within which there are 2716 couplers that can be grouped into 78 classes. Only one of these classes contains the 4 pure couplers
$$
\frac{2}{3}\left(
\begin{array}{cc|cc|cc}
 0 & 0 & 0 & 0 & 0 & 0 \\
 0 & 0 & 0 & 0 & 0 & 0 \\ \hline
 0 & 0 & 0 & 0 & 0 & 0 \\
 0 & 0 & 0 & -1 & 0 & 1 \\ \hline
 0 & 0 & 0 & 0 & 1 & 0 \\
 0 & 0 & 0 & 1 & 0 & 0 \\
\end{array}
\right), \frac{2}{3}\left(
\begin{array}{cc|cc|cc}
 0 & 0 & 0 & 0 & 0 & 0 \\
 0 & 0 & 0 & 0 & 0 & 0 \\ \hline
 0 & 0 & 0 & 0 & 0 & 0 \\
 0 & 1 & 0 & -1 & 0 & 0 \\ \hline
 1 & 0 & 0 & 0 & 0 & 0 \\
 0 & 0 & 0 & 1 & 0 & 0 \\
\end{array}
\right), \frac{2}{3}\left(
\begin{array}{cc|cc|cc}
 0 & 0 & 0 & 0 & 1 & 0 \\
 0 & 0 & 0 & 1 & 0 & 0 \\ \hline
 0 & 0 & 0 & 0 & 0 & 0 \\
 0 & 0 & 0 & -1 & 0 & 1 \\ \hline
 0 & 0 & 0 & 0 & 0 & 0 \\
 0 & 0 & 0 & 0 & 0 & 0 \\
\end{array}
\right), \frac{2}{3}\left(
\begin{array}{cc|cc|cc}
 1 & 0 & 0 & 0 & 0 & 0 \\
 0 & 0 & 0 & 1 & 0 & 0 \\ \hline
 0 & 0 & 0 & 0 & 0 & 0 \\
 0 & 1 & 0 & -1 & 0 & 0 \\ \hline
 0 & 0 & 0 & 0 & 0 & 0 \\ 
 0 & 0 & 0 & 0 & 0 & 0 \\
\end{array}
\right).
$$

In this case, the maximum product of the probability of successful swapping and the inner product of the normalised state generated with successful swap is 53/144. The corresponding probability and CHSH values are 17/36 and 53/68 respectively. A coupler that achieves this is
$$
\frac{1}{9}\left(
\begin{array}{cc|cc|cc}
 0 & 0 & 0 & 1 & 0 & 0 \\
 2 & 0 & 0 & 0 & 0 & 1 \\ \hline
 0 & 0 & 0 & 0 & 0 & 0 \\
 0 & -1 & 0 & -2 & 0 & 5 \\ \hline
 0 & 1 & 0 & 0 & 7 & 0 \\
 0 & 0 & 0 & 6 & 1 & 0
\end{array}
\right),
$$
which is not a pure coupler.

A full list of the weakly minimally $2$-preserving extremal effects for each of the maximal effect spaces can be found in the Supplementary Material~\cite{sengupta_2025_17520323}. 

\section{Correlation Self-Testing of Quantum Theory against Party Swap Symmetric State Spaces}
\label{Section::CorrelationSelfTesting}
We proceed to show that quantum theory can be correlation self-tested against any party symmetric state spaces of the form $\mathbb{H}^{[\npr]}_{\alpha (2,2)}$ and $\mathbb{H}^{[1]}_{(3,2)}$. To do this we use the fact that no conditional distribution can violate more than one of the CHSH inequalities in the $[2,2]$ setting.

\begin{lemma}\label{Lemma::IntersectionEmpty}

Let $A,B,X,Y$ be four random variables with $|A|=|B|=|X|=|Y|=2$. Let $\{\chsh_i\}_{i=1}^8$ be the 8 CHSH games in the $[2,2]$ setting. If $\chsh_{i}\left[\p(A,B|X,Y)\right] > 3/4$, then $\chsh_{j}\left[\p(A,B|X,Y)\right] \leq 3/4$ for all $j \neq i$.
\end{lemma}
\noindent The proof is given in Appendix~\ref{Appendix::LemmaIntersection}.

Next, we prove an analytic upper bound on the maximum probability of winning the ACHSH game for any GPT characterised by the convex hull of local deterministic states and only one nonlocal state.

\begin{theorem}\label{Proposition::WinningBound}
    Let $\mathbb{H}^{[1]}_{\alpha,({\rm m,n})}[{\rm N}]$ be a bipartite state space characterised by the convex hull of $\{\mathrm{L}_1, \ldots , \mathrm{L}_l, {\rm N}_\alpha\}$, where $\{\mathrm{L}_i\}_{i=1}^l$ are local deterministic states and ${\rm N}_\alpha$ is the mixture of an extremal no-signalling state with noise (analogously to~\eqref{Eq::NoiseModel}). Let $E_{\rm coup} \subset {\rm Extreme}[\mathcal{E}_{\mathbb{H}^{[1]}_{\alpha,({\rm m,n})}[{\rm N}]}]$ be the set of minimally $2$-preserving extremal couplers. Further, for $e \in E_{\rm coup}$, let $p_{\rm succ}(e) = \left<u, \Phi_{e}^{(2,3)}(\mathrm{ N}_\alpha,\mathrm{N}_\alpha)\right>$, $s_e =  \Phi_{e}^{(2,3)}(\mathrm{ N}_\alpha,\mathrm{N}_\alpha) / p_{\rm succ}(e)$ and $\zeta_e$ be the maximum score of any distribution $\p_{s_e}(A,B|X,Y)$, generated by $s_e$ in the $[2,2]$ Bell setting, in any CHSH game. Finally, let $p_{\mathrm{win}}$ be the maximum probability of winning the ACHSH game when Alice and Charlie each share states in $\mathbb{H}^{[1]}_{\alpha,({\rm m,n})}[{\rm N}]$ with Bob. Then, 
       \begin{equation} 
       \label{Eq::MaxWinProb}
           p_{\mathrm{win}} \leq \begin{cases}
              \phantom{\Big(} \frac{3}{4} &\text{if} \quad E_{\rm coup} = \emptyset\\
              \phantom{\Big(} \frac{3}{4} + \max\limits_{e \in E_{\rm coup}} p_{\rm succ}(e)\left(\zeta_e -\frac{3}{4}\right) &\text{if} \quad E_{\rm coup} \neq \emptyset
           \end{cases}.
       \end{equation}
\end{theorem}
\noindent This is proven in Appendix~\ref{Appendix::WinningBound}.

\begin{theorem}
\label{Theorem::ACHSH1Roof}
    Let $\mathbf{p}_{\mathbb{H}^{[1]}_{\alpha,(2,2)}}$ and $\mathbf{p}_{\mathbb{H}^{[1]}_{\alpha=1,(3,2)}}$ denote the maximum winning probability of the ACHSH game when Alice and Charlie each share states in $\mathbb{H}^{[1]}_{\alpha,(2,2)}$ and $\mathbb{H}^{[1]}_{\alpha=1,(3,2)}$ respectively and let $\mathbf{p}_{\mathcal{Q}}$ denote the maximum winning probability of the ACHSH game in quantum theory. Then the following hold: 
    \begin{enumerate}
        \item $  \mathbf{p}_{\mathcal{Q}} > \mathbf{p}_{\mathbb{H}^{[1]}_{\alpha,(2,2)}}$,
        \item $  \mathbf{p}_{\mathcal{Q}} > \mathbf{p}_{\mathbb{H}^{[1]}_{\alpha=1,(3,2)}}$.
    \end{enumerate}  
\end{theorem}
\begin{proof}
    \textit{1.} For $\mathbb{H}^{[1]}_{\alpha,(2,2)}$ there is only one class of state space to check since all the PR boxes are equivalent up to local relabelling. In Section~\ref{Subsection::1Roof}, we found that there are only two types of extreme effects that are minimally $2$-preserving couplers: $e_{\rm p,\alpha}$ and of $e_{\rm m,\alpha}$.  From Eq.~\eqref{Eq::PostStateCHSH} and Eq.~\eqref{eq:p_perfect_coupling}, we get 
    $$
    \frac{3}{4} + p_{\rm succ}(e_{\rm p,\alpha})\left(\zeta_{e_{\rm p,\alpha}} -\frac{3}{4}\right) = \frac{\alpha  (\alpha +3)+1}{4 \alpha +2},\quad \frac{3}{4} + p_{\rm succ}(e_{\rm m,\alpha})\left(\zeta_{e_{\rm m\alpha}} -\frac{3}{4}\right) = \frac{\alpha  (5 \alpha +17)+4}{24 \alpha +8}.
    $$
    When $\alpha \in [1/2,1]$, both of these quantities are strictly less than $\mathbf{p}_{\mathcal{Q}}=(1+1/\sqrt{2})/2$. Further, both these functions are monotonically increasing in the range of $\alpha$ specified and at $\alpha=1$ evaluate to 5/6 and 13/16 respectively.

     \textit{2.} There are three classes of party swap symmetric $\mathbb{H}^{[1]}_{\alpha=1,(3,2)}$ state spaces  depending on whether the maximally entangled state is ${\mathrm{N_1, N_2}}$ or ${\mathrm{N_3}}$. In Section~\ref{Subsection::Couplers3Fid} we presented the probabilities and the respective CHSH scores that maximises their product in swapping scenarios involving a nonlocal state $\mathrm{N}$ for each state space. Plugging those numbers into the upper bound of $p_{\mathrm{win}}$ corresponding to $E_{coup} \neq \emptyset$ in Equation~\eqref{Eq::MaxWinProb}, always returns values strictly less than $\mathbf{p}_{\mathcal{Q}}$. 
\end{proof}
Theorem~\ref{Theorem::ACHSH1Roof} proves that quantum theory can be correlation self-tested against all the state spaces with one extremal non-local state considered in this paper. In addition to this, we have considered noisy state spaces $\mathbb{H}^{[1]}_{\alpha(3,2)}$ for $\alpha$ taking the set of discrete values between $1/2$ and $1$ with step size of $1/30$. We found that the maximum score in the ACHSH game, calculated as per Theorem~\ref{Proposition::WinningBound} remains strictly less than the quantum value for all the tested $\alpha$, leading us to conjecture that quantum theory can be correlation self-tested against any state spaces of the form $\mathbb{H}^{[1]}_{\alpha(3,2)}$ where $\alpha \in [1/2,1]$.

Next, we provide a generalisation of the first part of the previous result showing that no party symmetric state space $\mathbb{H}^{[\npr]}_{\alpha,(2,2)}$ can beat (or match) quantum theory in the ACHSH game.

\begin{theorem}
     Let $\mathbb{H}^{[\npr]}_{\alpha (2,2)}$ be a party symmetric state space with $1\leq \npr \leq 8$ and $1/2 \leq \alpha \leq 1$. Let $\mathbf{p}_{\mathbb{H}^{[\npr]}_{\alpha (2,2)}}$ be the maximum winning probability of the ACHSH game when Alice and Charlie each share a state in $\mathbb{H}^{[\npr]}_{\alpha (2,2)}$ with Bob. Further, let $\mathbf{p}_{\mathcal{Q}}$ denote the maximum winning probability of the ACHSH game in quantum theory. Then,
     \begin{equation}
         \mathbf{p}_{\mathcal{Q}} > \mathbf{p}_{\mathbb{H}^{[\npr]}_{\alpha (2,2)}}
     \end{equation}
     for any $\npr \in \{1,2, \ldots , 8\}$ and any $\alpha \in [1/2,1]$. 
\end{theorem}

\begin{proof}
  In the case $\npr=1$, i.e., considering $\mathbb{H}^{[1]}_{\alpha (2,2)}$, Theorem~\ref{Theorem::ACHSH1Roof} (part \textit{1}) implies  $\mathbf{p}_{\mathcal{Q}} > \mathbf{p}_{\mathbb{H}^{[\npr]}_{\alpha (2,2)}}$.

  For $2\leq \npr \leq 8$, Theorem~\ref{Theorem::NoCouplers} shows that $\mathbf{p}_{\mathbb{H}^{[\npr]}_{\alpha (2,2)}} \leq 3/4$ since there are no minimally $2$-preserving couplers.
\end{proof}

This theorem shows that it is impossible to match the quantum winning probability when the players have access to a single copy of the state space. If multiple copies were allowed, alternative strategies would be possible, which could involve nonlocality distillation. When $\npr \neq 1$, the absence of minimally $2$-preserving couplers makes it impossible to win the ACHSH game, beating the classical bound, even if distillation were possible. However, when $\npr=1$, the theorem above does not rule out the possibility of a strategy that relies on distillation to beat the quantum score.    

\section{Consequences of Minimal \texorpdfstring{$2$}{2}-Preservability}
\subsection{Recovering Tsirelson's Bound}
\label{Section::TsirelsonBound}
In this subsection we make a connection between minimal $2$-preservability and Tsirelson's bound. To do so we fix a party-symmetric state space, $\mathbb{H}^{[\npr]}_{\alpha (2,2)}$ with $m\in\{1,\ldots,7\}$ and consider a GPT in which the state space of any pair of systems is this state space. From the definition of a GPT (Definition~\ref{Def::GPT}), for any pair of bipartite states $\sigma, \omega \in \mathcal{S}^{\boxtimes 2}$, the states $\{ \sigma^{(1,2)} \otimes \omega^{(3,4)}, \sigma^{(1,4)} \otimes \omega^{(2,3)}\}$ are valid states of the 4-party state space $\mathcal{S}^{\boxtimes 4}$.  Consider also the maximal effect space $\mathcal{E}_{\mathcal{S}^{\boxtimes 2}}$ associated with $\mathcal{S}^{\boxtimes 2}$. The next theorem shows that all of these effects are minimally $2$-preserving if and only if Tsirelson's bound holds.

\begin{theorem}
\label{Theorem::TsirelsonBound}
    Let $\mathbb{H}^{[\npr]}_{\alpha(2,2)}$ be a party symmetric bipartite state space with $\npr\in\{1,\ldots,7\}$ and let $\mathbf{p}_{\mathcal{Q}}$ denote Tsirelson's bound. Then the following two statements are equivalent:
        \begin{enumerate}
           \item  any $e \in \mathcal{E}_{\mathbb{H}^{[\npr]}_{\alpha(2,2)}}$ is minimally $2$-preserving
            \item  $\max_i\chsh_i[s] \leq  \mathbf{p}_{\mathcal{Q}}$ for any state $s \in \mathbb{H}^{[\npr]}_{\alpha (2,2)}$.
        \end{enumerate}
\end{theorem}
\begin{proof}
We split the proof into cases. For $\npr=1$, we can use the analysis in Section~\ref{sec:CH1R} where the state space containing $\pr_2$ was used. We applied each of the 6 extremal effects $\{\id\otimes e_{\rm CH_{j}}\otimes\id\}_{j\neq2,2'}$ to $\pr_2\otimes\pr_2$ and found that in each case one of the CHSH games gives a score of $(\alpha^2+1)/2$ when applied to the post-measurement state (i.e., to the states formed by renormalising $\Tilde{\Phi}^{(2,3)}_{e_{\rm CH_{j \neq 2,2'}}}(\pr_{2,\alpha},\pr_{2,\alpha})$.)
To be valid states, these 6 CHSH values need to be below $3/4$, which is the case if and only if $\alpha\leq1/\sqrt{2}$. Thus, all the effects in $\mathcal{E}_{\mathbb{H}^{[1]}_{\alpha(2,2)}}$ are minimally $2$-preserving if and only if $\alpha \leq 1/\sqrt{2}$. The maximum CHSH score achievable in such a state space is
    \begin{equation}
      \max_{\frac{1}{2} \leq \alpha \leq \frac{1}{\sqrt{2}}}  {\rm CHSH_{2}} \left[\pr_{2,\alpha} \right] = \frac{1}{2}\left(1+\frac{1}{\sqrt{2}}\right) =  \mathbf{p}_{\mathcal{Q}}.
    \end{equation}

    When $\npr=2$, from the example in Section~\ref{Example::PR1PR2}, for the state space $\mathbb{H}^{[2]}_{\alpha (2,2)}[\pr_{1,2}]$ the set of extreme effects that are not minimally $2$-preserving are all the CH-type (Type~4) effects with $\alpha>1/\sqrt{2}$ and the couplers (Type~3 effects) for $\alpha>(1+\sqrt{41})/10$. Hence all the effects are minimally $2$-preserving if and only if $\alpha \leq 1/\sqrt{2}$.

The same argument extends to $3\leq\npr\leq7$.
\end{proof}
When $\npr=8$, all effects in the maximal effect space $\mathcal{E}_{\mathbb{H}^{[8]}_{\alpha(2,2)}}$ are minimally $2$-preserving for all $\alpha$, so Theorem~\ref{Theorem::TsirelsonBound} cannot be extended to this case.

\subsection{Generalisation of the No-Restriction Hypothesis}
\label{Section::GenNRH}
Given a state space, the no-restriction hypothesis states that the effect space contains all effects that give an inner product between 0 and 1 when applied to all states. As has been noted before~\cite{PhysRevLett.102.110402,PhysRevA.73.012101}, it can happen that an effect that obeys this hypothesis for a pair of systems leads to an invalid state when applied to two halves of two bipartite systems. Requiring that this does not happen (i.e., requiring weak minimal $2$-preservability) hence further constrains the maximal effect space in general. In the present work, we propose a further constraint for an effect to be valid, namely that it must be part of a valid measurement, and have given examples of effects that are weakly minimally $2$-preserving, but which cannot be part of a measurement in which all effects are weakly minimally $2$-preserving.  This suggests that a further criterion is needed to cater for compositional consistency, which can be seen as a further generalisation of the no-restriction hypothesis.

\subsection{Connection to Previous work~\texorpdfstring{\cite{PhysRevLett.119.020401,dallarno2023signaling}}{}}
\label{subsection::RelatedMin2Pres}
In~\cite{PhysRevLett.119.020401,dallarno2023signaling}, the idea of weak minimal $2$-preservability has been visited to argue which bipartite compositions of gbit state spaces admit the criterion. The set of compositions considered are the ones where the extremal states are subsets of the extremal states of $\mathbb{H}^{[8]}_{(2,2)}$ and the extremal effects are subsets of the extremal effects of $\mathbb{H}^{[0]}_{(2,2)}$, where only states and effects that lie on extremal rays of the state and effect cones, respectively, are considered. We attempted to reproduce their results and found the following list of compatible state and effect spaces, which we divide into 3 classes\footnote{For these three classes we need the notion of min-tensor product for effect spaces which is defined analogously to that of state spaces (cf.\ Definition~\ref{def::MinAndMax}). Note also that, in general, $\mathcal{E}_{\A} \ntens \mathcal{E}_{\B} \neq \mathcal{E}_{\mathcal{S}_{\A} \xtens \mathcal{S}_{\B}} $ (cf.\ Table~\ref{Table::BWEffects3Fid}).}

\begin{enumerate}
    \item $\mathbb{H}^{[0]}_{(2,2)}$ with $\mathrm{ConvHull}\left\{\mathcal{E}_{\mathcal{G}_2^2 } \ntens  \mathcal{E}_{\mathcal{G}_2^2 } \cup \{e_{\rm CH_i}\}_{\rm i=1}^4 \cup \{e'_{\rm CH_i}\}_{\rm i=1}^4   \right\}$ 
    \item $\mathbb{H}^{[8]}_{(2,2)}$ with $\mathcal{E}_{\mathbb{H}^{[8]}_{(2,2)}}$,
   
    \item \begin{itemize}
        \item $\mathbb{H}^{[1]}_{(2,2)}[\pr_1]$ with $\mathrm{ConvHull} \{ \mathcal{E}_{\mathcal{G}_2^2 } \ntens  \mathcal{E}_{\mathcal{G}_2^2 } \cup \{ 2/3e_{\mathrm{CH}_1} \}\}$,
        \item $\mathbb{H}^{[1]}_{(2,2)}[\pr_1']$ with $\mathrm{ConvHull} \{\mathcal{E}_{\mathcal{G}_2^2 } \ntens  \mathcal{E}_{\mathcal{G}_2^2 } \cup \{ 2/3e_{\mathrm{CH}_1'} \}\}$,
        \item $\mathbb{H}^{[1]}_{(2,2)}[\pr_2]$ with $\mathrm{ConvHull} \{ \mathcal{E}_{\mathcal{G}_2^2 } \ntens  \mathcal{E}_{\mathcal{G}_2^2 } \cup \{ 2/3e_{\mathrm{CH}_2} \}\}$,
        \item $\mathbb{H}^{[1]}_{(2,2)}[\pr_2']$ with $\mathrm{ConvHull} \{ \mathcal{E}_{\mathcal{G}_2^2 } \ntens  \mathcal{E}_{\mathcal{G}_2^2 } \cup \{ 2/3e_{\mathrm{CH}_2'} \}\}$,
        \item $\mathbb{H}^{[1]}_{(2,2)}[\pr_3]$ with $\mathrm{ConvHull} \{ \mathcal{E}_{\mathcal{G}_2^2 } \ntens  \mathcal{E}_{\mathcal{G}_2^2 } \cup \{ 2/3 e_{\mathrm{CH}_3} \}\}$,
        \item $\mathbb{H}^{[1]}_{(2,2)}[\pr_3']$ with $\mathrm{ConvHull} \{ \mathcal{E}_{\mathcal{G}_2^2 } \ntens  \mathcal{E}_{\mathcal{G}_2^2 } \cup \{2/3e'_{\mathrm{CH}_3} \}\}$,
        \item $\mathbb{H}^{[1]}_{(2,2)}[\pr_4]$ with $\mathrm{ConvHull} \{\mathcal{E}_{\mathcal{G}_2^2 } \ntens  \mathcal{E}_{\mathcal{G}_2^2 } \cup \{ 2/3e_{\mathrm{CH}_4} \}\}$,
        \item $\mathbb{H}^{[1]}_{(2,2)}[\pr_4']$ with $\mathrm{ConvHull} \{ \mathcal{E}_{\mathcal{G}_2^2 } \ntens  \mathcal{E}_{\mathcal{G}_2^2 } \cup \{ 2/3e'_{\mathrm{CH}_4} \}\}$.
    \end{itemize}
\end{enumerate}

Our list partially differs from the findings in~\cite{PhysRevLett.119.020401}, in that the authors find an additional class containing
\begin{itemize}
        \item $ \mathbb{H}^{[2]}_{(2,2)}[\pr_{3,3'}]$ with $\mathrm{ConvHull} \{ \mathcal{E}_{\mathcal{G}_2^2 } \ntens  \mathcal{E}_{\mathcal{G}_2^2 } \cup \{ e_{\mathrm{CH}_4}, e_{\mathrm{CH}_4}' \}\}$,
        \item $\mathbb{H}^{[2]}_{(2,2)}[\pr_{4,4'}]$ with $\mathrm{ConvHull} \{ \mathcal{E}_{\mathcal{G}_2^2 } \ntens  \mathcal{E}_{\mathcal{G}_2^2 } \cup \{ e_{\mathrm{CH}_3}, e_{\mathrm{CH}_3}' \}\}$,
\end{itemize} 
and instead of the last four examples of our class 3 they get the four cases
\begin{itemize}        
        \item $\mathbb{H}^{[1]}_{(2,2)}[\pr_3]$ with $\mathrm{ConvHull} \{ \mathcal{E}_{\mathcal{G}_2^2 } \ntens  \mathcal{E}_{\mathcal{G}_2^2 } \cup \{ e_{\mathrm{CH}_4'} \}\}$,
        \item $\mathbb{H}^{[1]}_{(2,2)}[\pr_3']$ with $\mathrm{ConvHull} \{ \mathcal{E}_{\mathcal{G}_2^2 } \ntens  \mathcal{E}_{\mathcal{G}_2^2 } \cup \{e_{\mathrm{CH}_4} \}\}$,
        \item $\mathbb{H}^{[1]}_{(2,2)}[\pr_4]$ with $\mathrm{ConvHull} \{ \mathcal{E}_{\mathcal{G}_2^2 } \ntens  \mathcal{E}_{\mathcal{G}_2^2 } \cup \{ e_{\mathrm{CH}_3'} \}\}$,
        \item $\mathbb{H}^{[1]}_{(2,2)}[\pr_4']$ with $\mathrm{ConvHull} \{ \mathcal{E}_{\mathcal{G}_2^2 } \ntens  \mathcal{E}_{\mathcal{G}_2^2 } \cup \{ e_{\mathrm{CH}_3} \}\}$.
\end{itemize}
This discrepancy arises because some of the individual state spaces within these cases require two system types due to the lack of party symmetry in the bipartite states. For instance, the final entry in class 3 involves $\pr'_4$, which is not invariant under party swap. Thus, there are effects that can be applied to $\pr'_4$ but not to $\pr'_4$ after a swap of the two subsystems. When considering weak minimal 2-preservability in~\cite{PhysRevLett.119.020401} the effects were applied without considering the two system types, leading to the additional cases written above. In more detail, if we call the system types $\A$ and $\B$, then $\pr'_4$ can be considered a system comprising $s_\A s_\B$ and when we compute the effect space of it, we get effects that apply to $s_\A s_\B$. However, with two copies of $\pr'_4$ we have systems $s_\A s_\B s_\A s_\B$. When applying an effect to the middle two systems, $s_\B s_\A$ we cannot directly apply an effect that acts on $s_\A s_\B$, but we need to reverse the order of the systems. It appears this was not done in~\cite{PhysRevLett.119.020401}. 
In~\cite{dallarno2023signaling}, the requirement of a single system type was implicitly added by requiring that the swap operation is valid. This removes the additional class and corrects the final 4 elements of class 3, thus matching our list.

These classifications do not characterise all bipartite compositions of gbits that satisfy the minimal (or weakly minimal) $2$-preservability criterion since they require all extremal states and effects to be ray extremal as well. On dropping this, our results show that there exist several other bipartite compositions for which all the effects are minimally (or weakly minimally) $2$-preserving, for instance, theories with state spaces of the form $\mathbb{H}^{[1]}_{\alpha=1,(2,2)}$ containing a party symmetric PR box that have a restricted effect space constructed from the convex hull of BW effects, the 9 coupling effects and their complementary effects. Compared to the results from~\cite{dallarno2023signaling}, our theories have $17$ additional extremal (entangled) effects (see Section~\ref{Subsection::1Roof}) and in the previous sections we found more general effects for the state space $\mathbb{H}^{[1]}_{\alpha=1,(2,2)}$, which means that we can also define additional consistent theories. The reason for the discrepancy between our work and~\cite{dallarno2023signaling} appears to be that in~\cite{dallarno2023signaling} only the ray extremal effects emanating from the zero effect are considered, but these do not give the maximal set of extremal effects. In general, the complement of each ray extremal effect gives rise to another extremal effect, but additional extremal effects arise that are not of this form and cannot be readily obtained from the ray extremal effects.

In addition, we find that there are bipartite compositions of $\mathcal{G}_2^2$ beyond those presented in~\cite{dallarno2023signaling}. Indeed, there can be bipartite state spaces where the extremal non-local states are not PR-boxes, while the local state space remains $\mathcal{G}_2^2$. For instance, $\mathbb{H}^{[1]}_{\alpha,(2,2)}$ with its maximal effect space is also minimally $2$-preserving (and therefore potentially completely preserving) for $\alpha\leq1/\sqrt{2}$ (cf.\ Theorem~\ref{Theorem::TsirelsonBound}). Other examples include the state spaces $\mathbb{H}^{[\npr]}_{\alpha, (2,2)}$ with $\npr$ even where the PR boxes are isotropically opposite and the restricted effect spaces are constructed by taking the convex hull of all the extreme effects of $\mathcal{E}_{\mathbb{H}^{[\npr]}_{\alpha,(2,2)}}$ after removing the CH type effects. We found that these examples are minimally $2$-preserving accounting for both states and effects.

\section{Discussion}

We considered asymmetric bipartite compositions of gbit state spaces that potentially allow for entanglement swapping. The state spaces of such compositions were taken to be the convex hull of local deterministic states and noisy PR boxes such that the overall state space is preserved under party swap symmetry. Within these models we found no examples that outperform quantum theory in the ACHSH game. We studied examples that have both effects enacting entanglement swapping and maximally nonlocal correlations, and showed that quantum correlations can be successfully self-tested against them.  

Our results rely on the fact that certain elements in the dual of the state space appear to be valid effects when some number of systems is considered, but fail to be valid effects when applied to larger systems.
Our results suggest adding a further requirement that an effect is only valid if it features in a valid measurement. The significance of effects that otherwise seem valid but are not part of any valid measurement is left as an open question.

In order to characterise the effect polytopes of the noisy state spaces as a function of the noise variable ($\alpha$), we were unable to directly use vertex enumeration software because these require specifying $\alpha$. Instead, we developed a technique in which the vertices of the largest polytope, corresponding to the extremal value $\alpha=1$, were calculated and then used together with the complete set of hyperplane constraints expressed in terms of $\alpha$ to reduce the larger polytope to obtain the required vertices in terms of $\alpha$. Although our case involved only one variable, the mentioned technique can be extended to more. Additionally, this technique can also be used in scenarios where a complete vertex enumeration is computationally expensive. In particular, we expect the technique to be useful in cases where we have a vertex description of a larger polytope that can be reduced to the polytope of interest using a small number of cuts.

We have shown that quantum theory can be correlation self-tested against a class of theories that allow post-quantum correlations. It is desirable to extend this to further theories hence providing even greater confidence in quantum theory being the correct description of the world. Although one can arbitrarily truncate BW state spaces to generate more examples, finding an argument covering all such state spaces may not be straightforward. In our work, the compositions of $\mathcal{G}_3^2$ we considered only involved a single entangled state. However, since there are multiple equivalence classes of maximally entangled states for this case, those with two or more entangled states would also be interesting to investigate.

The optimal strategy used in quantum theory in the ACHSH game involves perfect entanglement swapping, in which the systems held by Alice and Charlie are maximally entangled after post-selecting on each of Bob's outcomes. This feature can be mimicked if more than one composition rule is allowed, such as the one in~\cite{barnum2008teleportation} or~\cite{PhysRevA.110.022225}.
Explicitly, one possibility is allowing BW compositions between $A$ and $B$, between $B'$ and $C$ and between $A$ and $C$, while having only local composition allowed between $B$ and $B'$. With this one can perfectly win the ACHSH game when Bob shares two copies of $\pr_1$, one with Alice and another with Charlie, and Bob performs a four outcome joint measurement with appropriate CH type effects (which exist due to the min-tensor product between $B$ and $B'$). We have avoided constructions with multiple composition rules as these stray further from quantum theory. 

\section*{Acknowledgements}
We thank Rutvij Bhavsar for helpful discussions on Section~\ref{Sec::MinimalPreservability} and Michele Dall'Arno on Section~\ref{subsection::RelatedMin2Pres} through private communications. This work was supported by the Swiss National Science Foundation (Ambizione PZ00P2-208779), Departmental Studentship from the Department of Mathematics, University of York and l’Agence Nationale de la Recherche (ANR) project ANR-22-CE47-001.

\section*{Data Availability}
Supplementary Material to this paper can be found here~\cite{sengupta_2025_17520323}. 


%

\appendix

\section{Classical Probability Theory as a GPT}
\label{Appendix::ClassicalGPT}

\textit{Classical probability theory} 
can be viewed as a GPT 
in which local tomography of states requires only one fiducial measurement. For instance, if the system is a biased coin, then its state can be specified using only one fiducial measurement, corresponding to observing the outcome when the coin is tossed. To its event space $\{0,1\}$, we associate a probability distribution $\p_{\mathrm{coin}}$ with the property  $\p_{\mathrm{coin}}(0)=p$ and $\p_{\mathrm{coin}}(1)=1-p$ for some $p \in [0,1]$. The extreme states of the coin correspond to the deterministic outcomes when either $p=0$ or $p=1$. The probability state space corresponding to the state of a biased coin can be geometrically represented by a line segment in $\mathbb{R}^2$ with $(1,0)$ and $(0,1)$ as the extreme states. Similarly, for a certain event with 3 outcomes, the extreme states are:
$$
s_1=\begin{pmatrix}
    1\\
    0\\
    0\\
\end{pmatrix}, \quad 
s_2=\begin{pmatrix}
    0\\
    1\\
    0\\
    \end{pmatrix}, \quad 
s_3=\begin{pmatrix}
    0\\
    0\\
    1\\
\end{pmatrix},
$$
i.e., the vertices of a triangle. In general, the state space for a $k$-outcome measurement can be geometrically represented by a $(k-1)$-simplex. The min- and the max-tensor product of state spaces that can be represented by simplices coincide and is also a simplex~\cite{Aubrun2021,GPTbroadcasting}.

\section{Quantum Strategy in the ACHSH Game}

\label{sub::QStrat}

In quantum theory the players can use a strategy where Alice shares a two qubit maximally entangled state $\rho_{AB}$ with Bob and Charlie shares another two qubit maximally entangled state $\rho_{B'C}$ with Bob. Then Bob performs a joint measurement in the Bell basis on his two qubits ($B$ and $B'$). This is an entanglement swapping operation and, for each outcome of the measurement, Alice and Charlie will be left with a maximally entangled state. For example, if the Bell basis is denoted by 
\begin{align}
\begin{split}
\ket{\psi_{00}} &= \frac{1}{\sqrt{2}}\left(\ket{00}_{BB'}+\ket{11}_{BB'}\right),\\
\ket{\psi_{01}} &= \frac{1}{\sqrt{2}}\left(\ket{00}_{BB'}-\ket{11}_{BB'}\right),\\
\ket{\psi_{10}} &= \frac{1}{\sqrt{2}}\left(\ket{01}_{BB'}+\ket{10}_{BB'}\right),\\
\ket{\psi_{11}} &= \frac{1}{\sqrt{2}}\left(\ket{01}_{BB'}-\ket{10}_{BB'}\right),\\
\end{split}
\end{align}
then the resultant state held by Alice and Charlie after the projection $\ket{\psi_{00}}$ is $\frac{1}{\sqrt{2}}\left(\ket{00}_{AC}+\ket{11}_{AC}\right)$ with an associated probability of $\frac{1}{4}$, and so on. Further, denoting $\ket{\theta}=\cos{\theta}\ket{0}+\sin{\theta}\ket{1}$, Alice and Charlie execute the following operations. 
\begin{itemize}
    \item When $X=0$, Alice measures in $\{\ket{0},\ket{\pi}\}$ basis.
    \item When $X=1$, Alice measures in $\{\ket{\pi /2},\ket{3\pi /2}\}$ basis.
    \item When $Z=0$, Charlie measures in $\{\ket{\pi / 4},\ket{5 \pi /4}\}$ basis.
    \item When $Z=1$, Charlie measures in $\{\ket{3 \pi / 4},\ket{7 \pi /4}\}$ basis.
    \item For each measurement, if the first element of the basis is obtained, the outcome ($A$/$C$) is set to 0, otherwise it is set to 1.
\end{itemize}
Using the notation specified in Eq.~\eqref{Eq::TsirelsonNotation}, we obtain
\begin{equation}
    \mathrm{p}\left(A,C|X,Z,B=00) \right)=\frac{1}{4}\left(\begin{array}{cc|cc}
1+\epsilon &1-\epsilon &1-\epsilon & 1+\epsilon\\
1-\epsilon &1+\epsilon & 1+\epsilon &1-\epsilon\\
\hline 
1+\epsilon &1-\epsilon & 1+\epsilon & 1-\epsilon \\
1-\epsilon & 1+\epsilon &1-\epsilon & 1+\epsilon\\
\end{array}\right),
\end{equation}
where $\epsilon= 1/\sqrt{2}$. With this state the players win the CHSH game $a \oplus c= \overline{x}\cdot z $ with a score of $2\left(1+\frac{1}{\sqrt{2}}\right)$. Putting together winning scores for the other outcomes with their associated probabilities, the overall winning probability sums to $\frac{1}{2}\left(1+\frac{1}{\sqrt{2}}\right)$. Since this is the maximum score that can be achieved in any CHSH game, it must be the optimum strategy for the ACHSH game.

\section{Classification of extreme effects of \texorpdfstring{$\mathcal{E}_{\mathbb{H}^{[0]}_{(3,2)}}$}{} and \texorpdfstring{$\mathcal{E}_{\mathbb{H}^{[1344]}_{(3,2)}}$}{}}
\label{Appendix::Effects3Fid}

\begin{table}[H]
    \centering
    \begin{tabular}{|c|c|c|}
    \hline
        Table~\ref{Table::BWEffects3Fid} Effects, $248$ & $\frac{1}{3}\left(
\begin{array}{cc|cc|cc}
 0 & 1 & 0 & 1 & 0 & 0 \\
 0 & 0 & 0 & 0 & 0 & 1 \\ \hline
 0 & 0 & 0 & 0 & 0 & 0 \\
 0 & -1 & 0 & -1 & 1 & 0 \\ \hline
 0 & 0 & 0 & 1 & 0 & 0 \\
 0 & 1 & 0 & 0 & 0 & 0 \\
\end{array}
\right), 576 $ &  $\left(
\begin{array}{cc|cc|cc}
 0 & 1 & 0 & 0 & 0 & 0 \\
 0 & 0 & 1 & 0 & 0 & 0 \\ \hline
 0 & 0 & 0 & 0 & 0 & 0 \\
 0 & -1 & 0 & 1 & 0 & 0 \\ \hline
 0 & 0 & 0 & 0 & 0 & 0 \\
 0 & 0 & 0 & 0 & 0 & 0 \\
\end{array}
\right), 72 $ \\ \hline 
$\frac{1}{5}\left(
\begin{array}{cc|cc|cc}
 0 & 2 & 0 & 0 & 0 & 0 \\
 0 & 0 & 2 & 0 & -1 & 0 \\ \hline
 0 & 0 & 0 & 0 & 0 & 0 \\
 0 & -2 & 0 & 1 & 2 & 0 \\ \hline
 0 & 1 & 1 & 0 & 1 & 0 \\
 0 & 0 & 0 & 1 & 0 & 0 \\
\end{array}
\right), 2304 $ &
        $\frac{1}{3}\left(
\begin{array}{cc|cc|cc}
 0 & 1 & 0 & 0 & 0 & 0 \\
 0 & 0 & 1 & -1 & 0 & 0 \\ \hline
 0 & 0 & 0 & 0 & 0 & 1 \\
 0 & -1 & 0 & 1 & 0 & 0 \\ \hline
 0 & 1 & 0 & 0 & 1 & 0 \\
 0 & 0 & 0 & 1 & 0 & 0 \\
\end{array}
\right), 2304$ & $\frac{1}{3}\left(
\begin{array}{cc|cc|cc}
 0 & 1 & 0 & 0 & 0 & 0 \\
 0 & -1 & 1 & 0 & 0 & 0 \\ \hline
 0 & 1 & 0 & 0 & 0 & 0 \\
 0 & 0 & 0 & 1 & 1 & 0 \\ \hline
 0 & 0 & 1 & 0 & 0 & 0 \\
 0 & 0 & 0 & 0 & 0 & 0 \\
\end{array}
\right), 2304$ \\ \hline 
$\frac{1}{3}\left(
\begin{array}{cc|cc|cc}
 0 & 1 & 0 & 0 & 1 & 0 \\
 0 & 0 & 1 & 0 & 0 & 0 \\ \hline
 0 & 0 & 0 & 0 & 0 & 0 \\
 0 & -1 & 0 & 1 & 1 & 0 \\ \hline
 0 & 0 & 1 & 0 & 0 & 0 \\
 0 & 0 & 0 & 0 & 0 & 0 \\
\end{array}
\right), 2304$ & $\frac{1}{5}\left(
\begin{array}{cc|cc|cc}
 0 & 0 & 0 & 2 & 0 & 0 \\
 0 & -1 & 0 & 0 & 0 & 1 \\ \hline
 0 & 0 & 0 & 0 & 0 & 0 \\
 2 & 0 & 0 & -1 & 0 & 1 \\ \hline
 0 & 0 & 0 & 1 & 1 & 0 \\
 0 & 2 & 0 & 0 & 0 & 1 \\
\end{array}
\right),2304$ &
        $\frac{1}{3}\left(
\begin{array}{cc|cc|cc}
 0 & 1 & 0 & 1 & 0 & 0 \\
 0 & 0 & 0 & 0 & 0 & 0 \\ \hline
 0 & 0 & 0 & 0 & 0 & 0 \\
 1 & 0 & 0 & -1 & 0 & 1 \\ \hline
 1 & 0 & 0 & 0 & 1 & 0 \\
 0 & 0 & 0 & 1 & 0 & 0 \\
\end{array}
\right), 2304$ \\ \hline
 $\frac{1}{3}\left(
\begin{array}{cc|cc|cc}
 0 & 1 & 0 & 0 & 0 & 0 \\
 0 & 0 & 0 & -1 & 0 & 1 \\ \hline
 0 & 1 & 0 & 0 & 0 & 0 \\
 0 & -1 & 0 & 1 & 1 & 0 \\ \hline
 1 & 0 & 0 & 0 & 0 & 1 \\
 0 & 0 & 0 & 1 & 0 & 0 \\
\end{array}
\right),1152$ & $\frac{1}{5}\left(
\begin{array}{cc|cc|cc}
 0 & 2 & 0 & 1 & 0 & 0 \\
 0 & 0 & 1 & 0 & 0 & 1 \\ \hline
 0 & 0 & 0 & 0 & 0 & 2 \\
 2 & 0 & 0 & 1 & 0 & 0 \\ \hline
 0 & 0 & 1 & 0 & 1 & 0 \\
 0 & -1 & 0 & 1 & 0 & 0 \\
\end{array}
\right), 2304$ & $\frac{1}{5}\left(
\begin{array}{cc|cc|cc}
 0 & 0 & 0 & 2 & 1 & 0 \\
 2 & 0 & 0 & 0 & 0 & 0 \\ \hline
 0 & 0 & 0 & 0 & 1 & 0 \\
 1 & 0 & 0 & 1 & 0 & 1 \\ \hline
 0 & 1 & 2 & 0 & 0 & -1 \\
 0 & 0 & 0 & 0 & 0 & 0 \\
\end{array}
\right),2304$ \\ \hline
        $\frac{1}{3}\left(
\begin{array}{cc|cc|cc}
 0 & 1 & 0 & 1 & 0 & 0 \\
 1 & 0 & 0 & 0 & 0 & 0 \\ \hline
 0 & 0 & 0 & 0 & 1 & 0 \\
 1 & 0 & 0 & 1 & 0 & 0 \\ \hline
 0 & 0 & 1 & 0 & 0 & 0 \\
 0 & 0 & 0 & 0 & 0 & 0 \\
\end{array}
\right),2304$ & $\frac{1}{3}\left(
\begin{array}{cc|cc|cc}
 0 & 1 & 0 & 1 & 0 & 0 \\
 0 & 0 & 1 & 0 & 0 & 0 \\ \hline
 0 & 0 & 0 & 0 & 1 & 0 \\
 0 & -1 & 0 & 1 & 0 & 0 \\ \hline
 0 & 0 & 1 & 0 & 0 & 1 \\
 1 & 0 & 0 & 0 & 0 & 0 \\
\end{array}
\right),2304$ & $\frac{1}{3}\left(
\begin{array}{cc|cc|cc}
 0 & 1 & 0 & 1 & 0 & 0 \\
 0 & 0 & 0 & 0 & 0 & 0 \\ \hline
 0 & 0 & 0 & 0 & 1 & 0 \\
 1 & 0 & 1 & 0 & 0 & 0 \\ \hline
 0 & 0 & 1 & 0 & 0 & 0 \\
 1 & 0 & 0 & 0 & 0 & 0 \\
\end{array}
\right),2304$ \\ \hline
$\frac{1}{3}\left(
\begin{array}{cc|cc|cc}
 0 & 0 & 0 & 1 & 1 & 0 \\
 1 & 0 & 0 & 0 & 0 & 0 \\ \hline
 0 & 0 & 0 & 0 & 0 & 1 \\
 1 & 0 & 0 & 0 & 0 & 0 \\ \hline
 0 & 1 & 1 & 0 & 0 & 0 \\
 0 & 0 & 0 & 0 & 0 & 1 \\
\end{array}
\right),576$  & $-$ & $-$ \\ \hline        
    \end{tabular}
     \caption{Effects of $\mathcal{E}_{\mathbb{H}^{[0]}_{(3,2)}}$ up to relabelling. Only the first 248 effects are separable. The number after each effect denotes the size of the class represented by that effect.}
    \label{Table::LocalEffects3Fid}    
\end{table}

\begin{table}[H]
    \centering
    \begin{tabular}{|c|c|c|c|}
    \hline
        $\left(
\begin{array}{cc|cc|cc}
 0 & 0 & 0 & 0 & 0 & 0 \\
 0 & 0 & 0 & 0 & 0 & 0 \\ \hline
 0 & 0 & 0 & 0 & 0 & 0 \\
 0 & 0 & 0 & 0 & 0 & 0 \\ \hline
 0 & 0 & 0 & 0 & 0 & 0 \\
 0 & 0 & 0 & 0 & 0 & 0 \\
\end{array}
\right),{1}$ & $\left(
\begin{array}{cc|cc|cc}
 1 & 0 & 0 & 0 & 0 & 0 \\
 0 & 0 & 0 & 0 & 0 & 0 \\ \hline
 0 & 0 & 0 & 0 & 0 & 0 \\
 0 & 0 & 0 & 0 & 0 & 0 \\ \hline
 0 & 0 & 0 & 0 & 0 & 0 \\
 0 & 0 & 0 & 0 & 0 & 0 \\
\end{array}
\right),{36}$ & $\left(
\begin{array}{cc|cc|cc}
 1 & 0 & 0 & 0 & 0 & 0 \\
 0 & 0 & 0 & 0 & 0 & 0 \\ \hline
 0 & 1 & 0 & 0 & 0 & 0 \\
 0 & 0 & 0 & 0 & 0 & 0 \\ \hline
 0 & 0 & 0 & 0 & 0 & 0 \\
 0 & 0 & 0 & 0 & 0 & 0 \\
\end{array}
\right),{144}$ & $\left(
\begin{array}{cc|cc|cc}
 1 & 0 & 0 & 0 & 0 & 0 \\
 0 & 1 & 0 & 0 & 0 & 0 \\ \hline
 0 & 0 & 0 & 0 & 0 & 0 \\
 0 & 0 & 0 & 0 & 0 & 0 \\ \hline
 0 & 0 & 0 & 0 & 0 & 0 \\
 0 & 0 & 0 & 0 & 0 & 0 \\
\end{array}
\right),{18}$  \\ \hline
       $\left(
\begin{array}{cc|cc|cc}
 1 & 1 & 0 & 0 & 0 & 0 \\
 0 & 0 & 0 & 0 & 0 & 0 \\ \hline
 0 & 0 & 0 & 0 & 0 & 0 \\
 0 & 0 & 0 & 0 & 0 & 0 \\ \hline
 0 & 0 & 0 & 0 & 0 & 0 \\
 0 & 0 & 0 & 0 & 0 & 0 \\
\end{array}
\right),{12}$  & $\left(
\begin{array}{cc|cc|cc}
 0 & 1 & 0 & 0 & 0 & 0 \\
 1 & 1 & 0 & 0 & 0 & 0 \\ \hline
 0 & 0 & 0 & 0 & 0 & 0 \\
 0 & 0 & 0 & 0 & 0 & 0 \\ \hline
 0 & 0 & 0 & 0 & 0 & 0 \\
 0 & 0 & 0 & 0 & 0 & 0 \\
\end{array}
\right),{36}$ & $\left(
\begin{array}{cc|cc|cc}
 1 & 1 & 0 & 0 & 0 & 0 \\
 1 & 1 & 0 & 0 & 0 & 0 \\ \hline
 0 & 0 & 0 & 0 & 0 & 0 \\
 0 & 0 & 0 & 0 & 0 & 0 \\ \hline
 0 & 0 & 0 & 0 & 0 & 0 \\
 0 & 0 & 0 & 0 & 0 & 0 \\
\end{array}
\right),{1}$ & $-$ \\
\hline
    \end{tabular}
   \caption{Extreme effects of $\mathcal{E}_{\mathbb{H}^{[1344]}_{(3,2)}}$ up to relabelling. All effects are separable. Only the first two are ray-extremal, i.e., extreme effects of $\mathcal{E}_{\mathcal{G}_3^2} \ntens  \mathcal{E}_{\mathcal{G}_3^2}$. The number beside an effect denotes the number of effects present in the class represented by that given effect.}
    \label{Table::BWEffects3Fid}
\end{table}

\section{Qubit Quantum Theory as a GPT}
\label{Appendix::QubitQTasGPT}

Qubit quantum theory is a GPT where local tomography requires three fiducial measurements with two outcomes each. Below we show how to connect the density matrix formalism of qubit quantum theory to the probability state formalism. In the density matrix formalism, a qubit is represented by a $2 \times 2$ density matrix, i.e., a positive semi-definite complex matrix with unit trace. The set of all such density matrices $\mathbb{D}(\mathbb{C}^2)$ is a strict subset of the real vector space $\mathbb{M}_{\rm h}(\mathbb{C}^2)$ of $2 \times 2$ Hermitian matrices. The effects are the POVM elements and  represented by a positive semi-definite complex matrix $E$ where $\mathbb{I}-E$ is also positive semi-definite, which ensures $0 \leq \tr[\rho E] \leq 1$ for any $\rho \in \mathbb{D}(\mathbb{C}^2)$. We denote the set of all POVM elements as $\mathcal{E}_{\mathbb{D}(\mathbb{C}^2)}$\footnote{The positive cone of complex square matrices is self-dual (see Example 2.24 of~\cite{boyd2004convex}). Therefore the positive cone generated by $\mathbb{D}(\mathbb{C}^2)$ is the same as the positive cone generated by  $\mathcal{E}_{\mathbb{D}(\mathbb{C}^2)}$.\label{Foot::SelfDual}}.  Since $\mathbb{D}(\mathbb{C}^2)$ and $\mathcal{E}_{\mathbb{D}(\mathbb{C}^2)}$ are both closed convex and compact, they form a well-defined state and effect space pair. The composite state space of two qubits is a subset of the real vector space $\mathbb{M}_{\rm h}(\mathbb{C}^2) \otimes \mathbb{M}_{\rm h}(\mathbb{C}^2)$. The composition rule $\Tilde{\otimes}$ identifies this subset as the set of $4 \times 4$ density matrices. Compositions of multiple qubits can be understood in a similar way. With these, the GPT  $\left(\mathbb{D}(\mathbb{C}^2),\mathcal{E}_{\mathbb{D}(\mathbb{C}^2)}, \Tilde{\otimes} \right)$ describes qubit quantum theory\footnote{Note that the minimal tensor product $\mathbb{D}(\mathbb{C}^2) \ntens \mathbb{D}(\mathbb{C}^2)$ is a strict subset of $\mathbb{D}(\mathbb{C}^2)^{\Tilde{\otimes}2}$ and describes the set of separable states in this formalism.} in the density matrix formalism. An analogous treatment is possible for quantum systems with higher (finite) dimensions.

The probability state formalism can be derived from above in the following way. Given a $2 \times 2$ density matrix, $\rho$, we first fix a set of fiducial measurements. A common choice is 
\begin{equation}
    \{M_x\}_{x \in \{0,1,2\}} \coloneqq \Bigg\{\Big\{\frac{\mathbb{I}+\sigma_{x+1}}{2},\frac{\mathbb{I}-\sigma_{x+1}}{2}\Big\}\Bigg\}_{x \in \{0,1,2\}}
\end{equation}
where 
$$
\sigma_1 = \begin{pmatrix}
    \begin{array}{cc}
        0 & 1 \\
        1 & 0
    \end{array}
\end{pmatrix}\!\,, \sigma_2=\begin{pmatrix}\begin{array}{cc}
        0 & -i \\
        i & 0
    \end{array}    
\end{pmatrix} \text{ and } \sigma_3= \begin{pmatrix}
    \begin{array}{cc}
        1 & 0 \\
        0 & -1
    \end{array}
\end{pmatrix}
$$
denote the Pauli matrices. The state tomography of $\rho$ with this choice of measurement then leads to an alternative representation
\begin{equation}
    \rho \longmapsto \frac{1}{2}\begin{pmatrix}
          \tr\left[ (\mathbb{I}+\sigma_1)\rho\right]\\
          \tr\left[ (\mathbb{I}-\sigma_1)\rho\right]\\
          \hline
          \tr\left[ (\mathbb{I}+\sigma_2)\rho\right]\\
          \tr\left[ (\mathbb{I}-\sigma_2)\rho\right]\\
          \hline
          \tr\left[ (\mathbb{I}+\sigma_3)\rho\right]\\
          \tr\left[ (\mathbb{I}-\sigma_3)\rho\right]\\
      \end{pmatrix} 
 = \begin{pmatrix}
          p(0|0)  \\
         p(1|0)  \\
        \hline
         p(0|1)  \\
         p(1|1)  \\
        \hline
         p(0|2)  \\
         p(1|2)  \\
      \end{pmatrix} \eqqcolon  \p_{\rho}.
      \label{eq::StateTom}
\end{equation}
We call $\mathrm{p}_{\rho}$ the \emph{probability state} associated with $\rho$ and denote the set of all probability states obtained upon performing state tomography on qubits using Pauli measurements, $\mathbb{P}[\mathbb{D}(\mathbb{C}^2)]$ the \textit{probability state space} of a qubit. That this set is convex and compact follows from the convexity and compactness of $\mathbb{D}(\mathbb{C}^2)$ and the linearity of the map. The composite probability state space for two qubits can be similarly derived by performing local tomography with all pairs of fiducial measurements on all bipartite states in $\mathbb{D}(\mathbb{C}^2)^{\Tilde{\otimes}2}$. We denote this joint state space as $\mathbb{P}[\mathbb{D}(\mathbb{C}^2)^{\Tilde{\otimes}2}]$. A chain of inclusions  $\mathbb{P}[\mathbb{D}(\mathbb{C}^2)] \ntens  \mathbb{P}[\mathbb{D}(\mathbb{C}^2)]  \subset \mathbb{P}[\mathbb{D}(\mathbb{C}^2)^{\Tilde{\otimes} 2}] \subset \mathbb{P}[\mathbb{D}(\mathbb{C}^2)] \xtens  \mathbb{P}[\mathbb{D}(\mathbb{C}^2)]$ holds since probability states corresponding to entangled qubits are not necessarily separable and $\mathbb{P}[\mathbb{D}(\mathbb{C}^2)^{\Tilde{\otimes} 2}]$ is not the maximal state space when only separable effects are considered.

\section{Construction of the Effect Polytope}

\subsection{Construction of the Effect Polytope of \texorpdfstring{$\mathbb{H}^{[1]}_{\alpha (2,2)}[\pr_2]$}{}}
\label{Appendix::EffectPoly}

The extreme effects of the local state space $\mathbb{H}^{[0]}_{ (2,2)}$ can be categorised into 8 equivalence classes based on relabelling symmetries. A representative from each class are given below:
$$
e_{0} \coloneqq \left(
\begin{array}{cc|cc}
 0 & 0 & 0 & 0 \\
 0 & 0 & 0 & 0 \\ \hline
 0 & 0 & 0 & 0 \\
 0 & 0 & 0 & 0 \\
\end{array}
\right),
e_{\rm Class I} \coloneqq \left(
\begin{array}{cc|cc}
 1 & 0 & 0 & 0 \\
 0 & 0 & 0 & 0 \\ \hline
 0 & 0 & 0 & 0 \\
 0 & 0 & 0 & 0 \\
\end{array}
\right),  e_{\rm Class II} \coloneq \left(
\begin{array}{cc|cc}
 1 & 0 & 0 & 0 \\
 0 & 1 & 0 & 0 \\ \hline
 0 & 0 & 0 & 0 \\
 0 & 0 & 0 & 0 \\
\end{array}
\right), e_{\rm Class III} \coloneq \left(
\begin{array}{cc|cc}
 1 & 0 & 0 & 0 \\
 0 & 0 & 0 & 0 \\ \hline
 0 & 1 & 0 & 0 \\
 0 & 0 & 0 & 0 \\
\end{array}
\right), $$

$$ e_{\rm Class IV} \coloneq  \left(
\begin{array}{cc|cc}
 1 & 1 & 0 & 0 \\
 0 & 0 & 0 & 0 \\ \hline
 0 & 0 & 0 & 0 \\
 0 & 0 & 0 & 0 \\
\end{array}
\right),
e_{\rm Class V} \coloneq \left(
\begin{array}{cc|cc}
 0 & 1 & 0 & 0 \\
 1 & 1 & 0 & 0 \\ \hline
 0 & 0 & 0 & 0 \\
 0 & 0 & 0 & 0 \\
\end{array}
\right), e_{\rm CH_{2}} = \left(
\begin{array}{cc|cc}
 0 & 0 & 0 & 0 \\
 0 & -1 & 0 & 1 \\ \hline
 0 & 0 & 1 & 0 \\
 0 & 1 & 0 & 0 \\
\end{array}
\right),
u = \left(
\begin{array}{cc|cc}
 1 & 1 & 0 & 0 \\
 1 & 1 & 0 & 0 \\ \hline
 0 & 0 & 0 & 0 \\
 0 & 0 & 0 & 0 \\
\end{array}
\right),
$$
 One can generate the rest of the effects from each class by applying all relabelling symmetries followed by discarding duplicates. There are 16 effects in Class I, 8 in Class II, 32 in Class III, 8 in Class IV, 16 in Class V and 8 CH type effects. These constitute the 90 extreme effects of the local state space. Effects in Class I have their complementary effect in Class V, for instance $e_{\rm Class I}+e_{\rm Class V}=u$. For the remaining classes, the complementary effects are in the same class, apart from the zero effect $e_0$ whose complementary effect is the unit effect $u$. Furthermore, when the CH type effects are removed from this list, what remains are the extreme BW effects introduced in~\ref{Subsection::ExampleBW} (for each of the CH effects, there is a PR box whose inner product with that effect is not between $0$ and $1$, but the other extremal effects of $\mathbb{H}^{[0]}_{(2,2)}$ remain valid in the BW state space $\mathbb{H}^{[8]}_{(2,2)}$).
 
\begin{table}[H]
    \centering
    \begin{tabular}{|c|c|c|c|c|c|c|c|c|c|c|c|c|c|}
    \hline
         $\phantom{\Big(}$ & \multicolumn{2}{c|}{Class I}    & \multicolumn{2}{c|} {Class II} & \multicolumn{3}{c|}{Class III} & \multicolumn{1}{c|}{Class IV} & \multicolumn{2}{c|}{Class V}& \multicolumn{3}{c|}{Class CH} \\ \hline
        $\left<\cdot,\pr_{2,\alpha}\right>$ & $ \phantom{\Big(} \frac{1-\alpha}{4} \phantom{\Big(}$ & $\phantom{\Big(}\frac{1+\alpha}{4} \phantom{\Big(}$ & $\phantom{\Big(} \frac{1-\alpha}{2} \phantom{\Big(}$ & $\phantom{\Big(} \frac{1+\alpha}{2} \phantom{\Big(}$ & $\phantom{\Big(} \frac{1-\alpha}{2} \phantom{\Big(}$ & $\phantom{\Big(} \frac{1+\alpha}{2} \phantom{\Big(}$ & $\phantom{\Big(} \frac{1}{2} \phantom{\Big(}$ & $\phantom{\Big(} \frac{1}{2} \phantom{\Big(}$ &  $\phantom{\Big(} \frac{3-\alpha}{4} \phantom{\Big(}$ & $\phantom{\Big(} \frac{3+\alpha}{4} \phantom{\Big(}$ & $\phantom{\Big(} \frac{1}{2} \phantom{\Big(}$ & $\phantom{\Big(} \frac{1+2\alpha}{2} \phantom{\Big(}$  & $\phantom{\Big(} \frac{1-2\alpha}{2} \phantom{\Big(}$  \\
         \hline
       $\phantom{\Big(}$ Count $\phantom{\Big(}$ & $8$ & $8$ & $4$ & $4$ & $8$ & $8$ & $16$ & $8$ & $8$ & $8$ & $6$ & $1$ & $1$  \\
         \hline
    \end{tabular}
    \caption{Inner products between $\mathrm{PR}_{2,\alpha}$ and the extreme effects of $\mathbb{H}^{[0]}_{(2,2)}$ (excluding the zero and unit effect). When $\alpha> 1/2$, two CH type effects give inner products outside the interval $[0,1]$. All remaining extreme effects give inner products inside $[0,1]$.}
    \label{Table:InnerProdEffectPoly}
\end{table}

We find the extreme effects of the state space $\mathbb{H}^{[1]}_{\alpha (2,2)}[\pr_2]$ by starting with the extreme effects for $\mathbb{H}^{[0]}_{(2,2)}$, finding the extremal effects that become invalid because of the additional extremal state, and then using these to construct the new extremal effects that replace the removed ones (if any). This corresponds to the following sequence of steps.
\begin{itemize}
    \item \textbf{Step 1:} Consider the hyperplanes $\left<\mathbf{x},\pr_{2,\alpha}\right>=0$  and $\left<\mathbf{x},\pr_{2,\alpha}\right>=1$ and define the set of discarded elements
    $$
    \mathcal{E}_{\mathrm{disc}} \coloneq \bigg\{e \in \mathrm{Extreme}\left[\mathcal{E}_{\mathbb{H}^{[0]}_{(2,2)}}\right]\ |\ \left<e,\pr_{2,\alpha}\right> \notin [0,1] \bigg\}
    $$
    \item \textbf{Step 2:} For each $e \in \mathcal{E}_{\mathrm{disc}}$ and $f \in \mathrm{Extreme}\left[\mathcal{E}_{\mathbb{H}^{[0]}_{(2,2)}}\right]$, construct the line segment $l_{e,f}(w) \coloneq w e + (1-w)f$, where $w\in[0,1]$.
     For each $l_{e,f}$, calculate $w'$ such that either $\left<l_{e,f}(w'),\pr_{2,\alpha}\right>=0$  or $\left<l_{e,f}(w'),\pr_{2,\alpha}\right>=1$ and store $l_{e,f}(w')$ in $\mathcal{E}_0$ or $\mathcal{E}_1$, respectively.
    \item \textbf{Step 3:} Select an element from $\mathcal{E}_{0/1}$ and try to represent it as a convex combination of other elements of that set.
        To do this, we first take the element and evaluate it for a discrete set of values of $\alpha$. For each value of $\alpha$, we run a linear program to check whether this element is equal to a convex combination the other elements of the set $\mathcal{E}_{0/1}$ for that $\alpha$. 
        \begin{enumerate}
            \item If a convex combination is found, we note the convex weights and the associated decomposition. Next, we interpolate the values of these convex weights to guess their analytic forms, which we can confirm by substitution. However, this does not guarantee that the elements for which we could not find a convex decomposition are extreme. We prove this using the next step.
            \item If a convex decomposition is not found, we construct a hyperplane as a function of $\alpha$ that separates this element from the rest of the set, showing that the chosen element is extremal for all $\alpha$.
        \end{enumerate} 
    
    \item \textbf{Step 4:} Take the union of extreme elements of $\mathcal{E}_0,\mathcal{E}_1$ and the effects $ \mathrm{Extreme}\left[\mathcal{E}_{\mathbb{H}^{[0]}_{(2,2)}}\right] \setminus \mathcal{E}_{\mathrm{disc}}$.
\end{itemize}
By following these steps we found that the extreme effects in the set $\mathcal{E}_1$ up to equivalence of relabelling symmetries are 
$$
e_1(\text{Type 1}) \coloneqq  \frac{1-\alpha}{\alpha} e_{\rm CH_2} +\left(1-\frac{1-\alpha}{\alpha}\right)\left(
\begin{array}{cc|cc}
 0 & 1 & 0 & 0 \\
 1 & 0 & 0 & 0 \\ \hline
 0 & 0 & 0 & 0 \\
 0 & 0 & 0 & 0 \\
\end{array}
\right), \ \ e_{1'}(\text{Type } 1') \coloneqq  \frac{1-\alpha}{\alpha} e_{\rm CH_2} +\left(1-\frac{1-\alpha}{\alpha}\right)
\left(
\begin{array}{cc|cc}
 0 & 1 & 0 & 0 \\
 0 & 0 & 0 & 0 \\ \hline
 1 & 0 & 0 & 0 \\
 0 & 0 & 0 & 0 \\
\end{array}
\right),
$$
$$
  e_2 (\text{Type 2}) \coloneq \frac{1-\alpha}{3\alpha-1}e_{\rm CH_2} + \left(1- \frac{1-\alpha}{3\alpha-1}\right) \left(
\begin{array}{cc|cc}
 1 & 1 & 0 & 0 \\
 1 & 0 & 0 & 0 \\ \hline
 0 & 0 & 0 & 0 \\
 0 & 0 & 0 & 0 \\
\end{array}
\right),\ \ e_3 (\text{Type 3}) \coloneq \frac{3-\alpha}{3\alpha+1} e_{\rm CH_2} + \left(1-\frac{3-\alpha}{3\alpha+1}\right)\left(
\begin{array}{cc|cc}
 0 & 1 & 0 & 0 \\
 0 & 0 & 0 & 0 \\ \hline
 0 & 0 & 0 & 0 \\
 0 & 0 & 0 & 0 \\
\end{array}
\right) 
$$
$$
\text{ and } e_4 (\text{Type 4}) \coloneq  2/(1+2\alpha) e_{\rm CH_2}.
$$
\noindent
The remaining extreme effects can be found by applying all relabelling symmetries to each of those presented above and checking that the inner product with $\pr_{2,\alpha}$ is 1. The extreme effects in the set $\mathcal{E}_0$, up to equivalence of relabelling symmetries, are the complementary effects of the effects in $\mathcal{E}_1$. We present more details below.

\textbf{Step 1:} The only extreme effects of $\mathcal{E}_{\mathbb{H}^{[0]}_{(2,2)}}$ that give an inner product outside the interval $[0,1]$ with $\pr_{2,\alpha}$ are two CH type effects, see Table~\ref{Table:InnerProdEffectPoly}. In particular, 
$$
e_{\mathrm{CH}_2}=\left(
\begin{array}{cc|cc}
 0 & 0 & 0 & 0 \\
 0 & -1 & 0 & 1 \\ \hline
 0 & 0 & 1 & 0 \\
 0 & 1 & 0 & 0 \\
\end{array}
\right) \quad \text{and} \quad
e_{\mathrm{CH}_2}'=\left(
\begin{array}{cc|cc}
 0 & 0 & 0 & 1 \\
 0 & 1 & 0 & 0 \\ \hline
 0 & 0 & 0 & 0 \\
 0 & -1 & 1 & 0 \\
\end{array}
\right).
$$

\textbf{Step 2:} The addition of $\pr_{2,\alpha}$ introduces two hyperplanes through the polytope $\mathcal{E}_{\mathbb{H}^{[0]}_{ (2,2)}}$, given by $\left<\mathbf{x}.\pr_{2,\alpha}\right>=0$ and $\left<\mathbf{x}.\pr_{2,\alpha}\right>=1$. The set $\mathcal{E}_{\mathbb{H}^{[1]}_{\alpha (2,2)}} [\pr_2]$ can be characterised as:
\begin{equation}
    \mathcal{E}_{\mathbb{H}^{[1]}_{\alpha (2,2)}} = \Big\{ e \in \mathcal{E}_{\mathbb{H}^{[0]}_{ (2,2)}} |\  0 \leq \left<e,\pr_{2,\alpha}\right> \leq 1 \Big\}
\end{equation}
i.e., the set confined in the inner half-spaces of the two hyperplanes above. The extreme effects of this polytope can be collected in two groups based on whether they are lying on the hyperplanes or not. When $\alpha \in (1/2,1)$, none of the extreme effects from Table~\ref{Table:InnerProdEffectPoly} lie on the hyperplanes and hence are extreme. To find the extreme effects lying on the hyperplanes, one can draw line segments between $e_{\rm CH_2}/e_{\rm CH_2}'$  and effects lying inside the hyperplanes and find all the points of intersection between the line segments and the hyperplanes. One can then find the convex hull of the set of intersection points on each hyperplane and find which points are extreme. The union of the extreme points found in this manner from each hyperplane constitute the remaining extreme effects of $\mathcal{E}_{\mathbb{H}^{[1]}_{\alpha (2,2)}[\pr_2]}$.  

For the hyperplane $\left<\mathbf{x}.\pr_{2,\alpha}\right>=1$, we consider line segments of the form $w(\alpha) e_{\rm CH_2} + (1- w(\alpha)) f$ where $f$ is an extreme local effect and calculate the weight $w(\alpha)$ such that  corresponding effect will lie on the hyperplane. Table~\ref{Table::WeightsPoly} summarises these weights alongside the extreme local effects such that the corresponding effect will lie on the hyperplane.

 \begin{table}[H]
     \centering
     \begin{tabular}{|c|c|c|c|c|c|c|c|c|}
     \hline
        $\phantom{\Big(} w(\alpha)\phantom{\Big(}$   &   $\phantom{\Big(}$  Class I $\phantom{\Big(}$&   $\phantom{\Big(}$ Class II $\phantom{\Big(}$ &    $\phantom{\Big(}$  Class III  $\phantom{\Big(}$ &   $\phantom{\Big(}$ Class IV $\phantom{\Big(}$ &   $\phantom{\Big(}$  Class V $\phantom{\Big(}$ &  $\phantom{\Big(}$  Class CH $\phantom{\Big(}$ &   $\phantom{\Big(}$ zero $\phantom{\Big(}$ &   $\phantom{\Big(}$ unit $\phantom{\Big(}$ \\  \hline        
         $\phantom{\Big(} \frac{1-\alpha}{\alpha} \phantom{\Big(}$ & $-$ & 4 & 8 & $-$ & $-$ & $-$ & $-$  & $-$  \\ 
          \hline
         $\phantom{\Big(} \frac{1-\alpha}{3\alpha -1} \phantom{\Big(} $ & $-$ & $-$ & $-$ & $-$ & 8 & $-$ & $-$ & $-$  \\ 
          \hline
         $\phantom{\Big(} \frac{3-\alpha}{3\alpha+1} \phantom{\Big(}$ & 8 & $-$ & $-$ & $-$ & $-$ & $-$ & $-$ & $-$  \\ 
          \hline 
           $\phantom{\Big(} \frac{2}{1+2\alpha} \phantom{\Big(}$ & $-$ & $-$ &$-$  & $-$ & $-$ & $-$ & 1 & $-$ \\ 
          \hline
         $\phantom{\Big(} \frac{3+\alpha}{1+5\alpha} \phantom{\Big(}$ & 8 & $-$ & $-$ & $-$ & $-$ & $-$ & $-$ & $-$  \\ 
          \hline
         $\phantom{\Big(} \frac{1+\alpha}{3\alpha} \phantom{\Big(}$ & $-$ & 4 & 8 & $-$ & $-$ & $-$ & $-$ & $-$  \\ 
          \hline
        $\phantom{\Big(} \frac{1}{2\alpha} \phantom{\Big(}$  & $-$ & $-$ & 16 & 8 & $-$ & 6 & $-$ & $-$  \\ 
          \hline
          $\phantom{\Big(} \frac{1+\alpha}{5\alpha-1} \phantom{\Big(}$ &  $-$ & $-$ & $-$ & $-$ & 8 & $-$ & $-$ & $-$  \\ 
          \hline
         $\phantom{\Big(} 0 \phantom{\Big(}$ & $-$ & $-$ & $-$ & $-$ & $-$ & $-$ & $-$ & 1  \\ 
          \hline         
     \end{tabular}
     \caption{Table summarises weights $w(\alpha)$ on $e_{\rm CH_2}$ such that $\left<w(\alpha) e_{\rm CH_2} + (1- w(\alpha)) f,\pr_{2,\alpha} \right>=1$, where $f$ is an extreme local effect. The numbers denote how many extreme local effects combine with $e_{\rm CH_2}$ with the corresponding weight to generate an effect on the hyperplane $\left<\mathbf{x}.\pr_{2,\alpha}\right>=1$.}
     \label{Table::WeightsPoly}
 \end{table}

\textbf{Step 3.1:} Not all the effects deduced from this procedure are extreme. We found that all the effects corresponding to weights $(3+\alpha)/(1+5\alpha), (1+\alpha)/3\alpha, 1/2\alpha$ and $(1+\alpha)/(5\alpha-1)$ can be written as a convex sums of effects obtained from the first four rows of Table~\ref{Table::WeightsPoly}. From each weight and class combination we pick one effect and show their convex decompositions below. Note that the matrix representation we use for effects has a redundancy meaning that there are many matrices that represent the same effect. We use $\nseq$ to represent an equivalence of the effects on the left and right where the matrices themselves may not satisfy the equality (such an equivalence can be checked, for instance, by computing the inner product between the representation and every local deterministic distribution, see e.g.~\cite{PhysRevA.108.062212} for more detail).
\begin{equation}
    \frac{3+\alpha}{1+5\alpha} e_{\rm CH_2} +  \frac{4\alpha-2}{1+5\alpha} e_{\rm Class I}\  \nseq\  \frac{2\alpha}{1+5\alpha} \left(
\begin{array}{cc|cc}
 0 & 0 & 0 & 0 \\
 0 & \frac{\alpha -1}{\alpha } & 0 & \frac{1-\alpha }{\alpha } \\ \hline
 1-\frac{1-\alpha }{\alpha } & 0 & \frac{1-\alpha }{\alpha } & 0 \\
 0 & \frac{1-\alpha }{\alpha } & 0 & 1-\frac{1-\alpha }{\alpha } \\
\end{array}
\right) + \frac{3 \alpha +1}{5 \alpha +1} \left(
\begin{array}{cc|cc}
 0 & 0 & 1-\frac{3-\alpha }{3 \alpha +1} & 0 \\
 0 & \frac{\alpha -3}{3 \alpha +1} & 0 & \frac{3-\alpha }{3 \alpha +1} \\ \hline
 0 & 0 & \frac{3-\alpha }{3 \alpha +1} & 0 \\
 0 & \frac{3-\alpha }{3 \alpha +1} & 0 & 0 \\
\end{array}
\right)
\end{equation}
To see that the effects on the right are indeed arise from the first four rows, notice that
$$
\left(
\begin{array}{cc|cc}
 0 & 0 & 0 & 0 \\
 0 & \frac{\alpha -1}{\alpha } & 0 & \frac{1-\alpha }{\alpha } \\ \hline
 1-\frac{1-\alpha }{\alpha } & 0 & \frac{1-\alpha }{\alpha } & 0 \\
 0 & \frac{1-\alpha }{\alpha } & 0 & 1-\frac{1-\alpha }{\alpha } \\
\end{array}
\right) = \frac{1-\alpha}{\alpha} e_{\rm CH_2} + \left(1-\frac{1-\alpha}{\alpha}\right) \left(
\begin{array}{cc|cc}
 0 & 0 & 0 & 0 \\
 0 & 0 & 0 & 0 \\ \hline
 1 & 0 & 0 & 0 \\
 0 & 0 & 0 & 1 \\
\end{array}
\right) \quad \text{and}
$$
$$
\left(
\begin{array}{cc|cc}
 0 & 0 & 1-\frac{3-\alpha }{3 \alpha +1} & 0 \\
 0 & \frac{\alpha -3}{3 \alpha +1} & 0 & \frac{3-\alpha }{3 \alpha +1} \\ \hline
 0 & 0 & \frac{3-\alpha }{3 \alpha +1} & 0 \\
 0 & \frac{3-\alpha }{3 \alpha +1} & 0 & 0 \\
\end{array}
\right) =\frac{3-\alpha }{3 \alpha +1} e_{\rm CH_2} + \left(1-\frac{3-\alpha }{3 \alpha +1}\right) \left(
\begin{array}{cc|cc}
 0 & 0 & 1 & 0 \\
 0 & 0 & 0 & 0 \\ \hline
 0 & 0 & 0 & 0 \\
 0 & 0 & 0 & 0 \\
\end{array}
\right), 
$$
where the local effect with all 0 entries except for two 1s is from Class III and the one with all 0 entries apart from one 1 is from Class I. For the remaining effects we omit this decomposition which can be readily deduced from the form of the matrices.
 \begin{align}
 \begin{split}
     \frac{1+\alpha}{\alpha} e_{\rm CH_2} +\frac{2\alpha - 1}{3\alpha} e_{\rm Class II} &\  \nseq\  \frac{1}{3}\left(
\begin{array}{cc|cc}
 0 & 0 & 1-\frac{1-\alpha }{\alpha } & 0 \\
 0 & \frac{\alpha -1}{\alpha } & 0 & 1 \\ \hline
 0 & 0 & \frac{1-\alpha }{\alpha } & 0 \\
 0 & \frac{1-\alpha }{\alpha } & 0 & 0 \\
\end{array}
\right) + \frac{1}{3} \left(
\begin{array}{cc|cc}
 0 & 0 & 0 & 0 \\
 0 & \frac{\alpha -1}{\alpha } & 0 & \frac{1-\alpha }{\alpha } \\ \hline
 1-\frac{1-\alpha }{\alpha } & 0 & \frac{1-\alpha }{\alpha } & 0 \\
 0 & 1 & 0 & 0 \\
\end{array}
\right)\\
&+ \frac{1}{3}\left(
\begin{array}{cc|cc}
 0 & 0 & 0 & 0 \\
 0 & \frac{\alpha -1}{\alpha } & 0 & \frac{1-\alpha }{\alpha } \\ \hline
 0 & 0 & 1 & 0 \\
 0 & \frac{1-\alpha }{\alpha } & 0 & 1-\frac{1-\alpha }{\alpha } \\
\end{array}
\right)\\
\end{split}
 \end{align}

\begin{align}
    \begin{split}
         \frac{1+\alpha}{\alpha} e_{\rm CH_2} +\frac{2\alpha - 1}{3\alpha} e_{\rm Class III} &\  \nseq\  \ \frac{1}{3}\left(
\begin{array}{cc|cc}
 0 & 1-\frac{1-\alpha }{\alpha } & 0 & 0 \\
 0 & \frac{\alpha -1}{\alpha } & 0 & \frac{1-\alpha }{\alpha } \\ \hline
 1-\frac{1-\alpha }{\alpha } & 0 & \frac{1-\alpha }{\alpha } & 0 \\
 0 & \frac{1-\alpha }{\alpha } & 0 & 0 \\
\end{array}
\right) + \frac{1}{3}\left(
\begin{array}{cc|cc}
 0 & 0 & 1-\frac{1-\alpha }{\alpha } & 0 \\
 0 & \frac{\alpha -1}{\alpha } & 0 & 1 \\ \hline
 0 & 0 & \frac{1-\alpha }{\alpha } & 0 \\
 0 & \frac{1-\alpha }{\alpha } & 0 & 0 \\
\end{array}
\right) \\
         &+ \frac{1}{3}\left(
\begin{array}{cc|cc}
 0 & 0 & 0 & 0 \\
 0 & \frac{\alpha -1}{\alpha } & 0 & \frac{1-\alpha }{\alpha } \\ \hline
 0 & 0 & 1 & 0 \\
 0 & \frac{1-\alpha }{\alpha } & 0 & 1-\frac{1-\alpha }{\alpha } \\
\end{array}
\right)
    \end{split}
\end{align}
\begin{align}    
         \frac{1}{2\alpha} e_{\rm CH_2} +\frac{2\alpha - 1}{2\alpha} \left(
\begin{array}{cc|cc}
 0 & 1 & 0 & 0 \\
 0 & 0 & 1 & 0 \\ \hline
 0 & 0 & 0 & 0 \\
 0 & 0 & 0 & 0 \\
\end{array}
\right) &\  \nseq\  \ \frac{1}{2}\left(
\begin{array}{cc|cc}
 0 & 1-\frac{1-\alpha }{\alpha } & 0 & 0 \\
 1-\frac{1-\alpha }{\alpha } & \frac{\alpha -1}{\alpha } & 0 & \frac{1-\alpha }{\alpha } \\ \hline
 0 & 0 & \frac{1-\alpha }{\alpha } & 0 \\
 0 & \frac{1-\alpha }{\alpha } & 0 & 0 \\
\end{array}
\right)+ \frac{1}{2} \left(
\begin{array}{cc|cc}
 0 & 0 & 0 & 0 \\
 0 & \frac{\alpha -1}{\alpha } & 0 & \frac{1-\alpha }{\alpha } \\ \hline
 0 & 0 & 1 & 0 \\
 0 & 1 & 0 & 0 \\
\end{array}
\right)   
\end{align}
\begin{equation}
    \frac{1}{2\alpha} e_{\rm CH_2} +\frac{2\alpha - 1}{2\alpha} e_{\rm Class IV} \  \nseq\  \ \frac{1}{2}\left(
\begin{array}{cc|cc}
 0 & 1-\frac{1-\alpha }{\alpha } & 0 & 0 \\
 0 & \frac{\alpha -1}{\alpha } & 0 & \frac{1-\alpha }{\alpha } \\ \hline
 1-\frac{1-\alpha }{\alpha } & 0 & \frac{1-\alpha }{\alpha } & 0 \\
 0 & \frac{1-\alpha }{\alpha } & 0 & 0 \\
\end{array}
\right) + \frac{1}{2}\left(
\begin{array}{cc|cc}
 0 & 0 & 1-\frac{1-\alpha }{\alpha } & 0 \\
 0 & \frac{\alpha -1}{\alpha } & 0 & \frac{1-\alpha }{\alpha } \\ \hline
 0 & 0 & \frac{1-\alpha }{\alpha } & 0 \\
 0 & \frac{1-\alpha }{\alpha } & 0 & 1-\frac{1-\alpha }{\alpha } \\
\end{array}
\right)
\end{equation}
\begin{equation}
     \frac{1}{2\alpha} e_{\rm CH_2} +\frac{2\alpha - 1}{2\alpha} e_{\rm CH_1} \  \nseq\  \ \frac{1}{2}\left(
\begin{array}{cc|cc}
 0 & 0 & 1-\frac{1-\alpha }{\alpha } & 0 \\
 0 & \frac{\alpha -1}{\alpha } & 0 & 1 \\ \hline
 0 & 0 & \frac{1-\alpha }{\alpha } & 0 \\
 0 & \frac{1-\alpha }{\alpha } & 0 & 0 \\
\end{array}
\right) + \frac{1}{2}\left(
\begin{array}{cc|cc}
 0 & 0 & 0 & 0 \\
 0 & \frac{\alpha -1}{\alpha } & 0 & \frac{1-\alpha }{\alpha } \\ \hline
 1-\frac{1-\alpha }{\alpha } & 0 & \frac{1-\alpha }{\alpha } & 0 \\
 0 & 1 & 0 & 0 \\
\end{array}
\right)
\end{equation}

\begin{equation}
\begin{split}
     \frac{1+\alpha}{5\alpha -1 }  e_{\rm CH_2} + \frac{4\alpha -2}{5\alpha -1}\left(
\begin{array}{cc|cc}
 1 & 1 & 0 & 0 \\
 0 & 1 & 0 & 0 \\ \hline
 0 & 0 & 0 & 0 \\
 0 & 0 & 0 & 0 \\
\end{array}
\right)& \  \nseq\  \ \frac{3 \alpha -1 }{5 \alpha -1} \left(
\begin{array}{cc|cc}
 0 & 0 & 0 & 0 \\
 0 & \frac{\alpha -1}{3 \alpha -1} & 0 & \frac{1-\alpha }{3 \alpha -1} \\ \hline
 1-\frac{1-\alpha }{3 \alpha -1} & 1-\frac{1-\alpha }{3 \alpha -1} & \frac{1-\alpha }{3 \alpha -1} & 0 \\
 0 & 1 & 0 & 0 \\
\end{array}
\right)\\ &+  \left(1- \frac{3 \alpha -1 }{5 \alpha -1}\right) \left(
\begin{array}{cc|cc}
 0 & 0 & 1-\frac{1-\alpha }{\alpha } & 0 \\
 0 & \frac{\alpha -1}{\alpha } & 0 & \frac{1-\alpha }{\alpha } \\ \hline
 0 & 0 & \frac{1-\alpha }{\alpha } & 0 \\
 0 & \frac{1-\alpha }{\alpha } & 0 & 1-\frac{1-\alpha }{\alpha } \\
\end{array}
\right)
\end{split}   
\end{equation}
Similar decompositions are possible for every other element arising from combinations of the last four non zero weights and the various classes.

\textbf{Step 3.2:} The effects arising from the first four rows of Table~\ref{Table::WeightsPoly}, we claim, are extreme. We call them Type 1 for weight $\frac{1-\alpha}{\alpha}$, Type 2 for weight $\frac{1-\alpha}{3\alpha-1}$, Type 3 for weight $\frac{3-\alpha}{3\alpha +1}$ and Type 4 for weight $\frac{2}{1+2\alpha}$. Notice that when $\alpha=1$, the Type 3 effects correspond to the noisy couplers and the Type 4 effect corresponds to the pure coupler.  To show that these effects are extreme, first consider that if a point on a polytope is extreme, the shape of the polytope will change if that point is removed and the new polytope is constructed from the convex hull of the remaining vertices. In essence, there will be a supporting hyperplane corresponding to a face of this new polytope that will witness the removed point (hyperplane separation theorem). For our purposes, we first collect all the effects from Table~\ref{Table:InnerProdEffectPoly} that satisfy $0 \leq \left<\Tilde{e},\pr_{2,\alpha} \right> \leq 1$ and all the effects generated from Table~\ref{Table::WeightsPoly} and their complementary effects lying on the hyperplane $ \left<\mathbf{x},\pr_{2,\alpha} \right>=0$. From this collection we remove one effect from either Type 1 and then perform a facet enumeration on the reduced set. This gives us a list of inequalities corresponding to the face-defining supporting hyperplanes of the reduced polytope. We then filter out the hyperplane that witnesses the removed effect. For instance, consider
$$
e_1 = \frac{1-\alpha}{\alpha} e_{\rm CH_2} +\left(1-\frac{1-\alpha}{\alpha}\right)\left(
\begin{array}{cc|cc}
 0 & 1 & 0 & 0 \\
 1 & 0 & 0 & 0 \\ \hline
 0 & 0 & 0 & 0 \\
 0 & 0 & 0 & 0 \\
\end{array}
\right) \quad \text{and}$$ 
$$ W_1 \coloneq -\frac{1}{2 \alpha -6}\left(
\begin{array}{cc|cc}
  (\alpha -1) & - (\alpha -3) &  (\alpha +1) & - (\alpha -1) \\
 - (\alpha -3) & 3 (\alpha -1) & - (\alpha -1) & 5 \alpha -3 \\ \hline
  (\alpha +1) & - (\alpha -1) & - (\alpha -3) &  (\alpha -1) \\
  (\alpha -1) &  (\alpha +1) & -3 (\alpha -1) & 3 \alpha -1 \\
\end{array}
\right).
$$
One can verify that every effect $f_1$ in the reduced polytope obtained after removing $e_1$ satisfies $\left<f_1.W_1\right> \leq 1$. However, 
$$
\left<e_1.W_1\right> = \frac{3 \alpha (\alpha -2)  +1}{\alpha(\alpha -3)  }
$$
which is $1$ when $\alpha = 1/2$ or $1$ but greater than $1$ for $\alpha\in(1/2,1)$. Since $e_1$ converges to $e_{\rm CH_2}$ as $\alpha \rightarrow 1/2$ and converges to the deterministic effect as $\alpha \rightarrow 1$, $W_1$ witnesses $e_1$. For a Type $1'$ effect, consider
$$
e_{1'} = \frac{1-\alpha}{\alpha} e_{\rm CH_2} +\left(1-\frac{1-\alpha}{\alpha}\right)\left(
\begin{array}{cc|cc}
 0 & 1 & 0 & 0 \\
 0 & 0 & 0 & 0 \\ \hline
 1 & 0 & 0 & 0 \\
 0 & 0 & 0 & 0 \\
\end{array}
\right) \quad \text{and}$$ 
$$ W_{1'} \coloneq -\frac{1}{2 \alpha -6}\left(
\begin{array}{cc|cc}
  (\alpha -1) & - (\alpha -3) &  (\alpha +1) & - (\alpha -1) \\
 (\alpha +1) & 3 (\alpha -1) & - (\alpha -1) & 5 \alpha -3 \\ \hline
  - (\alpha -3) & - (\alpha -1) & - (\alpha -3) &  (\alpha -1) \\
  (\alpha -1) &  (\alpha +1) & -3 (\alpha -1) & 3 \alpha -1 \\
\end{array}
\right),
$$
with which, one gets $\left<f_{1'},W_{1'}\right> \leq 1$ for any effect $f_{1'}$ in the reduced polytope obtained after removing $e_{1'}$, but $\left<e_{1'}.W_{1'}\right>=\left<e_{1}.W_{1}\right> > 1$ when $\alpha \in (0,1)$. For a Type 2 effect, consider
$$
e_2= \frac{1-\alpha}{3\alpha-1}e_{\rm CH_2} + \left(1- \frac{1-\alpha}{3\alpha-1}\right) \left(
\begin{array}{cc|cc}
 1 & 1 & 0 & 0 \\
 1 & 0 & 0 & 0 \\ \hline
 0 & 0 & 0 & 0 \\
 0 & 0 & 0 & 0 \\
\end{array}
\right) \quad \text{and}
$$
$$
W_2 \coloneq  \frac{1}{22-10 \alpha } \left(
\begin{array}{cc|cc}
 7-7 \alpha  & 7- \alpha  & 3 \alpha +3 & 8-8 \alpha  \\
 8-2 \alpha  & 4-4 \alpha  & 3-3 \alpha  & 8-2 \alpha  \\ \hline
 9-3 \alpha  & 5-5 \alpha  & 7- \alpha  & 5-5 \alpha  \\
 5-5 \alpha  &  \alpha +5 & 7-7 \alpha  & 3 \alpha +3 \\
\end{array}
\right)
$$
with which, one gets $\left<f_2,W_2\right> \leq 1$ for any effect $f_2$ in the reduced polytope obtained after removing $e_2$ but
$$
\left<e_2.W_2\right>= \frac{21 \alpha ^2-47 \alpha +14}{15 \alpha ^2-38 \alpha +11}
$$
which is $1$ when $\alpha =1/2$ or $1$ but greater than $1$ for $\alpha\in(1/2,1)$. For a Type 3 effect, consider
$$
e_3=\frac{3-\alpha}{3\alpha+1} e_{\rm CH_2} + \left(1-\frac{3-\alpha}{3\alpha+1}\right)\left(
\begin{array}{cc|cc}
 0 & 1 & 0 & 0 \\
 0 & 0 & 0 & 0 \\ \hline
 0 & 0 & 0 & 0 \\
 0 & 0 & 0 & 0 \\
\end{array}
\right) \quad \text{and}
$$
$$
W_3 \coloneq \frac{1}{4}\left(
\begin{array}{cc|cc}
 -2 & 4 &  \alpha +1 & 1- \alpha  \\
 2 \alpha -2 & -2 & -2 & 2 \alpha -2 \\ \hline
 2 \alpha -2 & -2&  \alpha +1 & 1- \alpha  \\
 -2 &  \alpha +1 & 1- \alpha  &  \alpha +1 \\
\end{array}
\right)
$$
\noindent
with which one has $\left<f_3.W_1\right> \leq 1$ for any effect $f_3$ in the reduced polytope obtained after removing $e_3$ but 
$$
\left<e_3.W_3\right> = \frac{(13-2 \alpha ) \alpha -1}{6 \alpha +2}
$$ 
which is $1$ when $\alpha=1/2$ but greater than $1$ otherwise. Finally for Type 4, consider $e_4 = 2/(1+2\alpha) e_{\rm CH_2}$ and
$$ W_4 \coloneq \left(
\begin{array}{cc|cc}
 - \alpha -3 &  \alpha -3 &  \alpha -3 & - \alpha -3 \\
  \alpha -3 & - \alpha -7 & - \alpha -3 &  \alpha -3 \\ \hline
  \alpha -3 & - \alpha -11 &  \alpha -3 & - \alpha -3 \\
 - \alpha -3 & \alpha -3 & - \alpha -3 &  \alpha -3 \\
\end{array}
\right).$$
One then has $\left<f_4.W_4\right> \geq 0$ for any effect $f_4$ from the reduced polytope obtained after removing $e_4$ but
$$
\left<e_4.W_4\right>  = 4-\frac{8}{2 \alpha +1}
$$
which is $0$ when $\alpha=1/2$ but negative otherwise. We suppress the details of the rest of the witnesses for the remaining effects from these 4 types. For the hyperplane $\left<\mathbf{x},\pr_{2,\alpha}\right>=0$ the extreme effects are exactly the complementary effects obtained above.

\textbf{Step 4:} Taking the union we find that this effect polytope is the convex hull of 146 extreme effects. These include 82 BW effects, 6 CH type effects, 29 effects lying on the hyperplane $\left<\mathbf{x},\pr_{2,\alpha}\right>=0$ and $\left<\mathbf{x},\pr_{2,\alpha}\right>=1$ each.

\subsection{Construction of the Effect Polytope of \texorpdfstring{$\mathbb{H}^{[1]}_{\alpha (2,2)}[\pr_{2,2'}]$}{}}

\label{Appendix::EffectPoly22p}

Next, take the state space $\mathbb{H}^{[2]}_{\alpha (2,2)}[\pr_{2,2'}]$  the state space characterised by the convex hull of $\mathbb{H}^{[1]}_{\alpha (2,2)}[\pr_2]$ and the noisy PR box $\pr_{2,\alpha}'$. The addition of $\pr_{2,\alpha}'$ to $\mathbb{H}^{[1]}_{\alpha (2,2)}[\pr_2]$ introduces two hyperplanes through its effect polytope, given by $\left<\mathbf{x},\pr_{2,\alpha}'\right>=1$ and $\left<\mathbf{x},\pr_{2,\alpha}'\right>=0$. One can perform a similar analysis as shown in the previous section to check which of the extreme effects of the full effect polytope of $\mathbb{H}^{[1]}_{\alpha (2,2)}[\pr_2]$ are still valid effects by ensuring that their inner products with $\pr_{2,\alpha}'$ is in the interval $[0,1]$. We found that the set of Type 1 and Type $1'$ effects from the previous section are extreme and lie on the hyperplane $\left<\mathbf{x},\pr_{2,\alpha}'\right>=0$. The only other class of extreme effect lying on this hyperplane is of the Type 2 form, a candidate of which is
$$
\left(
\begin{array}{cc|cc}
 \frac{2-4 \alpha }{3 \alpha -1} & \frac{2-4 \alpha }{3 \alpha -1} & 0 & 0 \\
 0 & -1 & 0 & \frac{1-\alpha }{3 \alpha -1} \\ \hline
 0 & 0 & 1 & \frac{2-4 \alpha }{1-3 \alpha } \\
 0 & \frac{1-\alpha }{3 \alpha -1} & \frac{2-4 \alpha }{1-3 \alpha } & \frac{2-4 \alpha }{1-3 \alpha }
   \\
\end{array}
\right) = \frac{1-\alpha}{3\alpha-1} e_{\mathrm{CH}_2} + \frac{4\alpha-2}{3\alpha-1}\left(u-\left(
\begin{array}{cc|cc}
 1 & 1 & 0 & 0 \\
 0 & 1 & 0 & 0 \\ \hline
 0 & 0 & 0 & 0 \\
 0 & 0 & 0 & 0 \\
\end{array}
\right)\right);
$$

More precisely, from Table~\ref{Tab:InnerProductOpp}, we can ensure that the Type 1, Type $1'$ and Type 2 effects of  $\mathbb{H}^{[1]}_{\alpha (2,2)}[\pr_2]$ are still valid effects but the effects in Type 3 and Type 4 are not. The extreme effects of $\mathbb{H}^{[1]}_{\alpha (2,2)}[\pr_2]$ that are not lying on the hyperplanes $\left<\mathbf{x},\pr_{2,\alpha}\right>=1$ and $\left<\mathbf{x},\pr_{2,\alpha}\right>=0$ are still valid. To calculate the new effects generated on the hyperplanes $\left<\mathbf{x},\pr_{2,\alpha}'\right>=1$ and $\left<\mathbf{x},\pr_{2,\alpha}'\right>=0$, we can follow the algorithm described in the previous section. Note from Table~~\ref{Tab:InnerProductOpp}, that the Type 1 effects of  $\mathbb{H}^{[1]}_{\alpha (2,2)}[\pr_2]$ lie on the second hyperplane. Upon calculating all the effects lying on this hyperplane and filtering out the extreme effects as described in the previous section, we found that these Type 1 effects are still extreme. The remaining extreme effects are of the Type 2 form. In particular, they can be written as $(1-\alpha)/(3\alpha -1)e_{\rm CH_2} + (4\alpha-2)/(3\alpha-1)(u-f_{V})$ where $f_{V}$ is a Class V effect with $\left<f_{V},\pr_{2'}\right>=1$. Since there are 8 Class V effects that has an inner product of 1 with $\pr_{2'}$, there are 8 corresponding extreme effects of this form. The extreme effects on the hyperplane $\left<\mathbf{x},\pr_{2,\alpha}'\right>=1$ can be calculated as the complements of the extreme effects on the hyperplane $\left<\mathbf{x},\pr_{2,\alpha}'\right>=0$.

\begin{table}[ht]
\centering
\begin{tabular}{|c|c|c|c|c|}
\hline
 $\phantom{\Big(} $&  $\phantom{\Big(} $  Type $1/1'$  $\phantom{\Big(} $ & $\phantom{\Big(} $  Type 2  $\phantom{\Big(} $ & $\phantom{\Big(} $  Type 3  $\phantom{\Big(} $ & $\phantom{\Big(} $  Type 4  $\phantom{\Big(} $ \\
 \hline
 $k(\alpha)$ & $0$  & $\phantom{\Big(}  \frac{1-2 \alpha }{1-3 \alpha }$ & $\frac{1-2 \alpha }{3 \alpha +1}$ & $\frac{1-2\alpha}{1+2\alpha}$ \\
 \hline
 $\phantom{\Big(} k(1/2)$ & $0$  & $0$ &  $0$ & $0$ \\
 \hline
 $\phantom{\Big(}  k(1)$ & $0$  & $1/2$ & $-1/4$ & $-1/3$ \\
 \hline
 \end{tabular}
  \caption{Inner product between extreme vectors $\Tilde{e}$  from the four types and $\pr_{2,\alpha}'$. Here $k(\alpha) = \left<\Tilde{e},\pr_{2,\alpha}'\right>$. Notice that Type 3 and Type 4 effects of $\mathbb{H}^{[1]}_{\alpha (2,2)}[\pr_2]$ are no longer valid effects of $\mathbb{H}^{[1]}_{\alpha (2,2)}[\pr_{22'}]$ because of the negative inner product.}
   \label{Tab:InnerProductOpp}
\end{table}

\subsection{Construction of the Effect Polytope of \texorpdfstring{$\mathbb{H}^{[1]}_{\alpha (2,2)}[\pr_{1,2}]$}{}}

\label{Appendix::EffectPoly12}

\begin{table}[ht]
    \centering
    \begin{tabular}{|c|c|c|c|c|}
    \hline
       $\phantom{\Big(} $ & Type 1 $\cup$ Type $1'$& Type 2 & Type 3 & Type 4  \\
       \hline
        $\phantom{\Big(} \{k(\alpha)\}$ & $\left\{1-\alpha ,\alpha ,\frac{1}{2}\right\}$ & $\left\{\frac{\alpha^2-3\alpha+1}{1-3 \alpha },\frac{\alpha^2+2\alpha-1}{3 \alpha -1}\right\}$ & $\left\{\frac{-\alpha ^2+\alpha +1}{3 \alpha +1},\frac{\alpha ^2+1}{3 \alpha +1}\right\}$ & $\frac{1}{2 \alpha +1}$  \\
        \hline
       $\phantom{\Big(} \{k(1/2)\}$ & $1/2$ & $1/2$ & $1/2$  & $1/2$   \\
       \hline
       $\phantom{\Big(}  \{k(1)\}$ & $\{0,1,1/2\}$ & $\{1,1/2\}$ & $1/4$ & $1/3$  \\
       \hline
    \end{tabular}
    \caption{Inner product between extreme vectors $\Tilde{e}$  from the four types and $\pr_{1,\alpha}$. Here $k(\alpha) = \left<\Tilde{e},\pr_{1,\alpha}\right>$ and the set $\{k(\alpha)\}$ runs over all effects from a given type. Since all the inner products are between 0 and 1 in the range $1/2 \leq \alpha \leq1 $, all extreme effects of $\mathbb{H}^{[1]}_{\alpha (2,2)}[\pr_2]$ are also effects of $\mathbb{H}^{[1]}_{\alpha (2,2)}[\pr_{12}]$.}
    \label{Tab:InnerProductNonOpp}
\end{table}

Next, let us consider the second state space $\mathbb{H}^{[2]}_{\alpha (2,2)}[\pr_{12}]$ where the two noisy PR boxes are not isotropically opposite to each other. Following the previous analysis we construct Table~\ref{Tab:InnerProductNonOpp} to check the inner product between the extreme effects of $\mathbb{H}^{[1]}_{\alpha (2,2)}[\pr_{2}]$ lying on the hyperplane $\left< \mathbf{x}, \mathrm{PR}_{2,\alpha} \right>=1$ and $\pr_{1,\alpha}$. From this table it is clear that all the extreme effects of $\mathbb{H}^{[1]}_{\alpha (2,2)}[\pr_{2}]$ have an inner product between 0 and 1 in the range $1/2 \leq \alpha \leq 1$ and therefore are valid effects of $\mathbb{H}^{[2]}_{\alpha (2,2)}[\pr_{12}]$ and in fact extreme and similarly the complementary effects. However the effects of the local polytope $e_{\rm CH_1}$ and $e_{\rm CH_{1'}}$ are no longer valid. One can use the algorithm from Subsection~\ref{Appendix::EffectPoly} to find that effects of the form Type 1,2,3 and 4 are new extreme effects on the hyperplane $\left<\mathbf{x},\pr_{2,\alpha}'\right>=1$ and their complimentary effects on the hyperplane $\left<\mathbf{x},\pr_{2,\alpha}'\right>=0$. The effect polytope of $\mathbb{H}^{[2]}_{\alpha (2,2)}[\pr_{12}]$ can be calculated by separately constructing the effect polytopes of $\mathbb{H}^{[1]}_{\alpha (2,2)}[\pr_{1}]$ and $\mathbb{H}^{[1]}_{\alpha (2,2)}[\pr_{2}]$, taking their union and then discarding the extreme effects whose inner products with either $\pr_1$ or $\pr_2$ is outside the interval $[0,1]$. In particular, these effects are  $e_{\rm CH_1},e_{\rm CH_1}',e_{\rm CH_2}$ and $e_{\rm CH_2}'$.

\section{Pure Couplers of \texorpdfstring{$\mathbb{H}^{[1]}_{(3,2)}[{\mathrm{N_1}}]$}{}}
\label{Appendix:PureCouplersN13322}

A representative of each class of minimally 2-preserving pure coupler is shown below. Subscripts denote the number of elements in each class. The full set of 88 can be generated by applying all the relabelling symmetries, then removing elements that are not minimally 2-preserving, then removing elements that are not pure couplers and deleting any duplicates.

$$
\frac{1}{3}\left(
\begin{array}{cc|cc|cc}
 1 & 0 & 0 & 0 & 1 & 0 \\
 0 & 0 & 0 & 1 & 0 & 0 \\ \hline
 0 & 0 & 0 & 0 & 0 & 0 \\
 0 & 1 & 0 & -1 & -1 & 0 \\ \hline
 1 & 0 & 0 & 0 & 0 & 0 \\
 0 & 0 & 0 & 0 & 1 & 0 \\
\end{array}
\right)_8,
\frac{2}{3}\left(
\begin{array}{cc|cc|cc}
 1 & 0 & 0 & 0 & 0 & 0 \\
 0 & 0 & 0 & 1 & 0 & 0 \\ \hline
 0 & 0 & 0 & 0 & 0 & 0 \\
 0 & 1 & 0 & -1 & 0 & 0 \\ \hline
 0 & 0 & 0 & 0 & 0 & 0 \\
 0 & 0 & 0 & 0 & 0 & 0 \\
\end{array}
\right)_1, \frac{1}{3}\left(
{\begin{array}{cc|cc|cc}
 1 & 0 & 0 & 0 & 1 & 0 \\
 0 & 0 & 0 & 1 & 0 & 0 \\ \hline
 0 & 0 & 0 & 0 & 0 & 0 \\
 0 & 1 & 0 & -1 & 1 & 0 \\ \hline
 0 & 0 & 0 & 1 & 0 & 0 \\
 0 & 0 & 0 & 0 & 0 & 0 \\
\end{array}}
\right)_8
$$
$$
\frac{1}{3}\left(
\begin{array}{cc|cc|cc}
 1 & 0 & 0 & 0 & 0 & 0 \\
 0 & 0 & 0 & 1 & 0 & 0 \\ \hline
 0 & 0 & 0 & 0 & 0 & 0 \\
 0 & 1 & 0 & -1 & 1 & 0 \\ \hline
 1 & 0 & 0 & 1 & 0 & 0 \\
 0 & 0 & 0 & 0 & 1 & 0 \\
\end{array}
\right)_8, \frac{1}{3}\left(
\begin{array}{cc|cc|cc}
 1 & 0 & 0 & 0 & 0 & 0 \\
 0 & 0 & 0 & 1 & 0 & 0 \\  \hline
 0 & 0 & 0 & 0 & 0 & 0 \\
 0 & 1 & 0 & -1 & 1 & 0 \\ \hline
 0 & 1 & 0 & -1 & 0 & 1 \\
 0 & 0 & 0 & 0 & 1 & 0 \\
\end{array}
\right)_{8},\frac{1}{3}\left(
{\begin{array}{cc|cc|cc}
 2 & 0 & 0 & 0 & 0 & 0 \\
 0 & 0 & 0 & 2 & 0 & 0 \\\hline
 0 & 0 & 0 & 0 & 0 & 0 \\
 0 & 2 & 0 & -2 & 1 & 0 \\\hline
 0 & 0 & 0 & 0 & 0 & 0 \\
 0 & 0 & 0 & 0 & 0 & 0 \\
\end{array}}
\right)_{8}
$$
$$
\frac{1}{3}\left(
\begin{array}{cc|cc|cc}
 2 & 0 & 0 & 0 & 0 & 0 \\
 0 & 0 & 0 & 2 & 0 & 0 \\ \hline
 0 & 0 & 0 & 0 & 0 & 0 \\
 0 & 2 & 0 & -2 & 0 & 0 \\ \hline
 0 & 0 & 0 & 0 & 1 & 0 \\
 0 & 0 & 0 & 0 & 0 & 0 \\
\end{array}
\right)_{3}, \frac{1}{3}\left(
\begin{array}{cc|cc|cc}
 2 & 0 & 0 & 0 & 0 & 0 \\
 0 & 0 & 0 & 2 & 0 & 0 \\ \hline
 0 & 0 & 0 & 0 & 0 & 0 \\
 0 & 2 & 0 & -2 & 0 & 0 \\ \hline
 0 & 1 & 0 & 0 & 0 & 0 \\
 0 & 0 & 0 & 0 & 1 & 0 \\
\end{array}
\right)_{8}, \frac{1}{3}\left(
\begin{array}{cc|cc|cc}
 2 & 0 & 0 & 0 & 1 & 0 \\
 0 & 0 & 0 & 2 & 1 & 0 \\ \hline
 0 & 0 & 0 & 0 & 0 & 0 \\
 0 & 2 & 0 & -2 & 0 & 0 \\ \hline
 0 & 0 & 0 & 0 & 0 & 0 \\
 0 & 0 & 0 & 0 & 0 & 0 \\
\end{array}
\right)_{2}
$$
$$
\frac{1}{3}\left(
\begin{array}{cc|cc|cc}
 2 & 0 & 0 & 0 & 1 & 0 \\
 0 & 0 & 0 & 2 & 1 & 0 \\ \hline
 0 & 0 & 0 & 0 & 0 & 0 \\
 0 & 2 & 0 & -2 & 0 & 0 \\\hline
 0 & 0 & 0 & 0 & 0 & 1 \\
 0 & 0 & 0 & 0 & 0 & 0 \\
\end{array}
\right)_{1}, \frac{1}{6}\left(
{\begin{array}{cc|cc|cc}
 4 & 0 & 0 & 0 & 1 & 0 \\
 0 & 0 & 0 & 4 & 0 & 0 \\ \hline
 0 & 0 & 0 & 0 & 1 & 0 \\
 0 & 4 & 0 & -4 & 0 & 0 \\ \hline
 0 & -1 & 0 & 0 & 0 & 1 \\
 0 & 0 & 0 & 0 & 0 & 0 \\
\end{array}}
\right)_{8}, \frac{1}{6}\left(
\begin{array}{cc|cc|cc}
 4 & 0 & 0 & 0 & 1 & 0 \\
 0 & 0 & 0 & 4 & 0 & 0 \\ \hline
 0 & 0 & 0 & 0 & 0 & 0 \\
 0 & 4 & 0 & -4 & 0 & 0 \\ \hline
 1 & 0 & 0 & -1 & 0 & 1 \\
 0 & 0 & 0 & 0 & 1 & 0 \\
\end{array}
\right)_{8}
$$
$$
\frac{1}{6}\left(
{\begin{array}{cc|cc|cc}
 4 & 0 & 0 & 0 & 0 & 0 \\
 0 & 0 & 0 & 4 & 1 & 0 \\ \hline
 0 & 0 & 0 & 0 & 1 & 0 \\
 0 & 4 & 0 & -4 & 0 & 0 \\ \hline
 0 & 0 & 1 & 0 & 0 & 0 \\ 
 0 & 0 & 0 & 0 & 0 & 0 \\
\end{array}}
\right)_{8},
\frac{1}{6}\left(
{\begin{array}{cc|cc|cc}
 4 & 0 & 0 & 0 & 0 & 0 \\
 0 & 0 & 0 & 4 & 1 & 0 \\ \hline
 0 & 0 & 0 & 0 & 0 & 0 \\
 0 & 4 & 0 & -4 & 0 & 0 \\ \hline
 1 & 0 & 1 & 0 & 0 & 0 \\
 0 & 0 & 0 & 0 & 1 & 0 \\
\end{array}}
\right)_{8},
\frac{1}{3}\left(
\begin{array}{cc|cc|cc}
 2 & 0 & 0 & 0 & 0 & 0 \\
 0 & 0 & 0 & 2 & 0 & 0 \\ \hline
 0 & 0 & 0 & 0 & 0 & 0 \\
 0 & 2 & 0 & -2 & 0 & 0 \\ \hline
 0 & 0 & 0 & 0 & 0 & 1 \\
 0 & 0 & 0 & 0 & 1 & 0 \\
\end{array}
\right)_{1}
$$

\section{Equivalence Classes of \texorpdfstring{$\npr$}{\npr} PR boxes}
\label{Appendix::EquivClasses}

\begin{table}[H]
    \centering
    \begin{tabular}{|c|c|c|c|c|}
    \hline
        \multicolumn{5}{|c|}{$\npr=2$}\\
        \hline
        Class & Description & Example & Party Symmetric & Count \\ 
        \hline \hline
        1 & Isotropically opposite pairs & $\{ \pr_1,\pr_1' \}$  & Yes & $4$ \\
        \hline
       2 & 1 party symmetric,&&&\\
       &1 party asymmetric& $\{ \pr_1,\pr_3 \}$  & No & $16$ \\
        \hline
       3 & Party symmetric or &&&\\
       &asymmetric that are&&&\\
       &not isotropically opposite&$\{ \pr_1,\pr_2 \}$  & Yes & $8$ \\
        \hline
    \end{tabular}
    \caption{$\npr=2$}
    \label{Tab::TwoRoofs}
\end{table}

\begin{table}[H]
    \centering
    \begin{tabular}{|c|c|c|c|c|}
    \hline
        \multicolumn{5}{|c|}{$\npr=3$}\\
        \hline
        Class & Description &   Example & Party Symmetric & Count \\ 
        \hline \hline
        1 &  Two PR boxes isotropically&&&\\
        &opposite to each other. If they are&&&\\
        &party symmetric, the third is not&&&\\
        &and vice versa & $\{ \pr_1,\pr_1',\pr_3 \}$  & No & $16$ \\
        \hline
        2 &  Two PR boxes that are not&&&\\
        &isotropically opposite. If &&&\\
        &these two are party symmetric&&&\\
        &the third is party asymmetric&&&\\
        &and vice versa. &  $\{ \pr_3,\pr_4', \pr_1 \}$  & Yes & $32$ \\
        \hline
       3 &  Either all party symmetric&&&\\
       &or party asymmetric. &  $\{ \pr_1,\pr_1', \pr_2 \}$  & Yes & $8$ \\
        \hline
    \end{tabular}
    \caption{$\npr=3$}
    \label{Tab::ThreeRoofs}
\end{table}

\begin{table}[H]
    \centering
    \begin{tabular}{|c|c|c|c|c|}
    \hline
        \multicolumn{5}{|c|}{$\npr=4$}\\
        \hline
        Class & Description &  Example & Party Symmetric & Count \\ 
        \hline \hline
        1 &   Two isotropically opposite pairs. &&&\\
        &1 pair party symmetric and 1 &&&\\
        &pair party asymmetric  &  $\{ \pr_1,\pr_1',\pr_3, \pr_3' \}$  & No & $4$ \\
        \hline
        2 &  Three party asymmetric with one &&&\\
        &party symmetric/three party &&&\\
        &symmetric with one party&&&\\
&asymmetric &  $\{ \pr_3,\pr_4,\pr_4', \pr_1 \}$  & No & $32$ \\
        \hline
       3 &     1 pair of party symmetric PR &&&\\
       &boxes and 1 pair of party &&&\\
       &asymmetric PR boxes. 1 pair&&&\\
       &isotropically opposite and 1&&&\\
       &1 pair isotropically not opposite. &  $\{ \pr_1,\pr_1', \pr_3, \pr_4' \}$  & Yes & $16$ \\
        \hline
        4 &  1 pair of party symmetric and 1&&&\\
        &pair of party asymmetric. Pairs differ by&&&\\
        &the same diagonal block. & \centering $\{ \pr_1,\pr_2',\pr_3, \pr_4' \}$  & Yes & $16$ \\
        \hline
        5 &  1 pair of party symmetric and 1 &&&\\
        &pair of party asymmetric. Pairs &&&\\
        &differ in different diagonal blocks. & \centering $\{ \pr_1,\pr_2,\pr_3, \pr_4' \}$  & Yes & $8$ \\
        \hline
       6 &  All party symmetric/asymmetric &  $\{ \pr_1,\pr_1', \pr_2, \pr_2' \}$  & Yes & $2$ \\
        \hline
    \end{tabular}
    \caption{$\npr=4$}
    \label{Tab::FourRoofs}
\end{table}

\begin{table}[H]
    \centering
    \begin{tabular}{|c|c|c|c|c|}
    \hline
        \multicolumn{5}{|c|}{$\npr=5$}\\
        \hline
        Class & Description &  Example & Party Symmetric & Count \\ 
        \hline \hline
        1 & 1 isotropically opposite pair party&&&\\
        &symmetric/asymmetric pair with &&&\\
        &three party asymmetric/symmetric &  $\{ \pr_1,\pr_1',\pr_3 \pr_3', \pr_4\}$  & No & $16$ \\
        \hline
        2 &  1 isotropically non-opposite pair&&&\\
        &party symmetric/asymmetric pair with &&&\\
        &three party asymmetric/symmetric &  $\{ \pr_1,\pr_2,\pr_2', \pr_3,\pr_4' \}$  & Yes & $32$ \\
        \hline
       3 &  All party symmetric/asymmetric&&&\\
       &and one party asymmetric/symmetric. & $\{ \pr_1, \pr_3, \pr_3', \pr_4, \pr_4' \}$  & Yes & $8$ \\
        \hline
    \end{tabular}
    \caption{$\npr=5$}
    \label{Tab::FiveRoofs}
\end{table}

\begin{table}[H]
    \centering
    \begin{tabular}{|c|c|c|c|c|}
    \hline
        \multicolumn{5}{|c|}{$\npr=6$}\\
        \hline
        Class & Description &  Example & Party Symmetric & Count \\ 
        \hline \hline
        1 &  Three party symmetric PR boxes&&&\\
        &with three party asymmetric PR&&&\\
        &boxes &  $\{ \pr_1,\pr_1', \pr_2, \pr_3 \pr_3', \pr_4\}$  & No & $16$ \\
        \hline
        2 &  1 isotropically opposite party&&&\\
        &symmetric/asymmetric pair with&&&\\
        &four party asymmetric/symmetric &  $\{ \pr_1, \pr_1', \pr_3,\pr_3',\pr_4, \pr_4' \}$  & Yes & $4$ \\
        \hline
       3 &  4 party asymmetric/symmetric PR&&&\\
       &boxes with two party symmetric/&&&\\
       &asymmetric PR boxes that are &&&\\
       &not isotropically opposite.&  $\{ \pr_1, \pr_2, \pr_3,\pr_3',\pr_4, \pr_4' \}$  & Yes & $8$ \\
        \hline
    \end{tabular}
    \caption{$\npr=6$}
    \label{Tab::SixRoofs}
\end{table}

\begin{table}[H]
    \centering
    \begin{tabular}{|c|c|c|c|c|}
    \hline
        \multicolumn{5}{|c|}{$\npr=7$}\\
        \hline 
        Class & Description &  Example & Party Symmetric & Count \\ 
        \hline \hline
         1 & All but one PR box &  $\{ \pr_1, \pr_2, \pr_2', \pr_3,\pr_3',\pr_4, \pr_4' \}$  & Yes & $8$ \\
        \hline
    \end{tabular}
    \caption{$\npr=7$}
    \label{Tab::SevenRoofs}
\end{table}

\section{Proof of Theorem~\ref{Theorem::NoCouplers}}
\label{Appendix::NoCouplers}

We show that $\tilde{e}$ cannot be both a coupler and weakly minimally $2$-preserving. This implies the statement in Theorem~\ref{Theorem::NoCouplers}.
 
 Let $\mathrm{Extreme}[\mathcal{E}_{\mathbb{H}^{[\npr]}_{\alpha(2,2)}} ]$ be the set of extremal effects of the effect polytope $\mathcal{E}_{\mathbb{H}^{[\npr]}_{\alpha(2,2)}}$ and $n$ denote the cardinality of $\mathrm{Extreme}[\mathcal{E}_{\mathbb{H}^{[\npr]}_{\alpha(2,2)}} ]$. Let us denote by $\underline{\mathds{O}}$ and $\underline{\mathds{1}}$ the vectors $\left(0 \ 0\  \ldots 0\right)_{1 \times n}$ and  $\left(1\ 1\  \ldots 1\right)_{1 \times n}$ respectively. Since the effects space is convex, any effect $e$ can be expressed as

\begin{equation}
    e = \sum_{j=1}^n \mathrm{x}_j e_j
    \label{Eq::EffectOrdering}
\end{equation}
where $e_j$ is an extremal effect and $\mathrm{x}_j \in [0,1]$ such that $\sum_{j}^n \mathrm{x}_j=1$. For $e$ to be weakly minimally $2$-preserving, we require that for any extremal effect $e_j$ and a pair of PR boxes $\pr_{k,\alpha}$  and $\pr_{l,\alpha}$ the inner product between $e_j$ and the sub-normalised state $\Phi_{e}\left(k,l\right)$ is non-negative. Since for every extremal effect $e_j$, the effect $u-e_j$ is also an extremal effect, this condition also implies that the above inner product cannot be more than 1. In other words, for an arbitrary pair of noisy PR boxes indexed by $(k,l)$, one requires that if $\underline{\mathbf{x}} \in \mathbb{R}^n_{\geq 0}$ represents the convex support of the effect $e$, then every entry of the vector,
\begin{equation}
    M_{k,l}.\underline{\mathbf{x}}^T\coloneqq  \begin{pmatrix}
        \left<e_1,\Phi_{e_1}\left(k,l\right)\right> && \left<e_1,\Phi_{e_2}\left(k,l\right)\right> && \ldots && \left<e_1,\Phi_{e_n}\left(k,l\right)\right>\\
        \left<e_2,\Phi_{e_1}\left(k,l\right)\right> && \left<e_2,\Phi_{e_2}\left(k,l\right)\right> && \ldots && \left<e_2,\Phi_{e_n}\left(k,l\right)\right>\\
        \vdots && \vdots && \ddots && \vdots \\
        \left<e_n,\Phi_{e_1}\left(k,l\right)\right> && \left<e_n,\Phi_{e_2}\left(k,l\right)\right> && \ldots && \left<e_n,\Phi_{e_n}\left(k,l\right)\right> \\
    \end{pmatrix}_{n \times n}.\begin{pmatrix}
        \mathrm{x_1}\\
        \mathrm{x_2}\\
        \vdots \\
        \mathrm{x_n}\\
    \end{pmatrix}_{n \times 1}, \label{expression::Mmatrix}
\end{equation}
must be non-negative. This can be viewed as the constraint:
\begin{equation}
    -M_{k,l}.\underline{\mathbf{x}}^T \leq \underline{\mathds{O}}^T  
\end{equation}
Additionally, since the vector $\underline{\mathbf{x}}$ represents the convex weights, one also needs the following convexity conditions to hold:
\begin{equation}
    \underline{\mathds{1}}.\underline{\mathbf{x}}^T \leq 1 \quad \text{and }  \quad -\underline{\mathds{1}}.\underline{\mathbf{x}}^T \leq -1.
\end{equation}
With this one can define a \textit{constraint} matrix \textbf{C}  and a \textit{bound} vector \underline{\textbf{b}} as:
\setcounter{MaxMatrixCols}{20}
\begin{equation}
    \mathbf{C} \coloneqq 
        \begin{pmatrix}
        \underline{\mathds{1}} \\ -\underline{\mathds{1}} \\ -M_{1,1} \\ -M_{1,2} \\ \vdots \\ -M_{\npr,\npr-1} \\ -M_{\npr,\npr}\\
        \end{pmatrix}_{(\npr^2n+2) \times n}
  \text{and } \quad \quad \underline{\mathbf{b}} \coloneqq \begin{pmatrix}
        1 \\ -1 \\ \underline{\mathds{O}}^T \\ \vdots \\   \underline{\mathds{O}}^T\\
    \end{pmatrix}_{(\npr^2n+2) \times 1}
    \label{expression::C}
    \end{equation}
respectively. The effect $e$, if weakly minimally $2$- preserving, will satisfy $\mathbf{C}.\underline{\mathbf{x}}^T \leq \underline{\mathbf{b}}$. Next, since we are interested in finding weakly minimally $2$-preserving couplers, we would also like $e$ to satisfy 

\begin{equation}
    \mathrm{CHSH}_{i}\left[\Tilde{\Phi}^{(2,3)}_{e}(\pr_{k,\alpha},\pr_{l,\alpha})\right] > \frac{3}{4} 
\end{equation}
 where ${\mathrm{ CHSH}}_i$ is a CHSH game that can be won by an amount more than 3/4 by correlations obtained upon performing the fiducial measurements on the allowed noisy PR boxes and $\Tilde{\Phi}^{(2,3)}_{e}(\pr_{k,\alpha},\pr_{l,\alpha})$ is the normalised state $\Phi^{(2,3)}_{e}(\pr_{k,\alpha},\pr_{l,\alpha})$. This is equivalent to requiring
 \begin{equation}
 \begin{split}
     &{\mathrm{CHSH}_i}\left[\Phi^{(2,3)}_{e}(\pr_{k,\alpha},\pr_{l,\alpha})\right] > \frac{3}{4}\left<u,\Phi^{(2,3)}_{e}(\pr_{k,\alpha},\pr_{l,\alpha})\right> \\
      \implies &{\mathrm{CHSH}_i}\left[\Phi^{(2,3)}_{e}(\pr_{k,\alpha},\pr_{l,\alpha})\right] - \frac{3}{4}\left<u,\Phi^{(2,3)}_{e}(\pr_{k,\alpha},\pr_{l,\alpha})\right>   >0\\
      \implies &\left<\mathrm{CHSH}_i-\frac{3}{4}u,\Phi^{(2,3)}_{e}(\pr_{k,\alpha},\pr_{l,\alpha})\right> > 0 \\
     \implies &\left<\mathrm{CHSH}_i-\frac{3}{4}u,\id \otimes \sum_j x_j e_j \otimes \id (\pr_{k,\alpha},\pr_{l,\alpha})\right> > 0 \\
      \implies &  \left(
          \left<\mathrm{CHSH}_i-\frac{3}{4}u, \Phi_{e_1}^{(2,3)}(k,l) \right>, \ldots  ,\left<\mathrm{CHSH}_i-\frac{3}{4}u, \Phi_{e_n}^{(2,3)}(k,l) \right>
      \right).\underline{\textbf x}^T  \eqqcolon \underline{\mathbf{f}}_{k,l|i}.\underline{\mathbf{x}}^T  > 0 \\
 \end{split}    
\end{equation}

\noindent
We do not know whether there exists any effect at all such that for the choice of $k,l$ and $i$, $\underline{\mathbf{f}}_{k,l|i}.\underline{\mathbf{x}}^T > 0$. Therefore, one can alternatively look for a vector $\underline{\mathbf{x}}$ which maximises the value $\underline{\mathbf{f}}_{k,l|i}.\underline{\mathbf{x}}^T $. This can be done with the help of a Linear Program (LP) defined below:
\begin{equation}
\mathbb{P}_{k,l|i} \coloneqq \quad
\begin{aligned}
& \underset{\underline{\mathbf{x}} \in \mathbb{R}^n}{\text{maximise:}}
& &\quad \underline{\mathbf{f}}_{k,l|i}.\underline{\mathbf{x}}^T\\
& \text{subject to:}& & \quad \mathbf{C}.\underline{\mathbf{x}}^T \leq \underline{\mathbf{b}} \\
& & & \quad \underline{\mathbf{x}}  \geq 0  \\
\end{aligned}
\label{Eq::primal}
\end{equation}
The dual program is defined as:
\begin{equation}
\mathbb{D}_{k,l|i} \coloneqq \quad
\begin{aligned}
& \underset{\underline{\mathbf{y}} \in \mathbb{R}^{|\underline{\mathbf{b}}|}}{\text{minimise:}}
& &\quad \underline{\mathbf{b}}^T.\underline{\mathbf{y}} \\
& \text{subject to:}& & \quad \mathbf{C}^T.\underline{\mathbf{y}} \geq \underline{\mathbf{f}}_{k,l|i} \\
& & & \quad \underline{\mathbf{y}}  \geq 0  \\
\end{aligned}.
\label{Eq::dual}
\end{equation}
To prove that a weakly minimally $2$-preserving coupler exists, it suffices to show that for at least one choice of $k',l'$ and $i'$, 
\begin{equation}
    \mathbb{P}_{k',l'|i'}>0. 
\end{equation}
Since when $\alpha \leq 1/\sqrt{2}$, no extremal effects of party symmetric state spaces are couplers, we only need to evaluate these LP pairs in the range $1/\sqrt{2} < \alpha \leq 1$. 

To get the analytic solution to these primal and dual problems we proceed as follows. For each case, we considered the effect polytope $\mathcal{E}_{\mathbb{H}^{[\npr]}_{\alpha (2,2)}}$ where we run through a discrete set of values of $\alpha$ between $22/30>1/\sqrt{2}$ and $1$ with a step-size $1/30$. For every step we have then solved the primal and dual problem pairs and used the solutions to make a guess of the analytic forms in terms of $\alpha$. We have then checked that these pair of guess vectors satisfy all of the analytic constraints for their respective problem and give the same optimal value for all $\alpha\in[1/\sqrt{2},1]$, confirming that we have found the optimum, since this shows $\mathbb{P}_{k',l'|i'}= \mathbb{D}_{k',l'|i'}$. [Note that when $\alpha=1/\sqrt{2}$, the solution to both problems is zero in every case, as expected.]

\subsection{\texorpdfstring{$\npr=2$}{m=2}}

There are 3 equivalent local relabelling classes of state spaces with 2 PR boxes. Out of these, party symmetric state spaces exist only in Class 1 and Class 3.  However, since the PR boxes in Class 1 are isotropically opposite pairs, there are no couplers for this state space (see Section~\ref{Subsection::2Roofs}) and one therefore only needs to check for couplers for a state space in Class 2. In Table~\ref{Table:2RoofsClass3}, we focus on $\mathrm{CHSH}_1$ scores.  The set of PR box pairs such that there exists an extremal effect in $\mathcal{E}_{\mathbb{H}^{[2]}_{\alpha (2,2)}[\pr_{2,3}]}$ for which a score of more than $3/4$ can be achieved in the $\mathrm{CHSH}_1$ game are:
$$
 \{(\pr_{1,\alpha},\pr_{1,\alpha}),(\pr_{1,\alpha},\pr_{2,\alpha}),(\pr_{2,\alpha},\pr_{1,\alpha})\}.
$$
In the table below and all following tables, we will only consider pairs of noisy PR boxes for which such violations are possible using extremal effects. In addition we directly provide the effects as a convex combination of extremal effects and the vectors $\mathbf{C}^T.\mathbf{y}$. From this data one can construct the optimising vectors $\mathbf{x}$ and $\mathbf{y}$.

To obtain the corresponding vector $\mathbf{y}$, first note from~\eqref{expression::C} that $\mathbf{C}$ can be seen as a column vector with each entry being a matrix. The first two entries are row matrices containing all 1s and have dimension $1 \times n$.  The entries labelled $-M_{k,l}$  are $n \times n$ matrices. Therefore, $\mathbf{C}^T$ is an $n \times (2+ n\npr^2)$ dimensional matrix in which the first two columns are all 1s, followed by an $n \times n $ block containing $-M^T_{1,1}$, and so on. Note that $(M^T_{k,l})_{i,j}=\langle e_j,\Phi_{e_i}(k,l)\rangle$, and we write $\left[e_j\right]_{{-M_{k,l}}}= - \left(\langle e_j,\Phi_{e_1}(k,l)\rangle, \langle e_j,\Phi_{e_2}(k,l)\rangle, \ldots, \langle e_j,\Phi_{e_n}(k,l)\rangle \right)$. For $\mathbf{C}^T.\mathbf{y}=a\left[e_{j}\right]_{{-M_{k,l}}}$, the vector $\mathbf{y}$ has a single non-zero entry, $a$, at the row $2+n(k-1)\npr+(l-1)n+j$.

For instance, the effect in the first row of Table~\ref{Table:2RoofsClass3} is $\theta e_{\mathrm{CH}_{1,\alpha}}+ (1-\theta) u$. The corresponding vector $\mathbf{x}$ has all zero entries, with the exception of $\theta$ at position $j$ such that $e_j=e_{\mathrm{CH_{1,\alpha}}}$ in~\eqref{Eq::EffectOrdering}, and $(1-\theta)$ at the index $j'$ such that $e_{j'}=u$ in~\eqref{Eq::EffectOrdering}. The respective entry under $\mathbf{C}^T.\mathbf{y}$ is $\frac{1+2\alpha}{4}\left[e_{\mathrm{CH_{1,\alpha}}}\right]_{{-M_{2,2}}}$. This corresponds to taking the aforementioned $j$ and $k=l=2$.

 This presentation style has been chosen to compress the data. Additionally, in the tables below we take
$$
\theta \coloneq \frac{3 \alpha  (\alpha +1)}{4 \alpha  (\alpha +1)-2} \text{ and } \theta':=\frac{2}{2 \alpha ^2+1}.
$$

\begin{table}[H]
    \centering
\begin{tabular}{|c||c|c|}
    \hline
        Pairs &  Effects & $\phantom{\Big(} \mathbf{C}^T.\mathbf{y} \phantom{\Big(} $  \\
        \hline \hline
       $\phantom{\Big(}\left(\pr_{1,\alpha},\pr_{1,\alpha}\right) \phantom{\Big(}$ & $\phantom{\Big(}\theta e_{\mathrm{CH}_{1,\alpha}}+ (1-\theta) u \phantom{\Big(}$ & $\phantom{\Big(} \frac{1+2\alpha}{4}\left[e_{\mathrm{CH_{1,\alpha}}}\right]_{{-M_{2,2}}}\phantom{\Big(} $ \\
        \hline
        $\phantom{\Big(}\left(\pr_{1,\alpha},\pr_{2,\alpha}\right) \phantom{\Big(}$  & $\theta e_{\mathrm{CH}_{2,\alpha}}+ (1-\theta) u $ & $\frac{1+2\alpha}{4}\left[e_{\mathrm{CH_{2,\alpha}}}\right]_{{-M_{1,1}}}$ \\
        \hline
        $\phantom{\Big(}\left(\pr_{2,\alpha},\pr_{1,\alpha}\right) \phantom{\Big(}$  & $\theta e_{\mathrm{CH}_{2,\alpha}}+ (1-\theta) u $ & $\frac{1+2\alpha}{4}\left[e_{\mathrm{CH_{2,\alpha}}}\right]_{{-M_{1,1}}}$ \\
        \hline
    \end{tabular}
    \caption{${\mathrm{ CHSH}_1}$ ($\npr=2$ Class 3)}
    \label{Table:2RoofsClass3}
\end{table}

\subsection{\texorpdfstring{$\npr=3$}{\npr=3}}

\subsubsection{Class 2}
\begin{table}[H]
    \centering
    \begin{tabular}{|c||c|c|}
    \hline
        Pairs &  Effects & $\phantom{\Big(} \mathbf{C}^T.\mathbf{y} \phantom{\Big(} $  \\
        \hline \hline
       $\phantom{\Big(}\left(\pr_{1,\alpha},\pr_{1,\alpha}\right) \phantom{\Big(}$ & $\phantom{\Big(}\theta e_{\mathrm{CH}_{1,\alpha}}+ (1-\theta) u \phantom{\Big(}$ & $\phantom{\Big(} \frac{1}{2}\left[e_{\mathrm{CH_{2}}}'\right]_{{-M_{3,4'}}}\phantom{\Big(} $ \\
        \hline
        $\phantom{\Big(}\left(\pr_{1,\alpha},\pr_{3,\alpha}\right) \phantom{\Big(}$  & $\theta e_{\mathrm{CH}_{4,\alpha}}'+ (1-\theta) u $ & $\frac{1}{2}\left[e_{\mathrm{CH_{2}}}'\right]_{{-M_{1,4'}}}$ \\
        \hline
        $\phantom{\Big(}\left(\pr_{1,\alpha},\pr_{4,\alpha}'\right) \phantom{\Big(}$  & $\theta e_{\mathrm{CH}_{3,\alpha}}+ (1-\theta) u $ & $\frac{1}{2}\left[e_{\mathrm{CH_{2}}}\right]_{{-M_{1,3}}}$ \\
        \hline
        $\phantom{\Big(}\left(\pr_{3,\alpha},\pr_{1,\alpha}\right) \phantom{\Big(}$ & $\phantom{\Big(}\theta e_{\mathrm{CH}_{4,\alpha}}'+ (1-\theta) u \phantom{\Big(}$ & $\phantom{\Big(} \frac{1}{2}\left[e_{\mathrm{CH_{2}}}'\right]_{{-M_{1,4'}}}\phantom{\Big(} $ \\
        \hline
        $\phantom{\Big(}\left(\pr_{3,\alpha},\pr_{3,\alpha}\right) \phantom{\Big(}$  & $\theta e_{\mathrm{CH}_{1,\alpha}}+ (1-\theta) u $ & $\frac{1}{2}\left[e_{\mathrm{CH_{2}}}'\right]_{{-M_{3,4'}}}$ \\
        \hline
       $\phantom{\Big(}\left(\pr_{3,\alpha},\pr_{4,\alpha}'\right) \phantom{\Big(}$  & $\theta' e_{\mathrm{CH}_{2}}'+ (1-\theta') u $ & $\frac{1}{2}\left[e_{\mathrm{CH_{2}}}'\right]_{{-M_{1,1}}}$ \\
        \hline
        $\phantom{\Big(}\left(\pr_{4,\alpha}',\pr_{1,\alpha}\right) \phantom{\Big(}$ & $\phantom{\Big(}\theta e_{\mathrm{CH}_{3,\alpha}}+ (1-\theta) u \phantom{\Big(}$ & $\phantom{\Big(} \frac{1}{2}\left[e_{\mathrm{CH_{2}}}\right]_{{-M_{1,3}}}\phantom{\Big(} $ \\
\hline
    \end{tabular}
    \caption{${\mathrm{ CHSH}_1}$ ($\npr=3$ Class 2)}
\end{table}

\begin{table}[H]
    \centering
    \begin{tabular}{|c||c|c|}
    \hline
        Pairs &  Effects & $\phantom{\Big(} \mathbf{C}^T.\mathbf{y} \phantom{\Big(} $  \\
        \hline \hline
       $\phantom{\Big(}\left(\pr_{1,\alpha},\pr_{1,\alpha}\right) \phantom{\Big(}$ & $\phantom{\Big(}\theta e_{\mathrm{CH}_{4,\alpha}}'+ (1-\theta) u \phantom{\Big(}$ & $\phantom{\Big(} \frac{1+2\alpha}{4}\left[e_{\mathrm{CH_{2}}}'\right]_{{-M_{1,4'}}}\phantom{\Big(} $ \\
        \hline
        $\phantom{\Big(}\left(\pr_{1,\alpha},\pr_{3,\alpha}\right) \phantom{\Big(}$  & $\theta e_{\mathrm{CH}_{1,\alpha}}+ (1-\theta) u $ & $\frac{1+2\alpha}{4}\left[e_{\mathrm{CH_{2}}}'\right]_{{-M_{3,4'}}}$ \\
        \hline
        $\phantom{\Big(}\left(\pr_{1,\alpha},\pr_{4,\alpha}'\right) \phantom{\Big(}$  & $\theta' e_{\mathrm{CH}_{2}}'+ (1-\theta') u $ & $\frac{1+2\alpha}{4}\left[e_{\mathrm{CH_{2}}}'\right]_{{-M_{1,1}}}$ \\
        \hline
        $\phantom{\Big(}\left(\pr_{3,\alpha},\pr_{1,\alpha}\right) \phantom{\Big(}$ & $\phantom{\Big(}\theta e_{\mathrm{CH}_{1,\alpha}}+ (1-\theta) u \phantom{\Big(}$ & $\phantom{\Big(} \frac{1+2\alpha}{4}\left[e_{\mathrm{CH_{2}}}'\right]_{{-M_{3,4'}}}\phantom{\Big(} $ \\
        \hline
        $\phantom{\Big(}\left(\pr_{3,\alpha},\pr_{3,\alpha}\right) \phantom{\Big(}$  & $\theta e_{\mathrm{CH}_{4,\alpha}}'+ (1-\theta) u $ & $\frac{1+2\alpha}{4}\left[e_{\mathrm{CH_{2}}}'\right]_{{-M_{1,4'}}}$ \\
        \hline
        $\phantom{\Big(}\left(\pr_{3,\alpha},\pr_{4,\alpha}'\right) \phantom{\Big(}$  & $\theta e_{\mathrm{CH}_{3,\alpha}}+ (1-\theta) u $ & $\frac{1+2\alpha}{4}\left[e_{\mathrm{CH_{2}}}\right]_{{-M_{1,3}}}$ \\
        \hline
        $\phantom{\Big(}\left(\pr_{4,\alpha}',\pr_{1,\alpha}\right) \phantom{\Big(}$ & $\phantom{\Big(}\theta' e_{\mathrm{CH}_{2}}+ (1-\theta') u \phantom{\Big(}$ & $\phantom{\Big(} \frac{1+2\alpha}{4}\left[e_{\mathrm{CH_{2}}}\right]_{{-M_{1,1}}}\phantom{\Big(} $ \\
\hline        
    \end{tabular}
    \caption{${\mathrm{ CHSH}_3}$ ($\npr=3$ Class 2)}
\end{table}

\begin{table}[H]
    \centering
    \begin{tabular}{|c||c|c|}
    \hline
        Pairs &  Effects & $\phantom{\Big(} \mathbf{C}^T.\mathbf{y} \phantom{\Big(} $  \\
        \hline \hline
       $\phantom{\Big(}\left(\pr_{1,\alpha},\pr_{1,\alpha}\right) \phantom{\Big(}$ & $\phantom{\Big(}\theta e_{\mathrm{CH}_{3,\alpha}}+ (1-\theta) u \phantom{\Big(}$ & $\phantom{\Big(} \frac{1+2\alpha}{4}\left[e_{\mathrm{CH_{2}}}\right]_{{-M_{1,3}}}\phantom{\Big(} $ \\
        \hline
        $\phantom{\Big(}\left(\pr_{1,\alpha},\pr_{3,\alpha}\right) \phantom{\Big(}$  & $\theta' e_{\mathrm{CH}_{2}}+ (1-\theta') u $ & $\frac{1+2\alpha}{4}\left[e_{\mathrm{CH_{2}}}\right]_{{-M_{1,1}}}$ \\
        \hline
        $\phantom{\Big(}\left(\pr_{1,\alpha},\pr_{4,\alpha}'\right) \phantom{\Big(}$  & $\theta e_{\mathrm{CH}_{1,\alpha}}+ (1-\theta) u $ & $\frac{1+2\alpha}{4}\left[e_{\mathrm{CH_{2}}}'\right]_{{-M_{3,4'}}}$ \\
        \hline
        $\phantom{\Big(}\left(\pr_{3,\alpha},\pr_{1,\alpha}\right) \phantom{\Big(}$ & $\phantom{\Big(}\theta' e_{\mathrm{CH}_{2}}'+ (1-\theta') u \phantom{\Big(}$ & $\phantom{\Big(} \frac{1+2\alpha}{4}\left[e_{\mathrm{CH_{2}}}'\right]_{{-M_{1,1}}}\phantom{\Big(} $ \\
        \hline
        $\phantom{\Big(}\left(\pr_{3,\alpha},\pr_{4,\alpha}'\right) \phantom{\Big(}$  & $\theta e_{\mathrm{CH}_{4,\alpha}}'+ (1-\theta) u $ & $\frac{1+2\alpha}{4}\left[e_{\mathrm{CH_{2}}}'\right]_{{-M_{1,4'}}}$ \\
        \hline
        $\phantom{\Big(}\left(\pr_{4,\alpha}',\pr_{1,\alpha}\right) \phantom{\Big(}$  & $\theta e_{\mathrm{CH}_{1,\alpha}}+ (1-\theta) u $ & $\frac{1+2\alpha}{4}\left[e_{\mathrm{CH_{2}}}'\right]_{{-M_{3,4'}}}$ \\
        \hline
        $\phantom{\Big(}\left(\pr_{4,\alpha}',\pr_{3,\alpha}\right) \phantom{\Big(}$ & $\phantom{\Big(}\theta e_{\mathrm{CH}_{4,\alpha}}'+ (1-\theta) u \phantom{\Big(}$ & $\phantom{\Big(} \frac{1+2\alpha}{4}\left[e_{\mathrm{CH_{2}}}'\right]_{{-M_{1,4'}}}\phantom{\Big(} $ \\
\hline
    \end{tabular}
    \caption{${\mathrm{ CHSH}_4'}$ ($\npr=3$ Class 2)}
\end{table}

\subsubsection{Class 3}

\begin{table}[H]
    \centering
    \begin{tabular}{|c||c|c|}
    \hline
        Pairs &  Effects & $\phantom{\Big(} \mathbf{C}^T.\mathbf{y} \phantom{\Big(} $  \\
        \hline \hline
       $\phantom{\Big(}\left(\pr_{1,\alpha},\pr_{2,\alpha}\right) \phantom{\Big(}$ & $\phantom{\Big(}\theta e_{\mathrm{CH}_{2,\alpha}}+ (1-\theta) u \phantom{\Big(}$ & $\phantom{\Big(} \frac{1+2\alpha}{4}\left[e_{\mathrm{CH_{2,\alpha}}}\right]_{{-M_{1,1}}}\phantom{\Big(} $ \\
        \hline
        $\phantom{\Big(}\left(\pr_{2,\alpha},\pr_{1,\alpha}\right) \phantom{\Big(}$  & $\theta e_{\mathrm{CH}_{2,\alpha}}+ (1-\theta) u $ & $\frac{1+2\alpha}{4}\left[e_{\mathrm{CH_{2,\alpha}}}\right]_{{-M_{1,1}}}$ \\
        \hline
    \end{tabular}
    \caption{${\mathrm{ CHSH}_1}$ ($\npr=3$ Class 3)}
\end{table}

\begin{table}[H]
    \centering
    \begin{tabular}{|c||c|c|}
    \hline
        Pairs &  Effects & $\phantom{\Big(} \mathbf{C}^T.\mathbf{y} \phantom{\Big(} $  \\
        \hline \hline
       $\phantom{\Big(}\left(\pr_{1,\alpha}',\pr_{2,\alpha}\right) \phantom{\Big(}$ & $\phantom{\Big(}\theta e_{\mathrm{CH}_{2,\alpha}}+ (1-\theta) u \phantom{\Big(}$ & $\phantom{\Big(} \frac{1+2\alpha}{4}\left[e_{\mathrm{CH_{2,\alpha}}}\right]_{{-M_{1,1}}}\phantom{\Big(} $ \\
        \hline
        $\phantom{\Big(}\left(\pr_{2,\alpha},\pr_{1,\alpha}'\right) \phantom{\Big(}$  & $\theta e_{\mathrm{CH}_{2,\alpha}}+ (1-\theta) u $ & $\frac{1+2\alpha}{4}\left[e_{\mathrm{CH_{2,\alpha}}}\right]_{{-M_{1,1}}}$ \\
        \hline
    \end{tabular}
    \caption{${\mathrm{ CHSH}_1'}$ ($\npr=3$ Class 3)}
\end{table}

\begin{table}[H]
    \centering
    \begin{tabular}{|c||c|c|}
    \hline
        Pairs &  Effects & $\phantom{\Big(} \mathbf{C}^T.\mathbf{y} \phantom{\Big(} $  \\
        \hline \hline
       $\phantom{\Big(}\left(\pr_{1,\alpha},\pr_{1,\alpha}'\right) \phantom{\Big(}$ & $\phantom{\Big(}\theta e_{\mathrm{CH}_{2,\alpha}}+ (1-\theta) u \phantom{\Big(}$ & $\phantom{\Big(} \frac{1+2\alpha}{4}\left[e_{\mathrm{CH_{2,\alpha}}}\right]_{{-M_{1,1}}}\phantom{\Big(} $ \\
        \hline
        $\phantom{\Big(}\left(\pr_{1,\alpha}',\pr_{1,\alpha}\right) \phantom{\Big(}$  & $\theta e_{\mathrm{CH}_{2,\alpha}}+ (1-\theta) u $ & $\frac{1+2\alpha}{4}\left[e_{\mathrm{CH_{2,\alpha}}}\right]_{{-M_{1,1}}}$ \\
        \hline
        $\phantom{\Big(}\left(\pr_{2,\alpha},\pr_{2,\alpha}\right) \phantom{\Big(}$  & $\theta e_{\mathrm{CH}_{2,\alpha}}+ (1-\theta) u $ & $\frac{1+2\alpha}{4}\left[e_{\mathrm{CH_{2,\alpha}}}\right]_{{-M_{1,1}}}$ \\
        \hline
    \end{tabular}
    \caption{${\mathrm{ CHSH}_2}$ ($\npr=3$ Class 3)}
\end{table}

\subsection{\texorpdfstring{$\npr=4$}{\npr=4}}

\subsubsection{Class 3}

\begin{table}[H]
    \centering
    \begin{tabular}{|c||c|c|}
    \hline
        Pairs &  Effects & $\phantom{\Big(} \mathbf{C}^T.\mathbf{y} \phantom{\Big(} $  \\
        \hline \hline
       $\phantom{\Big(}\left(\pr_{1,\alpha},\pr_{3,\alpha}\right) \phantom{\Big(}$  & $\theta e_{\mathrm{CH}_{4,\alpha}}'+ (1-\theta) u $ & $\frac{1+2\alpha}{4}\left[e_{\mathrm{CH_{3,\alpha}}}\right]_{{-M_{1,2}}}$ \\
        \hline
        $\phantom{\Big(}\left(\pr_{1,\alpha},\pr_{4,\alpha}'\right) \phantom{\Big(}$ & $\phantom{\Big(}\theta e_{\mathrm{CH}_{3,\alpha}}+ (1-\theta) u \phantom{\Big(}$ & $\phantom{\Big(} \frac{1+2\alpha}{4}\left[e_{\mathrm{CH_{4,\alpha}}}'\right]_{{-M_{1,2}}}\phantom{\Big(} $ \\
        \hline
        $\phantom{\Big(}\left(\pr_{3,\alpha},\pr_{1,\alpha}\right) \phantom{\Big(}$  & $\theta e_{\mathrm{CH}_{4,\alpha}}'+ (1-\theta) u $ & $\frac{1+2\alpha}{4}\left[e_{\mathrm{CH_{3,\alpha}}}\right]_{{-M_{1,2}}}$ \\
        \hline
       $\phantom{\Big(}\left(\pr_{3,\alpha},\pr_{4,\alpha}'\right) \phantom{\Big(}$  & $\theta' e_{\mathrm{CH}_{2}}'+ (1-\theta') u $ & $\frac{1}{2}\left[e_{\mathrm{CH_{2}}}'\right]_{{-M_{1,1}}}$ \\
        \hline
        $\phantom{\Big(}\left(\pr_{4,\alpha}',\pr_{1,\alpha}\right) \phantom{\Big(}$ & $\phantom{\Big(}\theta e_{\mathrm{CH}_{3,\alpha}}+ (1-\theta) u \phantom{\Big(}$ & $\phantom{\Big(} \frac{1+2\alpha}{4}\left[e_{\mathrm{CH_{4,\alpha}}}'\right]_{{-M_{1,2}}}\phantom{\Big(} $ \\
        \hline
        $\phantom{\Big(}\left(\pr_{4,\alpha}',\pr_{3,\alpha}\right) \phantom{\Big(}$  & $\theta' e_{\mathrm{CH}_{2}}+ (1-\theta') u $ & $\frac{1}{2}\left[e_{\mathrm{CH_{2}}}\right]_{{-M_{1,1}}}$ \\
        \hline
    \end{tabular}
    \caption{${\mathrm{ CHSH}_1}$ ($\npr=4$ Class 3)}
\end{table}
\begin{table}[H]
    \centering
    \begin{tabular}{|c||c|c|}
    \hline
        Pairs &  Effects & $\phantom{\Big(} \mathbf{C}^T.\mathbf{y} \phantom{\Big(} $  \\
        \hline \hline
       $\phantom{\Big(}\left(\pr_{1,\alpha}',\pr_{3,\alpha}\right) \phantom{\Big(}$  & $\theta e_{\mathrm{CH}_{4,\alpha}}'+ (1-\theta) u $ & $\frac{1+2\alpha}{4}\left[e_{\mathrm{CH_{3,\alpha}}}\right]_{{-M_{1,2}}}$ \\
        \hline
        $\phantom{\Big(}\left(\pr_{1,\alpha}',\pr_{4,\alpha}'\right) \phantom{\Big(}$ & $\phantom{\Big(}\theta e_{\mathrm{CH}_{3,\alpha}}+ (1-\theta) u \phantom{\Big(}$ & $\phantom{\Big(} \frac{1+2\alpha}{4}\left[e_{\mathrm{CH_{4,\alpha}}}'\right]_{{-M_{1,2}}}\phantom{\Big(} $ \\
        \hline
        $\phantom{\Big(}\left(\pr_{3,\alpha},\pr_{1,\alpha}'\right) \phantom{\Big(}$  & $\theta e_{\mathrm{CH}_{3,\alpha}}+ (1-\theta) u $ & $\frac{1+2\alpha}{4}\left[e_{\mathrm{CH_{3,\alpha}}}\right]_{{-M_{1,2}}}$ \\
        \hline
        $\phantom{\Big(}\left(\pr_{3,\alpha},\pr_{4,\alpha}'\right) \phantom{\Big(}$  & $\theta' e_{\mathrm{CH}_{2}}+ (1-\theta') u $ & $\frac{1}{2}\left[e_{\mathrm{CH_{2}}}\right]_{{-M_{1,1}}}$ \\
        \hline
        $\phantom{\Big(}\left(\pr_{4,\alpha}',\pr_{1,\alpha}'\right) \phantom{\Big(}$ & $\phantom{\Big(}\theta e_{\mathrm{CH}_{3,\alpha}}+ (1-\theta) u \phantom{\Big(}$ & $\phantom{\Big(} \frac{1+2\alpha}{4}\left[e_{\mathrm{CH_{4,\alpha}}}'\right]_{{-M_{1,2}}}\phantom{\Big(} $ \\
        \hline
       $\phantom{\Big(}\left(\pr_{4,\alpha}',\pr_{3,\alpha}\right) \phantom{\Big(}$  & $\theta' e_{\mathrm{CH}_{2,}}'+ (1-\theta') u $ & $\frac{1}{2}\left[e_{\mathrm{CH_{2}}}\right]_{{-M_{1,1}}}$ \\
        \hline
    \end{tabular}
    \caption{${\mathrm{ CHSH}_1'}$ ($\npr=4$ Class 3)}
\end{table}
\begin{table}[H]
    \centering
    \begin{tabular}{|c||c|c|}
    \hline
        Pairs &  Effects & $\phantom{\Big(} \mathbf{C}^T.\mathbf{y} \phantom{\Big(} $  \\
        \hline \hline
       $\phantom{\Big(}\left(\pr_{1,\alpha},\pr_{1,\alpha}\right) \phantom{\Big(}$ & $\phantom{\Big(}\theta e_{\mathrm{CH}_{4,\alpha}}'+ (1-\theta) u \phantom{\Big(}$ & $\phantom{\Big(} \frac{1+2\alpha}{4}\left[e_{\mathrm{CH_{4,\alpha}}}'\right]_{{-M_{1,2}}}\phantom{\Big(} $ \\
        \hline
       $\phantom{\Big(}\left(\pr_{1,\alpha},\pr_{4,\alpha}'\right) \phantom{\Big(}$  & $\theta' e_{\mathrm{CH}_{2}}'+ (1-\theta') u $ & $\frac{1}{2}\left[e_{\mathrm{CH_{2}}}'\right]_{{-M_{1,1}}}$ \\
        \hline
        $\phantom{\Big(}\left(\pr_{1,\alpha}',\pr_{1,\alpha}'\right) \phantom{\Big(}$  & $\theta e_{\mathrm{CH}_{4,\alpha}}'+ (1-\theta) u $ & $\frac{1+2\alpha}{4}\left[e_{\mathrm{CH_{3,\alpha}}}\right]_{{-M_{1,2}}}$ \\
        \hline
        $\phantom{\Big(}\left(\pr_{1,\alpha}',\pr_{4,\alpha}'\right) \phantom{\Big(}$  & $\theta' e_{\mathrm{CH}_{2}}+ (1-\theta') u $ & $\frac{1}{2}\left[e_{\mathrm{CH_{2}}}\right]_{{-M_{1,1}}}$ \\
        \hline
        $\phantom{\Big(}\left(\pr_{3,\alpha},\pr_{3,\alpha}\right) \phantom{\Big(}$ & $\phantom{\Big(}\theta e_{\mathrm{CH}_{4,\alpha}}'+ (1-\theta) u \phantom{\Big(}$ & $\phantom{\Big(} \frac{1+2\alpha}{4}\left[e_{\mathrm{CH_{3,\alpha}}}\right]_{{-M_{1,2}}}\phantom{\Big(} $ \\
        \hline
        $\phantom{\Big(}\left(\pr_{3,\alpha},\pr_{4,\alpha}'\right) \phantom{\Big(}$  & $\theta e_{\mathrm{CH}_{4,\alpha}}+ (1-\theta) u $ & $\frac{1+2\alpha}{4}\left[e_{\mathrm{CH_{4,\alpha}}}'\right]_{{-M_{1,2}}}$ \\
        \hline
        $\phantom{\Big(}\left(\pr_{4,\alpha}',\pr_{1,\alpha}\right) \phantom{\Big(}$  & $\theta' e_{\mathrm{CH}_{2}}+ (1-\theta') u $ & $\frac{1}{2}\left[e_{\mathrm{CH_{2}}}\right]_{{-M_{1,1}}}$ \\
\hline
    \end{tabular}
    \caption{${\mathrm{ CHSH}_3}$ ($\npr=4$ Class 3)}
\end{table}
\begin{table}[H]
    \centering
    \begin{tabular}{|c||c|c|}
    \hline
        Pairs &  Effects & $\phantom{\Big(} \mathbf{C}^T.\mathbf{y} \phantom{\Big(} $  \\
        \hline \hline
       $\phantom{\Big(}\left(\pr_{1,\alpha},\pr_{1,\alpha}\right) \phantom{\Big(}$ & $\phantom{\Big(}\theta e_{\mathrm{CH}_{4,\alpha}}'+ (1-\theta) u \phantom{\Big(}$ & $\phantom{\Big(} \frac{1+2\alpha}{4}\left[e_{\mathrm{CH_{4,\alpha}}}'\right]_{{-M_{1,2}}}\phantom{\Big(} $ \\
        \hline
       $\phantom{\Big(}\left(\pr_{1,\alpha},\pr_{3,\alpha}\right) \phantom{\Big(}$  & $\theta' e_{\mathrm{CH}_{2}}+ (1-\theta') u $ & $\frac{1}{2}\left[e_{\mathrm{CH_{2}}}\right]_{{-M_{1,1}}}$ \\
        \hline
        $\phantom{\Big(}\left(\pr_{1,\alpha}',\pr_{1,\alpha}'\right) \phantom{\Big(}$  & $\theta e_{\mathrm{CH}_{4,\alpha}}'+ (1-\theta) u $ & $\frac{1+2\alpha}{4}\left[e_{\mathrm{CH_{4,\alpha}}}'\right]_{{-M_{1,2}}}$ \\
        \hline
       $\phantom{\Big(}\left(\pr_{1,\alpha}',\pr_{3,\alpha}\right) \phantom{\Big(}$  & $\theta' e_{\mathrm{CH}_{2}}'+ (1-\theta') u $ & $\frac{1}{2}\left[e_{\mathrm{CH_{2}}}'\right]_{{-M_{1,1}}}$ \\
        \hline
        $\phantom{\Big(}\left(\pr_{3,\alpha},\pr_{1,\alpha}\right) \phantom{\Big(}$ & $\phantom{\Big(}\theta' e_{\mathrm{CH}_{2}}'+ (1-\theta') u \phantom{\Big(}$ & $\phantom{\Big(} \frac{1}{2}\left[e_{\mathrm{CH_{2}}}'\right]_{{-M_{1,1}}}\phantom{\Big(} $ \\
        \hline
        $\phantom{\Big(}\left(\pr_{3,\alpha},\pr_{1,\alpha}'\right) \phantom{\Big(}$  & $\theta' e_{\mathrm{CH}_{2}}+ (1-\theta') u $ & $\frac{1}{2}\left[e_{\mathrm{CH_{2}}}\right]_{{-M_{1,1}}}$ \\
        \hline
        $\phantom{\Big(}\left(\pr_{3,\alpha},\pr_{4,\alpha}'\right) \phantom{\Big(}$  & $\theta e_{\mathrm{CH}_{3,\alpha}}+ (1-\theta) u $ & $\frac{1+2\alpha}{4}\left[e_{\mathrm{CH_{3,\alpha}}}\right]_{{-M_{1,2}}}$ \\
\hline
    \end{tabular}
    \caption{${\mathrm{ CHSH}_4'}$ ($\npr=4$ Class 3)}
\end{table}

\subsubsection{Class 4}
\begin{table}[H]
    \centering
    \begin{tabular}{|c||c|c|}
    \hline
        Pairs &  Effects & $\phantom{\Big(} \mathbf{C}^T.\mathbf{y} \phantom{\Big(} $  \\
        \hline \hline
       $\phantom{\Big(}\left(\pr_{1,\alpha},\pr_{1,\alpha}\right) \phantom{\Big(}$ & $\phantom{\Big(}\theta e_{\mathrm{CH}_{1,\alpha}}+ (1-\theta) u \phantom{\Big(}$ & $\phantom{\Big(} \frac{1+2\alpha}{4}\left[e_{\mathrm{CH_{1,\alpha}}}\right]_{{-M_{2,2}}}\phantom{\Big(} $ \\
        \hline
        $\phantom{\Big(}\left(\pr_{1,\alpha},\pr_{2,\alpha}'\right) \phantom{\Big(}$  & $\theta e_{\mathrm{CH}_{2,\alpha}}'+ (1-\theta) u $ & $\frac{1+2\alpha}{4}\left[e_{\mathrm{CH_{2,\alpha}}}'\right]_{{-M_{1,1}}}$ \\
        \hline
        $\phantom{\Big(}\left(\pr_{1,\alpha},\pr_{3,\alpha}\right) \phantom{\Big(}$  & $\theta e_{\mathrm{CH}_{4,\alpha}}'+ (1-\theta) u $ & $\frac{1+2\alpha}{4}\left[e_{\mathrm{CH_{4,\alpha}}}'\right]_{{-M_{1,2}}}$ \\
        \hline
        $\phantom{\Big(}\left(\pr_{1,\alpha},\pr_{4,\alpha}'\right) \phantom{\Big(}$  & $\theta e_{\mathrm{CH}_{3,\alpha}}+ (1-\theta) u $ & $\frac{1+2\alpha}{4}\left[e_{\mathrm{CH_{3,\alpha}}}\right]_{{-M_{2,1}}}$ \\
        \hline
        $\phantom{\Big(}\left(\pr_{2,\alpha}',\pr_{1,\alpha}\right) \phantom{\Big(}$ & $\phantom{\Big(}\theta e_{\mathrm{CH}_{2,\alpha}}'+ (1-\theta) u \phantom{\Big(}$ & $\phantom{\Big(} \frac{1+2\alpha}{4}\left[e_{\mathrm{CH_{2,\alpha}}}'\right]_{{-M_{1,1}}}\phantom{\Big(} $ \\
        \hline
        $\phantom{\Big(}\left(\pr_{2,\alpha}',\pr_{4,\alpha}'\right) \phantom{\Big(}$  & $\theta e_{\mathrm{CH}_{4,\alpha}}'+ (1-\theta) u $ & $\frac{1+2\alpha}{4}\left[e_{\mathrm{CH_{4,\alpha}}}'\right]_{{-M_{1,2}}}$ \\
        \hline
        $\phantom{\Big(}\left(\pr_{3,\alpha},\pr_{1,\alpha}\right) \phantom{\Big(}$  & $\theta e_{\mathrm{CH}_{4,\alpha}}'+ (1-\theta) u $ & $\frac{1+2\alpha}{4}\left[e_{\mathrm{CH_{4,\alpha}}}'\right]_{{-M_{1,2}}}$ \\
\hline
    \end{tabular}
    \caption{${\mathrm{ CHSH}_1}$ ($\npr=4$ Class 4)}
\end{table}
\begin{table}[H]
    \centering
    \begin{tabular}{|c||c|c|}
    \hline
        Pairs &  Effects & $\phantom{\Big(} \mathbf{C}^T.\mathbf{y} \phantom{\Big(} $  \\
        \hline \hline
        $\phantom{\Big(}\left(\pr_{1,\alpha},\pr_{2,\alpha}'\right) \phantom{\Big(}$ & $\phantom{\Big(}\theta e_{\mathrm{CH}_{1,\alpha}}+ (1-\theta) u \phantom{\Big(}$ & $\phantom{\Big(} \frac{1+2\alpha}{4}\left[e_{\mathrm{CH_{1,\alpha}}}\right]_{{-M_{2,2}}}\phantom{\Big(} $ \\
        \hline
        $\phantom{\Big(}\left(\pr_{1,\alpha},\pr_{3,\alpha}\right) \phantom{\Big(}$  & $\theta e_{\mathrm{CH}_{3,\alpha}}+ (1-\theta) u $ & $\frac{1+2\alpha}{4}\left[e_{\mathrm{CH_{3,\alpha}}}\right]_{{-M_{2,1}}}$ \\
        \hline
        $\phantom{\Big(}\left(\pr_{2,\alpha}',\pr_{1,\alpha}\right) \phantom{\Big(}$  & $\theta e_{\mathrm{CH}_{1,\alpha}}+ (1-\theta) u $ & $\frac{1+2\alpha}{4}\left[e_{\mathrm{CH_{1,\alpha}}}\right]_{{-M_{2,2}}}$ \\
        \hline
        $\phantom{\Big(}\left(\pr_{2,\alpha}',\pr_{2,\alpha}'\right) \phantom{\Big(}$  & $\theta e_{\mathrm{CH}_{2,\alpha}}'+ (1-\theta) u $ & $\frac{1+2\alpha}{4}\left[e_{\mathrm{CH_{2,\alpha}}}'\right]_{{-M_{1,1}}}$ \\
        \hline
        $\phantom{\Big(}\left(\pr_{2,\alpha}',\pr_{3,\alpha}\right) \phantom{\Big(}$ & $\phantom{\Big(}\theta e_{\mathrm{CH}_{4,\alpha}}'+ (1-\theta) u \phantom{\Big(}$ & $\phantom{\Big(} \frac{1+2\alpha}{4}\left[e_{\mathrm{CH_{4,\alpha}}}'\right]_{{-M_{1,2}}}\phantom{\Big(} $ \\
        \hline
        $\phantom{\Big(}\left(\pr_{2,\alpha}',\pr_{4,\alpha}'\right) \phantom{\Big(}$  & $\theta e_{\mathrm{CH}_{3,\alpha}}+ (1-\theta) u $ & $\frac{1+2\alpha}{4}\left[e_{\mathrm{CH_{3,\alpha}}}\right]_{{-M_{2,1}}}$ \\
        \hline
        $\phantom{\Big(}\left(\pr_{3,\alpha},\pr_{2,\alpha}'\right) \phantom{\Big(}$  & $\theta e_{\mathrm{CH}_{4,\alpha}}'+ (1-\theta) u $ & $\frac{1+2\alpha}{4}\left[e_{\mathrm{CH_{4,\alpha}}}'\right]_{{-M_{1,2}}}$ \\
\hline
    \end{tabular}
    \caption{${\mathrm{ CHSH}_2'}$ ($\npr=4$ Class 4)}
\end{table}
\begin{table}[H]
    \centering
    \begin{tabular}{|c||c|c|}
    \hline
        Pairs &  Effects & $\phantom{\Big(} \mathbf{C}^T.\mathbf{y} \phantom{\Big(} $  \\
        \hline \hline
        $\phantom{\Big(}\left(\pr_{1,\alpha},\pr_{1,\alpha}\right) \phantom{\Big(}$ & $\phantom{\Big(}\theta e_{\mathrm{CH}_{4,\alpha}}'+ (1-\theta) u \phantom{\Big(}$ & $\phantom{\Big(} \frac{1+2\alpha}{4}\left[e_{\mathrm{CH_{4,\alpha}}}'\right]_{{-M_{1,2}}}\phantom{\Big(} $ \\
        \hline
        $\phantom{\Big(}\left(\pr_{1,\alpha},\pr_{2,\alpha}'\right) \phantom{\Big(}$  & $\theta e_{\mathrm{CH}_{3,\alpha}}+ (1-\theta) u $ & $\frac{1+2\alpha}{4}\left[e_{\mathrm{CH_{3,\alpha}}}\right]_{{-M_{2,1}}}$ \\
        \hline
        $\phantom{\Big(}\left(\pr_{1,\alpha},\pr_{3,\alpha}\right) \phantom{\Big(}$  & $\theta e_{\mathrm{CH}_{1,\alpha}}+ (1-\theta) u $ & $\frac{1+2\alpha}{4}\left[e_{\mathrm{CH_{1,\alpha}}}\right]_{{-M_{2,2}}}$ \\
        \hline
        $\phantom{\Big(}\left(\pr_{1,\alpha},\pr_{4,\alpha}'\right) \phantom{\Big(}$  & $\theta e_{\mathrm{CH}_{2,\alpha}}'+ (1-\theta) u $ & $\frac{1+2\alpha}{4}\left[e_{\mathrm{CH_{2,\alpha}}}'\right]_{{-M_{1,1}}}$ \\
        \hline
        $\phantom{\Big(}\left(\pr_{2,\alpha}',\pr_{2,\alpha}'\right) \phantom{\Big(}$ & $\phantom{\Big(}\theta e_{\mathrm{CH}_{4,\alpha}}'+ (1-\theta) u \phantom{\Big(}$ & $\phantom{\Big(} \frac{1+2\alpha}{4}\left[e_{\mathrm{CH_{4,\alpha}}}'\right]_{{-M_{1,2}}}\phantom{\Big(} $ \\
        \hline
        $\phantom{\Big(}\left(\pr_{2,\alpha}',\pr_{3,\alpha}\right) \phantom{\Big(}$  & $\theta e_{\mathrm{CH}_{2,\alpha}}'+ (1-\theta) u $ & $\frac{1+2\alpha}{4}\left[e_{\mathrm{CH_{2,\alpha}}}'\right]_{{-M_{1,1}}}$ \\
        \hline
        $\phantom{\Big(}\left(\pr_{3,\alpha},\pr_{1,\alpha}\right) \phantom{\Big(}$  & $\theta e_{\mathrm{CH}_{1,\alpha}}+ (1-\theta) u $ & $\frac{1+2\alpha}{4}\left[e_{\mathrm{CH_{1,\alpha}}}\right]_{{-M_{2,2}}}$ \\
\hline
    \end{tabular}
    \caption{${\mathrm{ CHSH}_3}$ ($\npr=4$ Class 4)}
\end{table}
\begin{table}[H]
    \centering
    \begin{tabular}{|c||c|c|}
    \hline
        Pairs &  Effects & $\phantom{\Big(} \mathbf{C}^T.\mathbf{y} \phantom{\Big(} $  \\
        \hline \hline
        $\phantom{\Big(}\left(\pr_{1,\alpha},\pr_{1,\alpha}\right) \phantom{\Big(}$ & $\phantom{\Big(}\theta e_{\mathrm{CH}_{3,\alpha}}+ (1-\theta) u \phantom{\Big(}$ & $\phantom{\Big(} \frac{1+2\alpha}{4}\left[e_{\mathrm{CH_{3,\alpha}}}\right]_{{-M_{2,1}}}\phantom{\Big(} $ \\
        \hline
        $\phantom{\Big(}\left(\pr_{1,\alpha},\pr_{4,\alpha}'\right) \phantom{\Big(}$  & $\theta e_{\mathrm{CH}_{1,\alpha}}+ (1-\theta) u $ & $\frac{1+2\alpha}{4}\left[e_{\mathrm{CH_{1,\alpha}}}\right]_{{-M_{2,2}}}$ \\
        \hline
        $\phantom{\Big(}\left(\pr_{2,\alpha}',\pr_{1,\alpha}\right) \phantom{\Big(}$  & $\theta e_{\mathrm{CH}_{4,\alpha}}'+ (1-\theta) u $ & $\frac{1+2\alpha}{4}\left[e_{\mathrm{CH_{4,\alpha}}}'\right]_{{-M_{1,2}}}$ \\
        \hline
        $\phantom{\Big(}\left(\pr_{2,\alpha}',\pr_{2,\alpha}'\right) \phantom{\Big(}$  & $\theta e_{\mathrm{CH}_{3,\alpha}}+ (1-\theta) u $ & $\frac{1+2\alpha}{4}\left[e_{\mathrm{CH_{3,\alpha}}}\right]_{{-M_{2,1}}}$ \\
        \hline
        $\phantom{\Big(}\left(\pr_{2,\alpha}',\pr_{3,\alpha}\right) \phantom{\Big(}$ & $\phantom{\Big(}\theta e_{\mathrm{CH}_{1,\alpha}}+ (1-\theta) u \phantom{\Big(}$ & $\phantom{\Big(} \frac{1+2\alpha}{4}\left[e_{\mathrm{CH_{1,\alpha}}}\right]_{{-M_{2,2}}}\phantom{\Big(} $ \\
        \hline
        $\phantom{\Big(}\left(\pr_{2,\alpha}',\pr_{4,\alpha}'\right) \phantom{\Big(}$  & $\theta e_{\mathrm{CH}_{2,\alpha}}'+ (1-\theta) u $ & $\frac{1+2\alpha}{4}\left[e_{\mathrm{CH_{2,\alpha}}}'\right]_{{-M_{1,1}}}$ \\
        \hline
        $\phantom{\Big(}\left(\pr_{3,\alpha},\pr_{1,\alpha}\right) \phantom{\Big(}$  & $\theta e_{\mathrm{CH}_{2,\alpha}}'+ (1-\theta) u $ & $\frac{1+2\alpha}{4}\left[e_{\mathrm{CH_{2,\alpha}}}'\right]_{{-M_{1,1}}}$ \\
\hline
    \end{tabular}
    \caption{${\mathrm{ CHSH}_4'}$ ($\npr=4$ Class 4)}
\end{table}

\subsubsection{Class 5}

\begin{table}[H]
    \centering
    \begin{tabular}{|c||c|c|}
    \hline
        Pairs &  Effects & $\phantom{\Big(} \mathbf{C}^T.\mathbf{y} \phantom{\Big(} $  \\
        \hline \hline
        $\phantom{\Big(}\left(\pr_{1,\alpha},\pr_{1,\alpha}\right) \phantom{\Big(}$ & $\phantom{\Big(}\theta e_{\mathrm{CH}_{1,\alpha}}+ (1-\theta) u \phantom{\Big(}$ & $\phantom{\Big(} \frac{1+2\alpha}{4}\left[e_{\mathrm{CH_{1,\alpha}}}\right]_{{-M_{2,2}}}\phantom{\Big(} $ \\
        \hline
        $\phantom{\Big(}\left(\pr_{1,\alpha},\pr_{2,\alpha}\right) \phantom{\Big(}$  & $\theta e_{\mathrm{CH}_{2,\alpha}}+ (1-\theta) u $ & $\frac{1+2\alpha}{4}\left[e_{\mathrm{CH_{2,\alpha}}}\right]_{{-M_{1,1}}}$ \\
        \hline
        $\phantom{\Big(}\left(\pr_{1,\alpha},\pr_{3,\alpha}\right) \phantom{\Big(}$  & $\theta e_{\mathrm{CH}_{4,\alpha}}'+ (1-\theta) u $ & $\frac{1+2\alpha}{4}\left[e_{\mathrm{CH_{4,\alpha}}}'\right]_{{-M_{2,1}}}$ \\
        \hline
        $\phantom{\Big(}\left(\pr_{1,\alpha},\pr_{4,\alpha}'\right) \phantom{\Big(}$  & $\theta e_{\mathrm{CH}_{3,\alpha}}+ (1-\theta) u $ & $\frac{1+2\alpha}{4}\left[e_{\mathrm{CH_{3,\alpha}}}\right]_{{-M_{1,2}}}$ \\
        \hline
        $\phantom{\Big(}\left(\pr_{2,\alpha},\pr_{1,\alpha}\right) \phantom{\Big(}$ & $\phantom{\Big(}\theta e_{\mathrm{CH}_{2,\alpha}}+ (1-\theta) u \phantom{\Big(}$ & $\phantom{\Big(} \frac{1+2\alpha}{4}\left[e_{\mathrm{CH_{2,\alpha}}}\right]_{{-M_{1,1}}}\phantom{\Big(} $ \\
        \hline
        $\phantom{\Big(}\left(\pr_{2,\alpha},\pr_{3,\alpha}\right) \phantom{\Big(}$  & $\theta e_{\mathrm{CH}_{3,\alpha}}+ (1-\theta) u $ & $\frac{1+2\alpha}{4}\left[e_{\mathrm{CH_{3,\alpha}}}\right]_{{-M_{1,2}}}$ \\
        \hline
        $\phantom{\Big(}\left(\pr_{3,\alpha},\pr_{1,\alpha}\right) \phantom{\Big(}$  & $\theta e_{\mathrm{CH}_{4,\alpha}}'+ (1-\theta) u $ & $\frac{1+2\alpha}{4}\left[e_{\mathrm{CH_{4,\alpha}}}'\right]_{{-M_{2,1}}}$ \\
\hline
    \end{tabular}
    \caption{${\mathrm{ CHSH}_1}$ ($\npr=4$ Class 5)}
\end{table}
\begin{table}[H]
    \centering
    \begin{tabular}{|c||c|c|}
    \hline
        Pairs &  Effects & $\phantom{\Big(} \mathbf{C}^T.\mathbf{y} \phantom{\Big(} $  \\
        \hline \hline
       $\phantom{\Big(}\left(\pr_{1,\alpha},\pr_{2,\alpha}\right) \phantom{\Big(}$ & $\phantom{\Big(}\theta e_{\mathrm{CH}_{1,\alpha}}+ (1-\theta) u \phantom{\Big(}$ & $\phantom{\Big(} \frac{1+2\alpha}{4}\left[e_{\mathrm{CH_{1,\alpha}}}\right]_{{-M_{2,2}}}\phantom{\Big(} $ \\
        \hline
        $\phantom{\Big(}\left(\pr_{1,\alpha},\pr_{4,\alpha}'\right) \phantom{\Big(}$  & $\theta e_{\mathrm{CH}_{4,\alpha}}'+ (1-\theta) u $ & $\frac{1+2\alpha}{4}\left[e_{\mathrm{CH_{4,\alpha}}}'\right]_{{-M_{2,1}}}$ \\
        \hline
        $\phantom{\Big(}\left(\pr_{2,\alpha},\pr_{1,\alpha}\right) \phantom{\Big(}$  & $\theta e_{\mathrm{CH}_{1,\alpha}}+ (1-\theta) u $ & $\frac{1+2\alpha}{4}\left[e_{\mathrm{CH_{1,\alpha}}}\right]_{{-M_{2,2}}}$ \\
        \hline
        $\phantom{\Big(}\left(\pr_{2,\alpha},\pr_{2,\alpha}\right) \phantom{\Big(}$  & $\theta e_{\mathrm{CH}_{2,\alpha}}+ (1-\theta) u $ & $\frac{1+2\alpha}{4}\left[e_{\mathrm{CH_{2,\alpha}}}\right]_{{-M_{1,1}}}$ \\
        \hline
        $\phantom{\Big(}\left(\pr_{2,\alpha},\pr_{3,\alpha}\right) \phantom{\Big(}$ & $\phantom{\Big(}\theta e_{\mathrm{CH}_{4,\alpha}}'+ (1-\theta) u \phantom{\Big(}$ & $\phantom{\Big(} \frac{1+2\alpha}{4}\left[e_{\mathrm{CH_{4,\alpha}}}'\right]_{{-M_{2,1}}}\phantom{\Big(} $ \\
        \hline
        $\phantom{\Big(}\left(\pr_{2,\alpha},\pr_{4,\alpha}'\right) \phantom{\Big(}$  & $\theta e_{\mathrm{CH}_{3,\alpha}}+ (1-\theta) u $ & $\frac{1+2\alpha}{4}\left[e_{\mathrm{CH_{3,\alpha}}}\right]_{{-M_{1,2}}}$ \\
        \hline
        $\phantom{\Big(}\left(\pr_{3,\alpha},\pr_{1,\alpha}\right) \phantom{\Big(}$  & $\theta e_{\mathrm{CH}_{3,\alpha}}+ (1-\theta) u $ & $\frac{1+2\alpha}{4}\left[e_{\mathrm{CH_{3,\alpha}}}\right]_{{-M_{1,2}}}$ \\
\hline
    \end{tabular}
    \caption{${\mathrm{ CHSH}_2}$ ($\npr=4$ Class 5)}
\end{table}
\begin{table}[H]
    \centering
    \begin{tabular}{|c||c|c|}
    \hline
        Pairs &  Effects & $\phantom{\Big(} \mathbf{C}^T.\mathbf{y} \phantom{\Big(} $  \\
        \hline \hline
        $\phantom{\Big(}\left(\pr_{1,\alpha},\pr_{1,\alpha}\right) \phantom{\Big(}$ & $\phantom{\Big(}\theta e_{\mathrm{CH}_{4,\alpha}}'+ (1-\theta) u \phantom{\Big(}$ & $\phantom{\Big(} \frac{1+2\alpha}{4}\left[e_{\mathrm{CH_{4,\alpha}}}'\right]_{{-M_{2,1}}}\phantom{\Big(} $ \\
        \hline
        $\phantom{\Big(}\left(\pr_{1,\alpha},\pr_{3,\alpha}\right) \phantom{\Big(}$  & $\theta e_{\mathrm{CH}_{1,\alpha}}+ (1-\theta) u $ & $\frac{1+2\alpha}{4}\left[e_{\mathrm{CH_{1,\alpha}}}\right]_{{-M_{2,2}}}$ \\
        \hline
        $\phantom{\Big(}\left(\pr_{2,\alpha},\pr_{1,\alpha}\right) \phantom{\Big(}$  & $\theta e_{\mathrm{CH}_{3,\alpha}}+ (1-\theta) u $ & $\frac{1+2\alpha}{4}\left[e_{\mathrm{CH_{3,\alpha}}}\right]_{{-M_{1,2}}}$ \\
        \hline
        $\phantom{\Big(}\left(\pr_{2,\alpha},\pr_{2,\alpha}\right) \phantom{\Big(}$  & $\theta e_{\mathrm{CH}_{4,\alpha}}'+ (1-\theta) u $ & $\frac{1+2\alpha}{4}\left[e_{\mathrm{CH_{4,\alpha}}}'\right]_{{-M_{2,1}}}$ \\
        \hline
        $\phantom{\Big(}\left(\pr_{2,\alpha},\pr_{3,\alpha}\right) \phantom{\Big(}$ & $\phantom{\Big(}\theta e_{\mathrm{CH}_{7,\alpha}}+ (1-\theta) u \phantom{\Big(}$ & $\phantom{\Big(} \frac{1+2\alpha}{4}\left[e_{\mathrm{CH_{2,\alpha}}}\right]_{{-M_{1,1}}}\phantom{\Big(} $ \\
        \hline
        $\phantom{\Big(}\left(\pr_{2,\alpha},\pr_{4,\alpha}'\right) \phantom{\Big(}$  & $\theta e_{\mathrm{CH}_{1,\alpha}}+ (1-\theta) u $ & $\frac{1+2\alpha}{4}\left[e_{\mathrm{CH_{1,\alpha}}}\right]_{{-M_{2,2}}}$ \\
        \hline
        $\phantom{\Big(}\left(\pr_{3,\alpha},\pr_{1,\alpha}\right) \phantom{\Big(}$  & $\theta e_{\mathrm{CH}_{1,\alpha}}+ (1-\theta) u $ & $\frac{1+2\alpha}{4}\left[e_{\mathrm{CH_{1,\alpha}}}\right]_{{-M_{2,2}}}$ \\
\hline
    \end{tabular}
    \caption{${\mathrm{ CHSH}_3}$ ($\npr=4$ Class 5)}
\end{table}
\begin{table}[H]
    \centering
    \begin{tabular}{|c||c|c|}
    \hline
        Pairs &  Effects & $\phantom{\Big(} \mathbf{C}^T.\mathbf{y} \phantom{\Big(} $  \\
        \hline \hline
        $\phantom{\Big(}\left(\pr_{1,\alpha},\pr_{1,\alpha}\right) \phantom{\Big(}$ & $\phantom{\Big(}\theta e_{\mathrm{CH}_{3,\alpha}}+ (1-\theta) u \phantom{\Big(}$ & $\phantom{\Big(} \frac{1+2\alpha}{4}\left[e_{\mathrm{CH_{3,\alpha}}}\right]_{{-M_{1,2}}}\phantom{\Big(} $ \\
        \hline
        $\phantom{\Big(}\left(\pr_{1,\alpha},\pr_{2,\alpha}\right) \phantom{\Big(}$  & $\theta e_{\mathrm{CH}_{4,\alpha}}'+ (1-\theta) u $ & $\frac{1+2\alpha}{4}\left[e_{\mathrm{CH_{4,\alpha}}}'\right]_{{-M_{2,1}}}$ \\
        \hline
        $\phantom{\Big(}\left(\pr_{1,\alpha},\pr_{3,\alpha}\right) \phantom{\Big(}$  & $\theta e_{\mathrm{CH}_{2,\alpha}}+ (1-\theta) u $ & $\frac{1+2\alpha}{4}\left[e_{\mathrm{CH_{2,\alpha}}}\right]_{{-M_{1,1}}}$ \\
        \hline
        $\phantom{\Big(}\left(\pr_{1,\alpha},\pr_{4,\alpha}'\right) \phantom{\Big(}$  & $\theta e_{\mathrm{CH}_{1,\alpha}}+ (1-\theta) u $ & $\frac{1+2\alpha}{4}\left[e_{\mathrm{CH_{1,\alpha}}}\right]_{{-M_{2,2}}}$ \\
        \hline
        $\phantom{\Big(}\left(\pr_{2,\alpha},\pr_{2,\alpha}\right) \phantom{\Big(}$ & $\phantom{\Big(}\theta e_{\mathrm{CH}_{3,\alpha}}+ (1-\theta) u \phantom{\Big(}$ & $\phantom{\Big(} \frac{1+2\alpha}{4}\left[e_{\mathrm{CH_{3,\alpha}}}\right]_{{-M_{1,2}}}\phantom{\Big(} $ \\
        \hline
        $\phantom{\Big(}\left(\pr_{2,\alpha},\pr_{4,\alpha}'\right) \phantom{\Big(}$  & $\theta e_{\mathrm{CH}_{2,\alpha}}+ (1-\theta) u $ & $\frac{1+2\alpha}{4}\left[e_{\mathrm{CH_{7,\alpha}}}\right]_{{-M_{1,1}}}$ \\
        \hline
        $\phantom{\Big(}\left(\pr_{3,\alpha},\pr_{2,\alpha}\right) \phantom{\Big(}$  & $\theta e_{\mathrm{CH}_{1,\alpha}}+ (1-\theta) u $ & $\frac{1+2\alpha}{4}\left[e_{\mathrm{CH_{1,\alpha}}}\right]_{{-M_{2,2}}}$ \\
\hline
    \end{tabular}
    \caption{${\mathrm{ CHSH}_4'}$ ($\npr=4$ Class 5)}
\end{table}

\subsection{\texorpdfstring{$\npr=5$}{\npr=5}}

\subsubsection{Class 2}

\begin{table}[H]
    \centering
    \begin{tabular}{|c||c|c|}
    \hline
        Pairs &  Effects & $\phantom{\Big(} \mathbf{C}^T.\mathbf{y} \phantom{\Big(} $  \\
        \hline \hline
       $\phantom{\Big(}\left(\pr_{1,\alpha},\pr_{1,\alpha}\right) \phantom{\Big(}$ & $\phantom{\Big(}\theta e_{\mathrm{CH}_{1,\alpha}}+ (1-\theta) u \phantom{\Big(}$ & $\phantom{\Big(} \frac{1+2\alpha}{4}\left[e_{\mathrm{CH_{1,\alpha}}}\right]_{{-M_{2,2}}}\phantom{\Big(} $ \\
        \hline
        $\phantom{\Big(}\left(\pr_{1,\alpha},\pr_{3,\alpha}\right) \phantom{\Big(}$  & $\theta e_{\mathrm{CH}_{4,\alpha}}'+ (1-\theta) u $ & $\frac{1+2\alpha}{4}\left[e_{\mathrm{CH_{4,\alpha}}}'\right]_{{-M_{1,3}}}$ \\
        \hline
         $\phantom{\Big(}\left(\pr_{1,\alpha},\pr_{4,\alpha}'\right) \phantom{\Big(}$ & $\phantom{\Big(}\theta e_{\mathrm{CH}_{3,\alpha}}+ (1-\theta) u \phantom{\Big(}$ & $\phantom{\Big(} \frac{1+2\alpha}{4}\left[e_{\mathrm{CH_{3,\alpha}}}\right]_{{-M_{1,2}}}\phantom{\Big(} $ \\
        \hline
        $\phantom{\Big(}\left(\pr_{2,\alpha},\pr_{2,\alpha}'\right) \phantom{\Big(}$  & $\theta e_{\mathrm{CH}_{1,\alpha}}+ (1-\theta) u $ & $\frac{1+2\alpha}{4}\left[e_{\mathrm{CH_{1,\alpha}}}\right]_{{-M_{2,2}}}$ \\
        \hline
        $\phantom{\Big(}\left(\pr_{2,\alpha},\pr_{3,\alpha}\right) \phantom{\Big(}$  & $\theta e_{\mathrm{CH}_{3,\alpha}}+ (1-\theta) u $ & $\frac{1+2\alpha}{4}\left[e_{\mathrm{CH_{3,\alpha}}}\right]_{{-M_{1,2}}}$ \\
        \hline
        $\phantom{\Big(}\left(\pr_{2,\alpha}',\pr_{2,\alpha}\right) \phantom{\Big(}$  & $\theta e_{\mathrm{CH}_{1,\alpha}}+ (1-\theta) u $ & $\frac{1+2\alpha}{4}\left[e_{\mathrm{CH_{1,\alpha}}}\right]_{{-M_{2,2}}}$ \\
        \hline
        $\phantom{\Big(}\left(\pr_{2,\alpha}',\pr_{4,\alpha}'\right) \phantom{\Big(}$ & $\phantom{\Big(}\theta e_{\mathrm{CH}_{4,\alpha}}'+ (1-\theta) u \phantom{\Big(}$ & $\phantom{\Big(} \frac{1+2\alpha}{4}\left[e_{\mathrm{CH_{4,\alpha}}}'\right]_{{-M_{1,3}}}\phantom{\Big(} $ \\
\hline
    \end{tabular}
    \caption{${\mathrm{ CHSH}_1}$ ($\npr=5$ Class 2)}
\end{table}
\begin{table}[H]
    \centering
    \begin{tabular}{|c||c|c|}
    \hline
        Pairs &  Effects & $\phantom{\Big(} \mathbf{C}^T.\mathbf{y} \phantom{\Big(} $  \\
        \hline \hline        
        $\phantom{\Big(}\left(\pr_{1,\alpha},\pr_{2,\alpha}\right) \phantom{\Big(}$  & $\theta e_{\mathrm{CH}_{1,\alpha}}+ (1-\theta) u $ & $\frac{1+2\alpha}{4}\left[e_{\mathrm{CH_{1,\alpha}}}\right]_{{-M_{2,2}}}$ \\
        \hline
        $\phantom{\Big(}\left(\pr_{1,\alpha},\pr_{4,\alpha}'\right) \phantom{\Big(}$ & $\phantom{\Big(}\theta e_{\mathrm{CH}_{4,\alpha}}'+ (1-\theta) u \phantom{\Big(}$ & $\phantom{\Big(} \frac{1+2\alpha}{4}\left[e_{\mathrm{CH_{4,\alpha}}}'\right]_{{-M_{1,3}}}\phantom{\Big(} $ \\
        \hline
        $\phantom{\Big(}\left(\pr_{2,\alpha},\pr_{1,\alpha}\right) \phantom{\Big(}$  & $\theta e_{\mathrm{CH}_{1,\alpha}}+ (1-\theta) u $ & $\frac{1+2\alpha}{4}\left[e_{\mathrm{CH_{1,\alpha}}}\right]_{{-M_{2,2}}}$ \\
        \hline
        $\phantom{\Big(}\left(\pr_{2,\alpha},\pr_{3,\alpha}\right) \phantom{\Big(}$  & $\theta e_{\mathrm{CH}_{4,\alpha}}'+ (1-\theta) u $ & $\frac{1+2\alpha}{4}\left[e_{\mathrm{CH_{4,\alpha}}}'\right]_{{-M_{1,3}}}$ \\
        \hline
        $\phantom{\Big(}\left(\pr_{2,\alpha},\pr_{4,\alpha}'\right) \phantom{\Big(}$ & $\phantom{\Big(}\theta e_{\mathrm{CH}_{3,\alpha}}+ (1-\theta) u \phantom{\Big(}$ & $\phantom{\Big(} \frac{1+2\alpha}{4}\left[e_{\mathrm{CH_{3,\alpha}}}\right]_{{-M_{1,2}}}\phantom{\Big(} $ \\
        \hline
        $\phantom{\Big(}\left(\pr_{3,\alpha},\pr_{1,\alpha}\right) \phantom{\Big(}$  & $\theta e_{\mathrm{CH}_{3,\alpha}}+ (1-\theta) u $ & $\frac{1+2\alpha}{4}\left[e_{\mathrm{CH_{3,\alpha}}}\right]_{{-M_{1,2}}}$ \\
        \hline        
        $\phantom{\Big(}\left(\pr_{3,\alpha},\pr_{2,\alpha}\right) \phantom{\Big(}$  & $\theta e_{\mathrm{CH}_{4,\alpha}}'+ (1-\theta) u $ & $\frac{1+2\alpha}{4}\left[e_{\mathrm{CH_{4,\alpha}}}'\right]_{{-M_{1,3}}}$ \\
\hline
    \end{tabular}
    \caption{${\mathrm{ CHSH}_2}$ ($\npr=5$ Class 2)}
\end{table}
\begin{table}[H]
    \centering
    \begin{tabular}{|c||c|c|}
    \hline
        Pairs &  Effects & $\phantom{\Big(} \mathbf{C}^T.\mathbf{y} \phantom{\Big(} $  \\
        \hline \hline
      $\phantom{\Big(}\left(\pr_{1,\alpha},\pr_{2,\alpha}'\right) \phantom{\Big(}$  & $\theta e_{\mathrm{CH}_{1,\alpha}}+ (1-\theta) u $ & $\frac{1+2\alpha}{4}\left[e_{\mathrm{CH_{1,\alpha}}}\right]_{{-M_{2,2}}}$ \\
        \hline
        $\phantom{\Big(}\left(\pr_{1,\alpha},\pr_{3,\alpha}\right) \phantom{\Big(}$ & $\phantom{\Big(}\theta e_{\mathrm{CH}_{3,\alpha}}+ (1-\theta) u \phantom{\Big(}$ & $\phantom{\Big(} \frac{1+2\alpha}{4}\left[e_{\mathrm{CH_{3,\alpha}}}\right]_{{-M_{1,2}}}\phantom{\Big(} $ \\
        \hline
        $\phantom{\Big(}\left(\pr_{2,\alpha}',\pr_{1,\alpha}\right) \phantom{\Big(}$  & $\theta e_{\mathrm{CH}_{1,\alpha}}+ (1-\theta) u $ & $\frac{1+2\alpha}{4}\left[e_{\mathrm{CH_{1,\alpha}}}\right]_{{-M_{2,2}}}$ \\
        \hline
        $\phantom{\Big(}\left(\pr_{2,\alpha}',\pr_{3,\alpha}\right) \phantom{\Big(}$  & $\theta e_{\mathrm{CH}_{4,\alpha}}'+ (1-\theta) u $ & $\frac{1+2\alpha}{4}\left[e_{\mathrm{CH_{4,\alpha}}}'\right]_{{-M_{1,3}}}$ \\
        \hline
        $\phantom{\Big(}\left(\pr_{2,\alpha}',\pr_{4,\alpha}'\right) \phantom{\Big(}$ & $\phantom{\Big(}\theta e_{\mathrm{CH}_{3,\alpha}}+ (1-\theta) u \phantom{\Big(}$ & $\phantom{\Big(} \frac{1+2\alpha}{4}\left[e_{\mathrm{CH_{3,\alpha}}}\right]_{{-M_{1,2}}}\phantom{\Big(} $ \\
        \hline
        $\phantom{\Big(}\left(\pr_{3,\alpha},\pr_{2,\alpha}'\right) \phantom{\Big(}$  & $\theta e_{\mathrm{CH}_{4,\alpha}}'+ (1-\theta) u $ & $\frac{1+2\alpha}{4}\left[e_{\mathrm{CH_{4,\alpha}}}'\right]_{{-M_{1,3}}}$ \\
        \hline        
        $\phantom{\Big(}\left(\pr_{4,\alpha}',\pr_{1,\alpha}\right) \phantom{\Big(}$  & $\theta e_{\mathrm{CH}_{4,\alpha}}'+ (1-\theta) u $ & $\frac{1+2\alpha}{4}\left[e_{\mathrm{CH_{4,\alpha}}}'\right]_{{-M_{1,3}}}$ \\
\hline
    \end{tabular}
    \caption{${\mathrm{ CHSH}_2'}$ ($\npr=5$ Class 2)}
\end{table}
\begin{table}[H]
    \centering
    \begin{tabular}{|c||c|c|}
    \hline
        Pairs &  Effects & $\phantom{\Big(} \mathbf{C}^T.\mathbf{y} \phantom{\Big(} $  \\
        \hline \hline
        $\phantom{\Big(}\left(\pr_{1,\alpha},\pr_{1,\alpha}\right) \phantom{\Big(}$  & $\theta e_{\mathrm{CH}_{4,\alpha}}'+ (1-\theta) u $ & $\frac{1+2\alpha}{4}\left[e_{\mathrm{CH_{4,\alpha}}}'\right]_{{-M_{1,3}}}$ \\
        \hline
        $\phantom{\Big(}\left(\pr_{1,\alpha},\pr_{2,\alpha}'\right) \phantom{\Big(}$ & $\phantom{\Big(}\theta e_{\mathrm{CH}_{3,\alpha}}+ (1-\theta) u \phantom{\Big(}$ & $\phantom{\Big(} \frac{1+2\alpha}{4}\left[e_{\mathrm{CH_{3,\alpha}}}\right]_{{-M_{1,2}}}\phantom{\Big(} $ \\
        \hline
        $\phantom{\Big(}\left(\pr_{1,\alpha},\pr_{3,\alpha}\right) \phantom{\Big(}$  & $\theta e_{\mathrm{CH}_{1,\alpha}}+ (1-\theta) u $ & $\frac{1+2\alpha}{4}\left[e_{\mathrm{CH_{1,\alpha}}}\right]_{{-M_{2,2}}}$ \\
        \hline
        $\phantom{\Big(}\left(\pr_{2,\alpha},\pr_{1,\alpha}\right) \phantom{\Big(}$  & $\theta e_{\mathrm{CH}_{3,\alpha}}+ (1-\theta) u $ & $\frac{1+2\alpha}{4}\left[e_{\mathrm{CH_{3,\alpha}}}\right]_{{-M_{1,2}}}$ \\
        \hline
        $\phantom{\Big(}\left(\pr_{2,\alpha},\pr_{2,\alpha}\right) \phantom{\Big(}$ & $\phantom{\Big(}\theta e_{\mathrm{CH}_{4,\alpha}}'+ (1-\theta) u \phantom{\Big(}$ & $\phantom{\Big(} \frac{1+2\alpha}{4}\left[e_{\mathrm{CH_{4,\alpha}}}'\right]_{{-M_{1,3}}}\phantom{\Big(} $ \\
        \hline
        $\phantom{\Big(}\left(\pr_{2,\alpha},\pr_{4,\alpha}'\right) \phantom{\Big(}$  & $\theta e_{\mathrm{CH}_{1,\alpha}}+ (1-\theta) u $ & $\frac{1+2\alpha}{4}\left[e_{\mathrm{CH_{1,\alpha}}}\right]_{{-M_{2,2}}}$ \\
        \hline        
        $\phantom{\Big(}\left(\pr_{2,\alpha}',\pr_{2,\alpha}'\right) \phantom{\Big(}$  & $\theta e_{\mathrm{CH}_{4,\alpha}}'+ (1-\theta) u $ & $\frac{1+2\alpha}{4}\left[e_{\mathrm{CH_{4,\alpha}}}'\right]_{{-M_{1,3}}}$ \\
\hline
    \end{tabular}
    \caption{${\mathrm{ CHSH}_3}$ ($\npr=5$ Class 2)}
\end{table}
\begin{table}[H]
    \centering
    \begin{tabular}{|c||c|c|}
    \hline
        Pairs &  Effects & $\phantom{\Big(} \mathbf{C}^T.\mathbf{y} \phantom{\Big(} $  \\
        \hline \hline
        $\phantom{\Big(}\left(\pr_{1,\alpha},\pr_{1,\alpha}\right) \phantom{\Big(}$  & $\theta e_{\mathrm{CH}_{3,\alpha}}+ (1-\theta) u $ & $\frac{1+2\alpha}{4}\left[e_{\mathrm{CH_{3,\alpha}}}\right]_{{-M_{1,2}}}$ \\
        \hline
        $\phantom{\Big(}\left(\pr_{1,\alpha},\pr_{2,\alpha}\right) \phantom{\Big(}$ & $\phantom{\Big(}\theta e_{\mathrm{CH}_{4,\alpha}}'+ (1-\theta) u \phantom{\Big(}$ & $\phantom{\Big(} \frac{1+2\alpha}{4}\left[e_{\mathrm{CH_{4,\alpha}}}'\right]_{{-M_{1,3}}}\phantom{\Big(} $ \\
        \hline
        $\phantom{\Big(}\left(\pr_{1,\alpha},\pr_{4,\alpha}'\right) \phantom{\Big(}$  & $\theta e_{\mathrm{CH}_{1,\alpha}}+ (1-\theta) u $ & $\frac{1+2\alpha}{4}\left[e_{\mathrm{CH_{1,\alpha}}}\right]_{{-M_{2,2}}}$ \\
        \hline
        $\phantom{\Big(}\left(\pr_{2,\alpha},\pr_{2,\alpha}\right) \phantom{\Big(}$  & $\theta e_{\mathrm{CH}_{3,\alpha}}+ (1-\theta) u $ & $\frac{1+2\alpha}{4}\left[e_{\mathrm{CH_{3,\alpha}}}\right]_{{-M_{1,2}}}$ \\
        \hline
        $\phantom{\Big(}\left(\pr_{2,\alpha}',\pr_{1,\alpha}\right) \phantom{\Big(}$ & $\phantom{\Big(}\theta e_{\mathrm{CH}_{4,\alpha}}'+ (1-\theta) u \phantom{\Big(}$ & $\phantom{\Big(} \frac{1+2\alpha}{4}\left[e_{\mathrm{CH_{4,\alpha}}}'\right]_{{-M_{1,3}}}\phantom{\Big(} $ \\
        \hline
        $\phantom{\Big(}\left(\pr_{2,\alpha}',\pr_{2,\alpha}'\right) \phantom{\Big(}$  & $\theta e_{\mathrm{CH}_{3,\alpha}}+ (1-\theta) u $ & $\frac{1+2\alpha}{4}\left[e_{\mathrm{CH_{3,\alpha}}}\right]_{{-M_{1,2}}}$ \\
        \hline        
        $\phantom{\Big(}\left(\pr_{2,\alpha}',\pr_{3,\alpha}\right) \phantom{\Big(}$  & $\theta e_{\mathrm{CH}_{1,\alpha}}+ (1-\theta) u $ & $\frac{1+2\alpha}{4}\left[e_{\mathrm{CH_{1,\alpha}}}\right]_{{-M_{2,2}}}$ \\
\hline
    \end{tabular}
    \caption{${\mathrm{ CHSH}_4'}$ ($\npr=5$ Class 2)}
\end{table}

\subsubsection{Class 3}

\begin{table}[H]
    \centering
    \begin{tabular}{|c||c|c|}
    \hline
        Pairs &  Effects & $\phantom{\Big(} \mathbf{C}^T.\mathbf{y} \phantom{\Big(} $  \\
        \hline \hline
       $\phantom{\Big(}\left(\pr_{1,\alpha},\pr_{1,\alpha}\right) \phantom{\Big(}$ & $\phantom{\Big(}\theta e_{\mathrm{CH}_{1,\alpha}}+ (1-\theta) u \phantom{\Big(}$ & $\phantom{\Big(} \frac{1+2\alpha}{4}\left[e_{\mathrm{CH_{1,\alpha}}}\right]_{{-M_{2,3}}}\phantom{\Big(} $ \\
        \hline
        $\phantom{\Big(}\left(\pr_{3,\alpha},\pr_{3,\alpha}\right) \phantom{\Big(}$  & $\theta e_{\mathrm{CH}_{1,\alpha}}+ (1-\theta) u $ & $\frac{1+2\alpha}{4}\left[e_{\mathrm{CH_{1,\alpha}}}\right]_{{-M_{2,3}}}$ \\
        \hline
        $\phantom{\Big(}\left(\pr_{3,\alpha},\pr_{4,\alpha}'\right) \phantom{\Big(}$  & $\theta' e_{\mathrm{CH}_{2}}'+ (1-\theta') u $ & $\frac{1}{2}\left[e_{\mathrm{CH_{2}}}'\right]_{{-M_{1,1}}}$ \\
        \hline
        $\phantom{\Big(}\left(\pr_{3,\alpha}',\pr_{3,\alpha}'\right) \phantom{\Big(}$  & $\theta e_{\mathrm{CH}_{1,\alpha}}+ (1-\theta) u $ & $\frac{1+2\alpha}{4}\left[e_{\mathrm{CH_{1,\alpha}}}\right]_{{-M_{2,3}}}$ \\
        \hline
        $\phantom{\Big(}\left(\pr_{3,\alpha}',\pr_{4,\alpha}\right) \phantom{\Big(}$ & $\phantom{\Big(}\theta' e_{\mathrm{CH}_{2}}'+ (1-\theta') u \phantom{\Big(}$ & $\phantom{\Big(} \frac{1}{2}\left[e_{\mathrm{CH_{2}}}'\right]_{{-M_{1,1}}}\phantom{\Big(} $ \\
        \hline
        $\phantom{\Big(}\left(\pr_{3,\alpha}',\pr_{4,\alpha}'\right) \phantom{\Big(}$  & $\theta' e_{\mathrm{CH}_{2}}+ (1-\theta') u $ & $\frac{1}{2}\left[e_{\mathrm{CH_{2}}}\right]_{{-M_{1,1}}}$ \\
        \hline
        $\phantom{\Big(}\left(\pr_{4,\alpha},\pr_{3,\alpha}\right) \phantom{\Big(}$  & $\theta' e_{\mathrm{CH}_{2}}'+ (1-\theta') u $ & $\frac{1}{2}\left[e_{\mathrm{CH_{2}}}'\right]_{{-M_{1,1}}}$ \\
\hline
    \end{tabular}
    \caption{${\mathrm{ CHSH}_1}$ ($\npr=5$ Class 3)}
\end{table}
\begin{table}[H]
    \centering
    \begin{tabular}{|c||c|c|}
    \hline
        Pairs &  Effects & $\phantom{\Big(} \mathbf{C}^T.\mathbf{y} \phantom{\Big(} $  \\
        \hline \hline
       $\phantom{\Big(}\left(\pr_{1,\alpha},\pr_{3,\alpha}\right) \phantom{\Big(}$  & $\theta e_{\mathrm{CH}_{1,\alpha}}+ (1-\theta) u $ & $\frac{1+2\alpha}{4}\left[e_{\mathrm{CH_{1,\alpha}}}\right]_{{-M_{2,3}}}$ \\
        \hline
        $\phantom{\Big(}\left(\pr_{1,\alpha},\pr_{4,\alpha}\right) \phantom{\Big(}$ & $\phantom{\Big(}\theta' e_{\mathrm{CH}_{2}}+ (1-\theta') u \phantom{\Big(}$ & $\phantom{\Big(} \frac{1}{2}\left[e_{\mathrm{CH_{2}}}\right]_{{-M_{1,1}}}\phantom{\Big(} $ \\
        \hline
        $\phantom{\Big(}\left(\pr_{1,\alpha},\pr_{4,\alpha}'\right) \phantom{\Big(}$  & $\theta' e_{\mathrm{CH}_{2}}'+ (1-\theta') u $ & $\frac{1}{2}\left[e_{\mathrm{CH_{2}}}'\right]_{{-M_{1,1}}}$ \\
        \hline
        $\phantom{\Big(}\left(\pr_{3,\alpha},\pr_{1,\alpha}\right) \phantom{\Big(}$  & $\theta e_{\mathrm{CH}_{1,\alpha}}+ (1-\theta) u $ & $\frac{1+2\alpha}{4}\left[e_{\mathrm{CH_{1,\alpha}}}\right]_{{-M_{2,3}}}$ \\
        \hline
        $\phantom{\Big(}\left(\pr_{4,\alpha},\pr_{1,\alpha}\right) \phantom{\Big(}$ & $\phantom{\Big(}\theta' e_{\mathrm{CH}_{2}}'+ (1-\theta') u \phantom{\Big(}$ & $\phantom{\Big(} \frac{1}{2}\left[e_{\mathrm{CH_{2}}}'\right]_{{-M_{1,1}}}\phantom{\Big(} $ \\
        \hline
        $\phantom{\Big(}\left(\pr_{4,\alpha}',\pr_{1,\alpha}\right) \phantom{\Big(}$  & $\theta' e_{\mathrm{CH}_{2}}+ (1-\theta') u $ & $\frac{1}{2}\left[e_{\mathrm{CH_{2}}}\right]_{{-M_{1,1}}}$ \\
        \hline
    \end{tabular}
    \caption{${\mathrm{ CHSH}_3}$ ($\npr=5$ Class 3)}
\end{table}
\begin{table}[H]
    \centering
    \begin{tabular}{|c||c|c|}
    \hline
        Pairs &  Effects & $\phantom{\Big(} \mathbf{C}^T.\mathbf{y} \phantom{\Big(} $  \\
        \hline \hline
       $\phantom{\Big(}\left(\pr_{1,\alpha},\pr_{3,\alpha}'\right) \phantom{\Big(}$  & $\theta e_{\mathrm{CH}_{1,\alpha}}+ (1-\theta) u $ & $\frac{1+2\alpha}{4}\left[e_{\mathrm{CH_{1,\alpha}}}\right]_{{-M_{2,3}}}$ \\
        \hline
        $\phantom{\Big(}\left(\pr_{1,\alpha},\pr_{4,\alpha}\right) \phantom{\Big(}$ & $\phantom{\Big(}\theta' e_{\mathrm{CH}_{2}}'+ (1-\theta') u \phantom{\Big(}$ & $\phantom{\Big(} \frac{1}{2}\left[e_{\mathrm{CH_{2}}}'\right]_{{-M_{1,1}}}\phantom{\Big(} $ \\
        \hline
        $\phantom{\Big(}\left(\pr_{1,\alpha},\pr_{4,\alpha}'\right) \phantom{\Big(}$  & $\theta' e_{\mathrm{CH}_{2}}+ (1-\theta') u $ & $\frac{1}{2}\left[e_{\mathrm{CH_{2}}}\right]_{{-M_{1,1}}}$ \\
        \hline
        $\phantom{\Big(}\left(\pr_{3,\alpha}',\pr_{1,\alpha}\right) \phantom{\Big(}$  & $\theta e_{\mathrm{CH}_{1,\alpha}}+ (1-\theta) u $ & $\frac{1+2\alpha}{4}\left[e_{\mathrm{CH_{1,\alpha}}}\right]_{{-M_{2,3}}}$ \\
        \hline
        $\phantom{\Big(}\left(\pr_{4,\alpha},\pr_{1,\alpha}\right) \phantom{\Big(}$ & $\phantom{\Big(}\theta' e_{\mathrm{CH}_{2}}+ (1-\theta') u \phantom{\Big(}$ & $\phantom{\Big(} \frac{1}{2}\left[e_{\mathrm{CH_{2}}}\right]_{{-M_{1,1}}}\phantom{\Big(} $ \\
        \hline
        $\phantom{\Big(}\left(\pr_{4,\alpha}',\pr_{1,\alpha}\right) \phantom{\Big(}$  & $\theta' e_{\mathrm{CH}_{2}}'+ (1-\theta') u $ & $\frac{1}{2}\left[e_{\mathrm{CH_{2}}}'\right]_{{-M_{1,1}}}$ \\
        \hline 
    \end{tabular}
    \caption{${\mathrm{ CHSH}_3'}$ ($\npr=5$ Class 3)}
\end{table}
\begin{table}[H]
    \centering
    \begin{tabular}{|c||c|c|}
    \hline
        Pairs &  Effects & $\phantom{\Big(} \mathbf{C}^T.\mathbf{y} \phantom{\Big(} $  \\
        \hline \hline
       $\phantom{\Big(}\left(\pr_{1,\alpha},\pr_{3,\alpha}\right) \phantom{\Big(}$  & $\theta' e_{\mathrm{CH}_{2}}'+ (1-\theta') u $ & $\frac{1}{2}\left[e_{\mathrm{CH_{2}}}'\right]_{{-M_{1,1}}}$ \\
        \hline
        $\phantom{\Big(}\left(\pr_{1,\alpha},\pr_{3,\alpha}'\right) \phantom{\Big(}$ & $\phantom{\Big(}\theta' e_{\mathrm{CH}_{2}}'+ (1-\theta') u \phantom{\Big(}$ & $\phantom{\Big(} \frac{1}{2}\left[e_{\mathrm{CH_{2}}}'\right]_{{-M_{1,1}}}\phantom{\Big(} $ \\
        \hline
        $\phantom{\Big(}\left(\pr_{1,\alpha},\pr_{4,\alpha}\right) \phantom{\Big(}$  & $\theta e_{\mathrm{CH}_{1,\alpha}}+ (1-\theta) u $ & $\frac{1+2\alpha}{4}\left[e_{\mathrm{CH_{1,\alpha}}}\right]_{{-M_{2,3}}}$ \\
        \hline
        $\phantom{\Big(}\left(\pr_{3,\alpha},\pr_{1,\alpha}\right) \phantom{\Big(}$  & $\theta' e_{\mathrm{CH}_{2}}+ (1-\theta') u $ & $\frac{1}{2}\left[e_{\mathrm{CH_{2}}}\right]_{{-M_{1,1}}}$ \\
        \hline
        $\phantom{\Big(}\left(\pr_{3,\alpha}',\pr_{1,\alpha}\right) \phantom{\Big(}$ & $\phantom{\Big(}\theta' e_{\mathrm{CH}_{2}}'+ (1-\theta') u \phantom{\Big(}$ & $\phantom{\Big(} \frac{1}{2}\left[e_{\mathrm{CH_{2}}}'\right]_{{-M_{1,1}}}\phantom{\Big(} $ \\
        \hline
        $\phantom{\Big(}\left(\pr_{4,\alpha},\pr_{1,\alpha}\right) \phantom{\Big(}$  & $\theta e_{\mathrm{CH}_{1,\alpha}}+ (1-\theta) u $ & $\frac{1+2\alpha}{4}\left[e_{\mathrm{CH_{1,\alpha}}}\right]_{{-M_{2,3}}}$ \\
        \hline
    \end{tabular}%
    \caption{${\mathrm{ CHSH}_4'}$ ($\npr=5$ Class 3)}
\end{table}

\subsection{\texorpdfstring{$\npr=6$}{\npr=6}}

\subsubsection{Class 3}

\begin{table}[H]
    \centering
    %
    \caption{${\mathrm{ CHSH}_4'}$ ($\npr=6$ Class 3)}
\end{table}

\subsection{\texorpdfstring{$\npr=7$}{\npr=7}}

The effect for each of these cases is $\theta e_{\mathrm{CH}_{1,\alpha}}+ (1-\theta) u $ and every $\mathbf{C}^T.\mathbf{y}$ is $\frac{1+2\alpha}{4}\left[e_{\mathrm{CH_{2,\alpha}}}\right]_{{-M_{1,1}}}$.

\begin{table}[h]
    \centering
    %
    \caption{Table describes pairs of noisy PR boxes relevant for the respective CHSH games.}
    \label{Table::describes pairs of noisy PR boxes relevant for the respective CHSH games.}
\end{table}

\section{Proof of Lemma~\ref{Lemma::IntersectionEmpty}}
\label{Appendix::LemmaIntersection}

 For a conditional probability distribution $\p(A,B|X,Y)$ to have a score less than 3/4 in a CHSH game is given by ${\rm CHSH_i}\left[\p(A,B|X,Y)\right] \leq 3/4$, where $\rm CHSH_i$ is one of the following:

$$
{\rm CHSH}_1=\frac{1}{4}\left(
\begin{array}{cc|cc}
 1 & 0 & 1 & 0 \\
 0 & 1 & 0 & 1 \\
 \hline
 1 & 0 & 0 & 1 \\
 0 & 1 & 1 & 0 \\
\end{array}
\right), {\rm CHSH}_2=\frac{1}{4}\left(
\begin{array}{cc|cc}
 0 & 1 & 1 & 0 \\
 1 & 0 & 0 & 1 \\
 \hline
 1 & 0 & 1 & 0 \\
 0 & 1 & 0 & 1 \\
\end{array}
\right), {\rm CHSH}_3=\frac{1}{4}\left(
\begin{array}{cc|cc}
 1 & 0 & 0 & 1 \\
 0 & 1 & 1 & 0 \\
 \hline
 1 & 0 & 1 & 0 \\
 0 & 1 & 0 & 1 \\
\end{array}
\right), {\rm CHSH}_4=\frac{1}{4}\left(
\begin{array}{cc|cc}
 0 & 1 & 0 & 1 \\
 1 & 0 & 1 & 0 \\
 \hline
 1 & 0 & 0 & 1 \\
 0 & 1 & 1 & 0 \\
\end{array}
\right), 
$$

$$
{\rm CHSH}_5=\frac{1}{4}\left(
\begin{array}{cc|cc}
 0 & 1 & 0 & 1 \\
 1 & 0 & 1 & 0 \\
 \hline
 0 & 1 & 1 & 0 \\
 1 & 0 & 0 & 1 \\
\end{array}
\right), {\rm CHSH}_6=\frac{1}{4}\left(
\begin{array}{cc|cc}
 1 & 0 & 0 & 1 \\
 0 & 1 & 1 & 0 \\
 \hline
 0 & 1 & 0 & 1 \\
 1 & 0 & 1 & 0 \\
\end{array}%
\right);
$$

There are 8 CHSH inequalities that can lead to 2 classes of pairs up to relabelling symmetry: $\{\chsh_1,\chsh_3\}$ and $\{\chsh_1,\chsh_5\}$. Therefore it suffices to prove the results for these pairs. Let us assume that ${\rm CHSH_1}[\rm{p(A,B|X,Y)}]> 3/4$. Note that the $2 \times 2$ blocks of $\rm{CHSH_{j \neq 1}}$ are all different when $\rm j=5$ whereas others are different from $\rm{CHSH_1}$ by two blocks. Let us collect them in two sets $\rm{CHSH_{5}}$ and $\rm{CHSH_{j \neq 1,5}}$. First, let us assume that ${\rm CHSH_3}[\p(A,B|X,Y)] > 3/4$ as well. Then it follows that
\begin{equation}
    \begin{split}
      (\rm CHSH_1 + CHSH_3)  [\p(A,B|X,Y)] & > \frac{3}{2}\\
        \implies  \frac{1}{4}\left(\left(\begin{array}{cc|cc}
   1  & 0 & 1 & 0 \\
   0  & 1 & 0 & 1 \\
   \hline
   1  & 0 & 0 & 1 \\
   0  & 1 & 1 & 0 \\ 
\end{array}\right) + \left(\begin{array}{cc|cc}
   1  & 0 & 0 & 1 \\
   0  & 1 & 1 & 0 \\
   \hline
   1  & 0 & 1 & 0 \\
   0  & 1 & 0 & 1 \\ 
\end{array}\right) \right) [\rm{p(A,B|X,Y)}] & > \frac{3}{2}\\
\implies 1 + 1 + 2\left(\rm p(0,0|0,0)+ \rm p(1,1|0,0)+ \rm p(0,0|1,0)+ \rm p(1,1|1,0)\right) &> 6\\
\implies \left(\rm p(0,0|0,0)+ \rm p(1,1|0,0)+ \rm p(0,0|1,0)+ \rm p(1,1|1,0)\right) &> 2\\
    \end{split}
\end{equation}
On the other hand, for any conditional probability distribution $ \rm{p(A,B|X,Y)}$, we have 
\begin{equation}
    \begin{split}
        \max_{\rm{p} \in \mathcal{P} } \left(\rm p(0,0|0,0)+ \rm p(1,1|0,0) + \rm p(0,0|1,0)+ \rm p(1,1|1,0) \right) &\leq 2,\\       
    \end{split}
\end{equation}
which is a contradiction. A similar contradiction can be reached if any other element from the second set was chosen. Next, let us assume that ${\rm CHSH_1}[\rm{p(A,B|X,Y)}]> 3/4$ and ${\rm CHSH_5}[\rm{p(A,B|X,Y)}]> 3/4$, then

\begin{equation}
    \begin{split}
      ( {\rm CHSH_1 + CHSH_5})[\rm{p(A,B|X,Y)}] &> \frac{3}{2}\\
        \implies   \left(\left(\begin{array}{cc|cc}
   1  & 0 & 1 & 0 \\
   0  & 1 & 0 & 1 \\
   \hline
   1  & 0 & 0 & 1 \\
   0  & 1 & 1 & 0 \\ 
\end{array}\right) + \left(\begin{array}{cc|cc}
   0  & 1 & 0 & 1 \\
   1  & 0 & 1 & 0 \\
   \hline
   0  & 1 & 1 & 0 \\
   1  & 0 & 0 & 1 \\ 
\end{array}\right)\right) [\rm{p(A,B|X,Y)}] &>6\\
\implies \left(\begin{array}{cc|cc}
   1  & 1 & 1 & 1 \\
   1  & 1 & 1 & 1 \\
   \hline
   1  & 1 & 1 & 1 \\
   1  & 1 & 1 & 1 \\ 
\end{array}\right) [\rm{p(A,B|X,Y)}] &> 6\\
\implies 4 &> 6 
    \end{split}
\end{equation}
which is also a contradiction.

\section{Proof of Theorem~\ref{Proposition::WinningBound}}
\label{Appendix::WinningBound}

    In the adaptive CHSH game, Bob performs a four outcome measurement $M=\{e_{b}\}_{b \in{00,01,10,11}}$. Corresponding to each outcome Alice and Charlie need to win 4 different CHSH games labelled $\{\chsh_b \}_b$. To perform entanglement swapping, Bob shares two instances of the maximally entangled state, $\rm N_\alpha$, one with Alice and one with Charlie. Since $e_{00}$ is minimally $2$-preserving,  $s_{e_{00}} \in \mathrm{ConvHull}\{\rm{L}_1, \rm{L}_2, \ldots , \mathrm{L}_n, {\rm N}_\alpha \}$. Let,    
    \begin{equation}      
    s_{e_{00}}=\sum_{j=1}^n \lambda_{00,j} L_j + \delta{\rm N}_\alpha,
    \label{eq:cvx_post_state}
    \end{equation}
    \noindent
    such that $\sum_{j=1}^n \lambda_{00,j}+\delta=1$ and $\lambda_{00,j},\delta\geq 0$ for all $j \in \{1,2, \ldots , n\}$. Recall further that Alice and Charlie fix their measurements and do not change them throughout the run of the game. Since $s_{e_{00}}$ admits the decomposition in Eq.~\eqref{eq:cvx_post_state}, the probability distribution $ \mathrm{p}_{s_{e_{00}}}\left(\rm A,B|X,Z\right)$ obtained after Alice and Charlie measure the state $s_{e_{00}}$  can be expressed as
    \begin{equation}
        \mathrm{p}_{s_{e_{00}}}\left(\rm A,B|X,Z\right) = \sum_{j=1}^n \lambda_{00,j} \mathrm{p}_{L_j}\left(\rm A,B|X,Z\right) + \delta \mathrm{p}_{{\rm N}_\alpha}\left(\rm A,B|X,Z\right),
        \label{eq:cvx_post_dist}
    \end{equation}
    where $\mathrm{p}_{L_j}\left(\rm A,B|X,Z\right)$ is the distribution obtained if Alice and Charlie had measured the local deterministic state $L_j$ and $\mathrm{p}_{{\rm N}_\alpha}\left(\rm A,B|X,Z\right)$  is the distribution obtained if Alice and Charlie had measured the entangled state ${\rm N}_\alpha$. Their objective is that $\mathrm{p}_{s_{e_{00}}}\left(\rm A,B|X,Z\right)$ wins the CHSH game $\chsh_{00}$. For this, the state $s_{e_{00}}$ must be entangled, i.e., $s_{e_{00}} \notin \mathrm{ConvHull} \left\{\mathrm{ L_1, L_2, \ldots, L_n} \right\}$. In other words $e_{00}$ must be a coupler. Assume that ${\rm CHSH_{00}}[\p_{e_{00}}\left(\rm A,B|X,Z\right)] > 3/4$.  Next, consider the state $s_{e_{01}}$ left with Alice and Charlie corresponding to the outcome of the effect $e_{01}$. From minimal $2$-preservability of $e_{01}$, $s_{e_{01}}$ will also have a decomposition as in Eq.~\eqref{eq:cvx_post_state} and since the measurements of Alice and Charlie are fixed, $s_{e_{01}}$ will generate a conditional probability distribution, $\p_{s_{e_{01}}}(A,B|X,Y)$, that admits a decomposition similar to Eq.~\eqref{eq:cvx_post_dist}. If $e_{01}$ is a coupler, the distribution $\p_{s_{e_{01}}}(A,B|X,Y)$ will win the game $\chsh_{00}$ by an amount more than 3/4. However, this time Alice and Charlie need to win $\chsh_{01}$ by an amount more than 3/4.  But, by Lemma~\ref{Lemma::IntersectionEmpty}, no 2-input 2-output conditional distribution can simultaneously win two CHSH games by an amount more than 3/4. This implies that Alice and Charlie can only win $\chsh_{01}$ by a score of at most 3/4. This argument can be extended to the remaining two post-selected states as well.  The maximum winning probability is achieved if the measurement choice helps Alice and Charlie to have a score of 3/4 for the remaining games. Replacing $e_{00}$ by $e$ gives the following upper bound on the winning probability:
    \begin{equation}
        \begin{split}
            p_{\rm win} &\leq p_{\mathrm{succ}}(e) \zeta_e + (1- p_{\mathrm{succ}}(e)) \frac{3}{4}\\
                        &\leq \frac{3}{4} +p_{\mathrm{succ}}(e) \left(\zeta_{e} - \frac{3}{4} \right) \\
        \end{split}
    \end{equation}
    \noindent
    This upper bound is maximised when the product $ p_{\mathrm{succ}} (e) \left(\zeta_{e} - 3/4 \right)$ is maximised. When there are no minimally $2$-preserving couplers, $\zeta_{e}$  can be at most $3/4$. Putting these together, we get
    \begin{equation} 
           p_{\mathrm{win}} \leq \begin{cases}
              \phantom{\Big(} \frac{3}{4} &\text{if} \quad E_{\rm coup} = \emptyset\\
              \phantom{\Big(} \frac{3}{4} + \max\limits_{e \in E_{\rm coup}} p_{\mathrm{succ}}(e)\left(\zeta_e -\frac{3}{4}\right) &\text{if} \quad E_{\rm coup} \neq \emptyset
           \end{cases}.
       \end{equation}

\end{document}